\documentclass[11pt]{article} 
\usepackage[T1]{fontenc}
\usepackage[utf8]{inputenc}

\usepackage{microtype}

\usepackage{subcaption}
\usepackage{graphicx}
\usepackage[backend=biber,sorting=none,sortcites,citestyle=numeric-comp,bibstyle=phys,doi=false,eprint=true,url=false,biblabel=brackets,isbn=true,articletitle=true]{biblatex}   
\usepackage{authblk}

\usepackage{amsmath, amstext}    
\usepackage{cases}
\usepackage{mathtools}
\usepackage{dsfont}
\usepackage[dvipsnames]{xcolor}
\usepackage[margin=1in]{geometry}
\usepackage{float}
\usepackage{ragged2e}
\DeclareCaptionJustification{justified}{\justifying}
\usepackage{placeins}
\usepackage[title,titletoc]{appendix}
\usepackage[colorlinks=true, linkcolor = BlueViolet, citecolor = BlueViolet, urlcolor = BlueViolet]{hyperref}
\usepackage{enumitem}

\usepackage[bitstream-charter]{mathdesign}

\usepackage{subdepth} 

\usepackage{amsthm}   

\allowdisplaybreaks   

\newtheorem{theorem}{Theorem}[section]

\newtheorem{lemma}[theorem]{Lemma}
\newtheorem{proposition}[theorem]{Proposition}
\newtheorem{corollary}[theorem]{Corollary}
\newtheorem{statement}[theorem]{Statement}

\DeclareMathOperator{\Tr}{Tr}
\DeclareMathOperator{\sinc}{sinc}

\DeclareMathOperator{\vectorspan}{span}

\DeclareMathOperator{\ifc}{\overline{\Theta}}

\DeclareMathOperator{\swap}{SWAP}

\newcommand{\uh}{\hat{U}_H}

\newcommand{\proj}{\hat{\Pi}}

\newcommand{\idop}{\hat{\mathds{1}}}
\newcommand{\diff}{\mathrm{d}}

\newcommand{\texcept}{W_{\text{ex}}}

\newcommand{\therm}{\mathrm{th}}

\newcommand{\obs}{\text{obs}}

\newcommand{\unitarygroup}{U}

\newcommand{\aux}{\text{aux}}

\title{Provable quantum thermalization without statistical averages}
\author[1]{Amit Vikram}
\affil[1]{JILA and Center for Theory of Quantum Matter, Department of Physics, University of Colorado, Boulder CO 80309 USA}
\date{}

\begin{document}
\pagenumbering{gobble}
\maketitle

\abstract{
We develop a rigorous system-agnostic method to predict quantum thermalization in an overwhelming fraction of accessible pure states in a many-body system, entirely in terms of certain out-of-time-ordered correlators of few-body observables. In contrast to previous rigorous results on thermalization with semiclassical counterparts, our method is not limited to statistical averages of observables, such as time averages in ergodicity or state averages in mixing. Moreover, consistent with such approaches, we retain the advantage of not requiring a detailed knowledge of energy eigenstate structure or thermodynamically large times, which can become intractable for systems with more than a handful of particles.  Our approach is centered on a geometric result that connects thermalization to the alignment of high dimensional subspaces in a Hilbert space, which is determined by the saturation of ``controllably nonlocal'' out-of-time-ordered correlators. We discuss the possibility of predicting quantum thermalization in almost all states at almost all times from finite-time observations of a single such correlator per observable. This formalism reduces the problem of establishing pure state quantum thermalization at finite times in almost all complex many-body states to a theoretically or experimentally accessible study of few-body correlators, even in thermodynamically large systems.
}

\newpage

\setcounter{tocdepth}{2}
\tableofcontents

\newpage

\pagenumbering{arabic}

\section{Introduction: Provability in quantum statistical mechanics}
\label{sec:intro}

\subsection{Quantum statistical mechanics without energy eigenstates}
\label{sec:computationalmotivation}

Quantum statistical mechanics in isolated systems has long been attributed to the properties of energy eigenstates~\cite{vonNeumannThermalization, JensenShankarETH, deutsch1991eth, srednicki1994eth}, even since its first mathematical treatment by von Neumann~\cite{vonNeumannThermalization}. While these energy eigenstates provide a useful anchor to express results on thermalization, very few theoretical or experimental methods can access these states in systems with any more than the most rudimentary levels of complexity in quantum dynamics. For example, a practical quantum simulation of a sufficiently large system, interpreted as an experimental technique to calculate its dynamics, has virtually no hope of faithfully recreating the energy spectrum of the system due to various approximations inherent in the simulation and experimental limitations such as noise and decoherence~\cite{QuantumSimulationReview2024}. With the development of a variety of eigenstate-insensitive theoretical techniques to calculate the dynamics of few-body observables in various many-body systems~\cite{MaldacenaStanford, NahumOTOCCircuits, ChanScrambling, ProsenErgodic, ClaeysErgodicCircuits, ClaeysLamacraftMaxVelCkt, BertiniPiroliScrambling, RQCreview_2023, XuSwingleScramblingTutorial}, as well as the availability of eigenstate-insensitive few-body control and measurement techniques in experimental platforms such as quantum simulators~\cite{ReyOTOCMeas2017, GoogleScrambling, Joshi2020, OTOC2022ExptFiniteTemp, QuantumSimulationReview2024, GoogleScrambling2}, one could make the case that an eigenstate-independent approach to quantum thermalization is desirable to place the general problem of predicting observable thermalization phenomena within reach of these techniques.

Beyond questions of practical convenience, a more fundamental conceptual limitation exists in eigenstate-based approaches to thermalization such as the eigenstate thermalization hypothesis~\cite{vonNeumannThermalization, JensenShankarETH, deutsch1991eth, srednicki1994eth, srednicki1999eth, rigol2008eth, DAlessio2016, deutsch2018eth} or the equilibration of energy-delocalized initial states~\cite{Tasaki1998, ReimannRealistic, ShortEqb, ShortDegenerate, TasakiTypicalityThermalization}. This is the fact that these approaches only carry rigorous implications for thermalization over \textit{infinite} (thermodynamically large) time scales in many-body systems, with no rigorous discernible consequences for thermalization over finite times (that are e.g. independent of system size). In particular, the timescale $T$ over which eigenstate-based thermalization may be established scales \textit{exponentially} with the size of the system, $T \gtrsim \exp(N)$ in an $N$-particle system~\cite{DAlessio2016, ShortDegenerate}. Finite-time thermalization within these frameworks has largely been addressed in a heuristic manner, essentially by \textit{assuming} that observables generically have sufficiently rapid dynamics that their observable short-time values reach such theoretical long-time averages within a practical timescale~\cite{DAlessio2016}.

From both perspectives, a foundational question at hand is whether it is possible to \textit{predict} quantum thermalization within observable timescales, given only a finite amount of information about a large system. One limitation comes from demonstrations that there exist specific initial states in specific quantum systems for which the thermalization of certain observables over almost all times is undecidable by a classical (or quantum) computer, modeled as a Turing machine~\cite{ThermalizationUndecidability1, ThermalizationUndecidability2}. These examples specifically prohibit a \textit{general} computational mode of deciding long-time thermalization in \textit{all} states.
At the other extreme, a full quantum simulation of the dynamics of any system in a given initial state can trivially decide thermalization in the state at a given instant of time, just by recreating the dynamics under question. However, sampling even an appreciable fraction of all possible initial states of interest is intractable in a large many-body system.
Consequently, both classical and quantum computational approaches to thermalization appear to be severely limited in their direct predictive power to specific simulations of a vanishingly small fraction of nonequilibrium states over finite times. 

Recently, these limitations have been lifted to a significant extent by a rigorous approach to predicting quantum thermalization over different combinations of \textit{almost all} initial states (in \textit{every} orthonormal basis) and \textit{almost all} times, based on finite-time autocorrelators of few-body observables~\cite{dynamicalqthermalization}. Such autocorrelators only require information from the finite-time dynamics of a \textit{single} (mixed) initial state of a quantum system per projective observable of interest\footnote{Noting that all Hermitian few-body observables can be written as a finite sum of few-body projectors, these results have direct implications for the more general class of not necessarily projective Hermitian observables.}, which is directly accessible in the aforementioned computational approaches (theoretical or experimental). This therefore enables general computational modes of predicting certain forms of thermalization within the fundamental limitations indicated above, essentially by constituting an extrapolation procedure from a small amount of calculable or measurable quantum information to thermalization dynamics in at least a large fraction (approaching $1$) of physically relevant states.

A convenient setting to describe these results is the thermalization of a core subsystem $\sigma$ of a finite number $N_{\sigma}$ of particles in an initial nonequilibrium (e.g. pure) state, in contact with a thermodynamically large external bath $\eta$ of $N_{\eta}$ particles. We assume that the joint system of $N = N_{\sigma} + N_{\eta}$ particles (Hilbert space dimension $D = 2^{N}$) undergoes Hamiltonian evolution over continuous or discrete time $t$ with dynamics given by the unitary operator $\uh(t) = e^{-i \hat{H} t}$. In this setting, as quantitatively reviewed in Sec.~\ref{sec:synopsis}, it was shown that a mere finite-time study of the autocorrelator of a few-body projective observable $\proj_A$
can rigorously imply the following \textbf{Statements} on its thermalization to the thermal value $\Pi_{A,\therm} \equiv \Tr[\proj_A]/D$ (corresponding to a full Hilbert space i.e.  microcanonical average~\cite{vonNeumannThermalization}) over comparably finite or even arbitrarily long times (here, stated qualitatively for simplicity):
\begin{enumerate}
    \item \label{intro1:statement1} For any initial state $\lvert \psi\rangle_{\sigma}$ of the core and \textit{any} (arbitrary) orthonormal basis $\left\lbrace \lvert \phi_k\rangle_{\eta}\right\rbrace_k$ of the external bath, the \textit{time-averaged} expectation value of the observable is close to its thermal value for almost all basis states of the bath, over \textit{any} sufficiently long window of time $\mathcal{T}$ (which may be finite)\footnote{Here, we use the notation $|M|$ to denote the cardinality of $M$ if $M$ is a discrete set, or its length according to some natural measure (e.g. the Lebesgue measure over the real numbers) if $M$ is a continuous set.}:
    \begin{equation}
        \frac{1}{|\mathcal{T}|}\int_{\mathcal{T}} \diff t \left[\vphantom{\int}{_\sigma}\langle \psi\rvert \otimes {_\eta}\langle \phi_k\rvert \uh^\dagger(t)\ \proj_A\ \uh(t) \lvert \psi\rangle_{\sigma} \otimes \lvert \phi_k\rangle_{\eta}\right] \approx \Pi_{A,\therm},\;\;\ \text{ for almost all } k.
        \label{eq:intro_timeavgthermalization}
    \end{equation}
    \item \label{intro1:statement2} For \textit{any} sufficiently mixed initial state of the external bath, or equivalently, a sufficiently large statistical ensemble $\mathcal{B}_{\eta}$ of bath states $\hat{\rho}_{\eta} = |\mathcal{B}_{\eta}|^{-1}\sum_{k \in \mathcal{B}_{\eta}} \lvert \phi_k\rangle_{\eta}\langle \phi_k\rvert$ (such as a thermal state that may develop in the bath if it is entangled with an external reservoir), the \textit{instantaneous} expectation value of the observable for any initial pure state of the core is close to its thermal value at almost all times in any sufficiently long time interval $\mathcal{T}$ (which may be finite):
    \begin{equation}
        \frac{1}{|\mathcal{B}_{\eta}|}\sum_{k\in \mathcal{B}_{\eta}}\left[\vphantom{\int}{_\sigma}\langle \psi\rvert \otimes {_\eta}\langle \phi_k\rvert \uh^\dagger(t)\ \proj_A\ \uh(t) \lvert \psi\rangle_{\sigma} \otimes \lvert \phi_k\rangle_{\eta}\right] \approx \Pi_{A,\therm},\;\;\ \text{ for almost all } t \in \mathcal{T}.
        \label{eq:intro_stateavgthermalization}
    \end{equation}
\end{enumerate}
These statements are respectively analogous to classical ergodicity and classical mixing~\cite{HalmosErgodic, SinaiCornfeld, Sinai1976} of \textit{specific} observables~\cite{KhinchinStatMech} in classical statistical mechanics. All that is necessary for these statements to hold (over a finite time interval) is that the autocorrelator of the observable remains close to its thermal value ($\proj_{A,\therm}^2$, the square of the thermal value  of the observable) within a finite resolution for an overwhelming fraction of times\footnote{It is sufficient for the time-average of the autocorrelator to be close to thermal for the first statement to hold.} in a (potentially finite) set of times $\mathcal{T}_{\obs}$ such that $\lvert \mathcal{T}\rvert \gg \lvert \mathcal{T}_{\obs}\rvert$:
\begin{equation}
    \frac{1}{D}\Tr[\proj_A(t)\proj_A] \approx \Pi_{A,\therm}^2\;\;\ \text{ for almost all } t \in \mathcal{T}_{\obs}.
\end{equation}
As an incidental implication, the same behavior of autocorrelators also constrains the structure of the matrix elements of the observable in the energy eigenbasis (if in place of a core-bath tensor product, one considers the initial states in the above expressions to be a basis of energy eigenstates), implying eigenstate thermalization~\cite{vonNeumannThermalization, JensenShankarETH, deutsch1991eth, srednicki1994eth} in almost all energy eigenstates and mildly suppressed off-diagonal matrix elements.

It is worth emphasizing that each of the two Statements \ref{intro1:statement1} and \ref{intro1:statement2} above involves a statistical average: an average over times in the first statement (but applying to almost all states of the bath), and an average over a large fraction of states of the bath in the second (but applying to almost all times). Our main problem of interest in this work is to eliminate the need for a statistical average entirely. This is because the phenomenon of quantum thermalization is generally expected and observed to occur at almost all times for most pure states of the bath in sufficiently complex systems, without further averages. We would like to explore the extent to which this stronger phenomenon of pure state thermalization is predictable for large systems with finite resources.

We note that Statement \ref{intro1:statement2} does allow the statistical average to be eliminated to a limited extent, with the help of typicality arguments~\cite{vonNeumannThermalization, tumulka_CT, CanonicalTypicalityPSW, NormalTypicality}. These arguments rigorously show that almost all pure states in a large Hilbert space, sampled according to the Haar measure, effectively behave like a statistical ensemble of all states (the maximally mixed state) within the Hilbert space. In the above setting, this means that a typical state of a bath behaves like the aforementioned statistical ensembles of bath states in Statement~\ref{intro1:statement2}, and should therefore attain thermal equilibrium in much the same way as a large mixed state. The primary shortcoming of such an argument is essentially the standard weakness of typicality arguments: the Haar measure is sensitive to the full size of the large Hilbert space, and extremely insensitive to most physical states that one may be interested in in practice, such as a finite set of e.g. computational basis states (bitstrings of qubits)~\cite{NielsenChuang}, whose size by the Haar measure is zero.
Therefore, the typicality argument applied to Statement~\ref{intro1:statement2} eliminates statistical averages for typical states of the bath, even in a nonequilibrium process where the core is initially nonthermal, but does not have sufficient resolution to address any given discrete set of bath states that may be of physical interest.

To address this shortcoming, we will develop a method of rigorously identifying (and by extension, proving or predicting) quantum thermalization in more accessible classes of pure initial states of the bath without any statistical average, just from calculations or observations of few-body dynamics over an accessible range of times, without assuming prior knowledge of the dynamics (i.e., we seek model-independent predictions). We will show that the key to such a method involves out-of-time-ordered correlators (OTOCs)~\cite{LarkinOvchinnikov, ButterflyShenkerStanford, StringyShenkerStanford, MSSotocBound, GJWOTOCReview} rather than mere autocorrelators. This is motivated below in Sec.~\ref{sec:statisticalaverages_class}, where we note that statistical averages exist in the methods of Ref.~\cite{dynamicalqthermalization} due to its compatibility with a classical limit (if relevant in a system). To eliminate these statistical averages, we must break compatibility with the classical limit, which requires sensitivity to distinctly quantum considerations~\cite{DiracQM, ShankarQM} such as the ordering of operators within correlation functions.

\subsection{Statistical averages are necessary in classical statistical mechanics}
\label{sec:statisticalaverages_class}

Our key intuition from classical statistical mechanics is that mixing, or any other ergodic property that concerns classical dynamics~\cite{HalmosErgodic, SinaiCornfeld, Sinai1976}, has no direct connection to thermal \textit{equilibrium} in individual states at specific times~\cite{KhinchinStatMech, GallavottiErgodic, BricmontErgodicityCritique, GoldsteinErgodicityCritique, LebowitzAPS2021}. This is in contrast to ergodic classical dynamics, which does correspond to \textit{time-averaged} thermalization in almost all individual states (but even then, is best restricted to an observable-specific variant if one is interested in accessible timescales~\cite{KhinchinStatMech}). Classical thermal equilibrium is instead attained largely for global intensive observables with a strong ``concentration'' property~\cite{KhinchinStatMech, GallavottiErgodic, BricmontErgodicityCritique, GoldsteinErgodicityCritique, LebowitzAPS2021} (e.g. pressure, various densities): the observable attains its thermal value, \textit{by default}, in almost all states, \textit{irrespective} of dynamics. Even in rare nonequilibrium initial states, it is sufficient for the dynamics to explore a non-negligible portion of the phase space for the observable to thermalize at almost all times, without requiring any stronger ergodic properties~\cite{GallavottiErgodic}. Such global intensive observables are of immediate relevance for thermodynamics, where one expects not to have enough resolution to probe each particle in the system individually.

Thermalization via concentration has a trivial quantum analogue: almost all eigenvalues of an observable concentrate near a thermal value, in which case almost all states in any orthonormal basis (with dynamics that explores even a minuscule fraction of the Hilbert space) see the thermal value of this observable for an overwhelming fraction of times. Such a property may be readily verified, e.g., for a hard-core bosonic or fermionic gas on a lattice $L$ of $|L|=N$ sites, say for the observable $n_M$ representing the particle number density in a subset $M$ of lattice sites with $|M| = f_M N$, spanning a fraction $0<f_M \leq 1$ of the original lattice. With $n_k$ representing the particle number at the $k$-th lattice site with eigenvalues $\lbrace 0, 1\rbrace$, the aforementioned number density observables are formally given by
\begin{equation}
    n_M = \frac{1}{f_M N}\sum_{k \in M} n_k,
    \label{eq:intro_local_density}
\end{equation}
which concentrate near $1/2$ (i.e. almost all eigenvalues are within a vanishingly small distance of $1/2$) for fixed $f_M>0$ as $N\to\infty$.
The main nontrivial question of thermalization dynamics for such observables concerns whether a rare non-equilibrium state explores enough of the near-thermal subspace of $n_M$, which is often understood in system-specific ways (e.g., no $n_L$-conservation), due to the fraction of nonthermal states being inaccessibly small (even within a single orthonormal basis) for general statements about ``almost all'' states to have strong implications. However, in follow-up work~\cite{dynamicalqaperiodicity}, we show that general \textit{state/ensemble-dependent} computable criteria may indeed be formulated to establish the thermalization of concentrated observables in such rare states/ensembles of physical interest over finite timescales.

A different scenario is relevant in \textit{local} statistical mechanics, beyond classical thermodynamics, where thermalization via concentration is impossible. Even in a thermodynamically large system, one can in some cases expect to be able to measure properties of a few degrees of freedom, e.g., probe a few specific sites on the above lattice. This situation may arise if one imagines coupling a finite subsystem compatible with external probes to a thermodynamically large bath. Moreover, in experiments on many quantum systems, such as in quantum simulators with local control, it is possible to measure even ultra-local observables, such as an individual $n_k$ corresponding to one lattice site~\cite{QuantumSimulationReview2024, AndreasPurityRM, GoogleScrambling, Joshi2020, GoogleScrambling2}. In a classical system, such an observable can never attain thermal equilibrium at a given instant of time: the thermal value is typically $(n_k)_{\therm} \in (0,1)$, e.g. $(n_k)_{\therm} = 1/2$ for infinite temperature thermalization, which is not a classically allowed value of $n_k$. To produce agreement with a thermal value, one usually needs a statistical average that aggregates different outcomes of $0$ and $1$ into the thermal value $1/2$: either an average over times, or an average over initial states. The former corresponds to ergodicity~\cite{HalmosErgodic, SinaiCornfeld, Sinai1976} (more precisely, an observable-specific variant of ergodicity~\cite{KhinchinStatMech}), where time averages in almost all initial states are nearly thermal; the latter corresponds to (observable-specific) mixing, where averages over a large ensemble of states at any sufficiently long time are nearly thermal. Indeed, even Eq.~\eqref{eq:intro_local_density} can be viewed as a different kind of statistical average --- an average over lattice sites --- and therefore not really immune to the classical requirement of statistical averages\footnote{Another common statistical average is an \textit{ensemble} average over dynamics, such as over a class of related Hamiltonians or unitary evolution operators~\cite{RQCreview_2023}. This again often has the effect of reducing quantum dynamics to nearly classical dynamics, but our concern here is with dynamics in an individual system, so we will largely ignore this kind of average except to justify typicality statements in Sec.~\ref{sec:dynamics}. However, we note that as we are dealing with bounded (non-negative) quantities, ensemble averages may help apply our general results to specific classes of models, where the smallness of ensemble-averaged quantities (if significantly small) implies their smallness in almost all systems of the ensemble~\cite{RossProbability}.}.

The phenomenon of local quantum thermalization is fundamentally different: while the measurement outcomes of $\hat{n}_k$ remain restricted to $\lbrace 0,1\rbrace$, the expectation value $\langle \psi\vert \hat{n}_k\rvert \psi\rangle$ in a pure state $\lvert \psi\rangle$ can attain the thermal value $(n_k)_{\therm}$ at any given time. Although an expectation value is measured as a statistical average of different outcomes in practice, in the mathematical formalism of quantum mechanics it remains an instantaneous property of quantum states that most directly corresponds to the instantaneous value of observables in classical states in the classical limit~\cite{DiracQM, ShankarQM}. Therefore, no explicit statistical average occurs in the formalism for thermalization in a pure state at a given instant of time. This is of particular interest to us for the following reason.

The methods of Ref.~\cite{dynamicalqthermalization} address the expectation values of observables in quantum states, but were constructed to be compatible in the appropriate limit(s) with a classical result~\cite{KhinchinStatMech} connecting classical autocorrelator behavior to the ergodic behavior of an observable (involving a time average), and semiclassical results~\cite{Shnirelman,  CdV, ZelditchOG, ZelditchMixing, Sunada, Zelditch, Anantharaman} connecting quantized eigenstate structure to classical ergodicity and mixing (involving time or state averages) for quantized classical systems. And perhaps unsurprisingly (given the above arguments) due to this mode of construction, they work best in predicting the thermalization of expectation values with statistical averages: either time-averaged for almost all pure states in any basis, or averaged over a nonvanishing fraction of basis states at almost all times. This persistent reliance on statistical averages in all previously developed classical or quantum methods must be circumvented to allow a computational framework to fully predict quantum thermalization.

\subsection{Overview and organization of this paper}
\label{sec:overview}

The goal of this paper is then to design a general model-independent method to predict quantum thermalization in individual pure states without statistical averages, strictly from few-body and finite-time observations even with only finite resolution. In particular, all we will assume are a Hilbert space (in practice, of large dimension), and unitary dynamics on this space (which may be time-dependent). As described in Sec.~\ref{sec:computationalmotivation}, this would provide an experimentally accessible probe of quantum thermalization, while also enabling analytical or computational proofs of thermalization in arbitrary quantum systems of any size.
Having isolated the limitations of Ref.~\cite{dynamicalqthermalization} to its compatibility with a classical limit, we will explicitly motivate our results via the requirement that they fail in the classical limit. We will show that this requirement can be met by considering the non-classical behavior of OTOCs~\cite{LarkinOvchinnikov, ButterflyShenkerStanford, StringyShenkerStanford, MSSotocBound, GJWOTOCReview}, which do not remotely appear in any form in the aforementioned classical~\cite{KhinchinStatMech} or semiclassical approaches~\cite{Shnirelman, CdV, ZelditchOG, ZelditchMixing, Sunada, Zelditch, Anantharaman}, in contrast to the conventional 2-point correlators in Ref.~\cite{dynamicalqthermalization} which can always have classical counterparts.

In Sec.~\ref{sec:synopsis}, we review the technical details of \cite{dynamicalqthermalization}, and summarize the new results of this paper, focusing on how they add to the predictive power developed in the previous work. Our main result is that the decay of ``controllably nonlocal'' OTOCs to sufficiently small values within the core system is necessary and sufficient for quantum thermalization in almost all basis states of the bath at any given time. The requisite small values are extremely sensitive to operator ordering, and therefore invisible in a classical limit (if one exists), satisfying our requirement above. While these predictions are initially restricted to times where these controllably nonlocal correlators can be constrained, we discuss (in Sec.~\ref{sec:Discussion_forecasts}) a possible strategy to extrapolate to arbitrary times from a finite set of times.

To be more specific, a ``controllably nonlocal'' OTOC involves an ordered higher-order product of a few-body observable of interest, say supported on $N_S$ qubits, with another \textit{extended} few-body observable supported on a larger set $N_{\sigma}$ qubits (say, within the core), where $N_{\sigma}$ is a large but finite number that determines the resolution to which one can establish thermalization. Due to its general independence of $N$, the total size of the system, this mode of prediction can be scaled to the thermodynamic limit of $N \to \infty$ particles in the bath, keeping $N_{\sigma}$ finite. This correlator may therefore be amenable (modulo a finite degree of tedium) to theoretical calculations in a variety of systems where local out-of-time-ordered correlators (with $N_S = N_{\sigma} = 1$) can be calculated~\cite{MaldacenaStanford, NahumOTOCCircuits, ChanScrambling, ClaeysLamacraftMaxVelCkt, BertiniPiroliScrambling, XuSwingleScramblingTutorial}, as well as experimental computation using a number of available protocols~\cite{ReyOTOCMeas2017, GoogleScrambling, Joshi2020, OTOC2022ExptFiniteTemp, GoogleScrambling2} with finite but expensive resources. For the latter, we typically need $N_{\sigma} \gtrsim 24$ and a measurement resolution of $10^{-7}$ to show $N_S=1$-qubit thermalization with a crude resolution of $10\%$, though weaker yet nontrivial statements are possible with $N_{\sigma} \sim 12$ and a resolution of $10^{-4}$. We also note that to establish thermalization to arbitrary resolution in theoretical calculations, it is typically necessary to take $N_{\sigma} \to \infty$ \textit{after} taking $N\to\infty$, which perhaps provides a nontrivial technical challenge. Up to such technicalities, this completes our elimination of statistical averages while enabling the rigorous prediction of quantum thermalization from a finite amount of information, at least at a given time. 

Sec.~\ref{sec:subspacethermalization} contains our core pair of results in an abstract mathematical form, namely Theorem~\ref{thm:quantumthermalization_subspaces} and Proposition~\ref{prop:thermalizationsubspacesconverse}. This section is expressed in the general language of projective observables and subspaces, with the intuition that a sufficiently high-dimensional projector is a good description of a few-body observable in a many-body setting, but is otherwise more general~\cite{vonNeumannThermalization}. It is shown that the an ordering-sensitive higher order product of two high-dimensional projectors of different sizes directly determines the thermalization of the observable represented by one of them in the subspace of states represented by the other.

In Sec.~\ref{sec:dynamics}, we consider the typical behavior of these OTOCs and map the abstract but general results of Sec.~\ref{sec:subspacethermalization} to the many-body setting emphasized here. In particular, we show that in typical cases, it is necessary for the dimensions of the projectors involved in the OTOCs to differ by a large constant for pure state thermalization within the subspace of states to even be possible at any instant of time. In the many-body language, this necessitates the consideration of the controllably nonlocal OTOCs described above. We also discuss some obstacles in trying to extrapolate our present mode of prediction from finite times to arbitrarily long intervals, which necessitate alternate strategies for such extrapolations compared to Ref.~\cite{dynamicalqthermalization}.

Finally, in Sec.~\ref{sec:Discussion}, we comment on some practical issues, qualitative insights and possible future directions, namely the experimental relevance of our criteria, the hierarchy of timescales across mixed state and pure state thermalization dynamics, and the possible model-independent extrapolation from finite-time data to infinite time pure state thermalization (as is possible for thermalization with statistical averages in Statements~\ref{intro1:statement1}, \ref{intro1:statement2}) or its circumvention via model-specific calculations.

The Appendices are organized as follows. Appendix~\ref{app:autocorrelator_method} mainly reviews the methods and some relevant results of Ref.~\cite{dynamicalqthermalization}, as well as the obstacles that prevent a direct extension of these methods to thermalization without averages. Appendix~\ref{app:Haar} provides rigorous typicality estimates for our results using standard integral relations and concentration properties for the unitary group, in line with the general ``typicality'' viewpoint by which different properties of sufficiently complex systems approach those of random unitary dynamics to various extents~\cite{vonNeumannThermalization, tumulka_CT, CanonicalTypicalityPSW, NormalTypicality}. Appendix~\ref{app:proofs} contains the proof of our main result, Theorem~\ref{thm:quantumthermalization_subspaces}, as well as a more concise version of one of the primary results of Ref.~\cite{dynamicalqthermalization}, Theorem~\ref{thm:autocorrelatorimpliescorrelator_Gen}, that is repeatedly invoked here for comparative purposes.

\section{Synopsis}
\label{sec:synopsis}

We show that out-of-time-ordered correlators (OTOCs)~\cite{LarkinOvchinnikov, ButterflyShenkerStanford, StringyShenkerStanford, MSSotocBound, GJWOTOCReview}, a standard and widely studied class of quantities in the context of ``quantum chaos'' (although here occuring in a slightly nonlocal variant that doesn't appear as widely studied), allow one to rigorously predict thermalization in an overwhelming fraction of large classes of pure states, thus resolving the problem of eliminating statistical averages from accessible thermalization predictions. These quantities are usually held to measure forms of dynamics that are generally distinct from thermalization, such as operator spreading and sensitivity to perturbations (famously associated with ``the butterfly effect'')~\cite{LarkinOvchinnikov, ButterflyShenkerStanford, StringyShenkerStanford, MSSotocBound, GJWOTOCReview, XuSwingleScramblingTutorial}. It has also been explicitly suggested~\cite{MSSotocBound, ShenkerThouless, Susskind_dS_scrambling} that OTOCs are sensitive to more nontrivial forms of quantum dynamics than thermalization, and measure a(n apparently) distinct phenomenon called ``scrambling''\footnote{The term ``scrambling'' is sometimes used to describe OTOC saturation~\cite{StringyShenkerStanford,MSSotocBound}. While this is a semantic choice, our preference is to associate scrambling with operational recovery protocols~\cite{HaydenPreskill, YoshidaKitaev} that led to the introduction of the term, which are more directly related to mutual information and operator entanglement~\cite{YoshidaKitaev, HosurQiRobertsYoshida, scramblingSWAP}. Entanglement-based scrambling dynamics can also be constrained more rigorously in a system-independent manner~\cite{dynamicalqspeedlimit, dynamicalqfastscrambling} than OTOC dynamics~\cite{MSSotocBound}. While an average of OTOCs over relevant complete bases of operators reduces to entanglement or mutual information for unitary dynamics~\cite{HosurQiRobertsYoshida, YoshidaYao}, we note that this averaged quantity may no longer bear the recognizable hallmarks of individual OTOCs. This is similar to an average of autocorrelation functions $\Tr[e^{-i \hat{H} t} \hat{A} e^{i\hat{H} t}\hat{A}^\dagger]$ over a complete basis of operators $\hat{A}$ giving~\cite{ChaosComplexityRMT} the spectral form factor $\lvert\Tr(e^{-i\hat{H} t})\rvert^2$, despite individual autocorrelators having different dynamics from the spectral form factor in general (e.g. energy eigenbasis/eigenoperator autocorrelators are constant or sinusoidal for any spectral form factor). From this viewpoint, we also regard the results of this work as providing a more direct operational characterization of individual OTOCs in terms of pure state thermalization for the most general quantum system with unitary dynamics, beyond their looser connection via operator-averaging to operational scrambling protocols.}. However, rather surprisingly (to us), we will find that OTOCs serve as direct probes of the very same standard phenomenon of quantum thermalization as well as its absence, in our case of pure states rather than correlators or mixed states, without discernible additional information of importance (in our general model-independent setting) as encoded in their decay and saturation behavior. 

Most crucially for our goal, local OTOCs can be analytically calculated in several systems~\cite{NahumOTOCCircuits, ChanScrambling, ClaeysLamacraftMaxVelCkt, BertiniPiroliScrambling, RQCreview_2023, XuSwingleScramblingTutorial} and experimentally measured via (often) practically tested measurement protocols~\cite{SwingleOTOCprotocol, HafeziOTOCprotocol, ReyOTOCMeas2017, Vermersch2019, Joshi2020, GoogleScrambling, OTOC2022ExptFiniteTemp, GoogleScrambling2}, which we take to portend comparable advantages for the \textit{slightly} nonlocal OTOCs we are interested in given sufficient near-term advances in computational and experimental techniques. They also do not have a well-defined saturation behavior in the classical limit, and it is precisely this saturation behavior that will be key to our results, satisfying our main requirement of being incompatible with a classical limit. We emphasize that mere local operator spreading or sensitivity to perturbations (which can have classical counterparts) are not sufficient for OTOC saturation, and should not be used to conclude thermalization. However, to the extent that they are responsible for the lead-up to OTOC saturation in local systems, they may be considered \textit{initial} mechanisms that can lead towards pure state thermalization, and our results allow a rigorous justification of this attribution.


\subsection{Previous work: Predicting thermalization with statistical averages}
\label{sec:synopsis_prevwork}

To summarize the results of Ref.~\cite{dynamicalqthermalization}, it is convenient to think of $2$-point correlation functions as inner products of operators, such as $\langle \hat{B}, \hat{A}\rangle = D^{-1}\Tr[\hat{B}^\dagger \hat{A}]$. Then, our technical problem statement is simple:
Given $\langle \hat{A}(0), \hat{A}(t)\rangle$ over $t\in\mathcal{T}_{\obs}$ under Hamiltonian (autonomous unitary) dynamics with\footnote{This convention still corresponds to operator evolution in the Schr\"{o}dinger picture of evolving states rather than the Heisenberg picture of dynamical operators~\cite{ShankarQM}; specifically, we imagine that $\hat{A}(t)$ describes the evolution of a density operator representing a general mixed state of the system (a correspondence that is straightforward to make if $\hat{A}$ is a projective observable) rather than an observable as such.} $\hat{A}(t) = \uh(t) \hat{A} \uh^\dagger(t)$, we want to predict $\langle \hat{B}(t_2), \hat{A}(t_1)\rangle = \langle \hat{B}(0), \hat{A}(t_1-t_2)\rangle$ for all other operators $\hat{B}$ and all times $(t_1,t_2)$ to the extent possible.

This is addressed by the following inequality, derived in Appendix~\ref{app:autocorrelator_method} [footnote \ref{footnote:autocorrelatortheoremspecifics}] as a special but representative case of Theorem~\ref{thm:autocorrelatorimpliescorrelator_Gen} (which combines and simplifies some results of Ref.~\cite{dynamicalqthermalization}): for any arbitrarily chosen constant $\xi > 0$, 
\begin{equation}
    \frac{1}{T}\int_{t_0}^{t_0+T}\diff t\ \frac{\langle \hat{B}(0), \hat{A}(t)\rangle}{\sqrt{\langle \hat{B},\hat{B}\rangle \langle \hat{A},\hat{A}\rangle}} \leq \frac{T_{\obs}}{\xi T}+\frac{\xi}{|\sin\xi|}\sqrt{\int_{-T_{\obs}}^{T_{\obs}}\frac{\diff t}{T_{\obs}}\ \left(1-\frac{|t|}{T_{\obs}}\right)\frac{\langle \hat{A}(0), \hat{A}(t)\rangle}{\langle \hat{A},\hat{A}\rangle}},
    \label{eq:synopsis_autocorrelator_to_correlator}
\end{equation}
from which it intuitively follows that the left hand side is small if the autocorrelator remains small on average over $[-T_{\obs}, T_{\obs}]$, and if $T \gg T_{\obs}/\xi$. Here, $\hat{A} = \hat{A}(0)$ and likewise for $\hat{B}$. The constant $\xi$ here sets a tradeoff between autocorrelator measurements and time intervals. For small $\xi \ll 1$, we only need to measure autocorrelators to low resolution, but then Eq.~\eqref{eq:synopsis_autocorrelator_to_correlator} would only provide a useful constraint for large $T \gg T_{\obs}/\xi$. For appreciable $\xi \sim 1$ e.g. near a node of $\sin \xi$, such as $\xi \approx \pi$, we can even constrain correlators for $T \sim T_{\obs}/\pi$, provided we know that the autocorrelator remains small on average to a resolution of approximately $(\xi-\pi)^2$.

The above inequality implies thermalization with statistical averages with the following setup:
\begin{enumerate}
    \item Take $\hat{B} = \hat{\rho}$ to be a density operator (i.e. an ensemble of initial states) with $\Tr[\hat{\rho}^2] \leq 1/(\mu D)$ for some finite constant $\mu$ satisfying $0 < \mu \leq 1$.  
    \item Without loss of generality, take $\hat{A}$ to have a thermal value of $A_{\therm} = 0$ (if not, consider $\hat{A} - A_{\therm}\idop$ in place of $\hat{A}$) --- this is just for convenience.
\end{enumerate}
Then, Eq.~\eqref{eq:synopsis_autocorrelator_to_correlator} implies that the decay of the autocorrelator of $\hat{A}$ to small values on average is sufficient to guarantee the time-averaged thermalization in the mixed state $\hat{\rho}$. As this applies to any mixed state $\hat{\rho}$ (including ensembles of pure states that all deviate from the thermal value with a specific sign), it follows that $\hat{A}$ attains a time-averaged thermal value almost all pure states in any finite set of orthonormal states (involving at least $\Theta(D)$ states so that $\Tr(\hat{\rho}^2) \leq 1/(\mu D)$) e.g. an orthonormal basis~\cite{dynamicalqthermalization}. For state-averaged thermalization at almost all times, one can instead take $\hat{B} = \hat{\rho} \otimes \hat{\rho}$ and $\hat{A} \to \hat{A}\otimes \hat{A}$ so that the time-averaged expectation value above is over a nonnegative integrand, which implies the thermalization of $\hat{A}$ in $\hat{\rho}$ at almost all times. These statements are discussed in more detail in Appendix~\ref{app:correlators_examples}. Adapting these statements to the core-bath subsystem decomposition in Sec.~\ref{sec:intro} gives Eqs.~\eqref{eq:intro_timeavgthermalization} and \eqref{eq:intro_stateavgthermalization}; see also Sec.~\ref{sec:restoredynamics} for additional specifics.

The overall strategy here is to substitute knowledge of the properties of energy eigenstates, which are difficult to calculate in general, with a knowledge of few-body correlators over finite times, which are among the most straightforward objects to calculate or measure. It is possible to anchor this implication on eigenstate structure, via a notion of ``energy-band thermalization''~\cite{dynamicalqthermalization} (which implies eigenstate thermalization in almost all states and, in its cloned form, mildly suppressed off-diagonal matrix elements), but it is not necessary to do so as seen in Eq.~\eqref{eq:synopsis_autocorrelator_to_correlator}. The qualitative implication structure then becomes, with e.g. the restriction to ``almost all states''  in a basis being understood:
\begin{align}
    \text{Autocorrelator decay} \iff \text{Energy-band } &\text{thermalization} \implies \text{(Averaged) finite-time thermalization.} \nonumber \\*
    &\rotatebox[origin=c]{270}{$\implies$} \\*
    \text{Eigenstate } &\text{thermalization} \implies \text{(Averaged) infinite-time thermalization.} \nonumber
\end{align}
This shows that autocorrelator decay is, to this resolution (i.e. thermalization in almost all as opposed to all states) a stronger statement than eigenstate thermalization. It is in this spirit that we find it appropriate to try to consider the direct implications of simple correlation functions of an observable for thermalization in arbitrary states, without referring to energy eigenstates.

As noted in Sec.~\ref{sec:intro}, the above statements seem to unavoidably require some form of statistical average: either a time-average or a state-average (with the latter being lifted somewhat for mathematically typical states whose physical relevance cannot be guaranteed). But guided by the implication structure above, we are now equipped to ask if there are additional few-body correlation functions of an observable that can constrain thermalization in pure states without such averages or typicality arguments. We emphasize that this will evolve into a somewhat drastic ask: we will develop a method to show pure state quantum thermalization, with only finite resources, in \textit{every} conceivable basis of the bath. This guarantees thermalization for any basis of physically relevant states, but is by no means necessary for any given finite set of relevant bases to thermalize, though we expect that the thermalization of states of physical interest is often tied to our criterion in practice e.g. when they belong to an \textit{atypical} basis (which may frequently include the computational basis).

\subsection{This work: OTOCs establish thermalization without averages}
\label{sec:synopsis_thiswork}

Our primary strategy, as motivated in Sec.~\ref{sec:intro}, is to identify a method of predicting thermalization in individual states that is guaranteed to fail for classical systems. This requirement of failure almost directly suggests a solution: classical dynamics does not alter the commutator of two commuting observables (even in the Koopman-von Neumann Hilbert space language~\cite{Koopman_KvN, vonNeumann_KvN1, vonNeumann_KvN2, HalmosErgodic, SinaiCornfeld}).
Thus, any method that relies on nonzero commutators will satisfy this requirement\footnote{\label{footnote:commutatorPB}While the commutator $[\hat{A}, \hat{B}]$ is well-known to reduce to the Poisson bracket $[A,B]_{\text{PB}}$ in the classical limit (for systems that have such a limit)~\cite{DiracQM, ShankarQM}, we recall that the correspondence involves Planck's constant $\hbar$, via $[\hat{A},\hat{B}] \sim i \hbar[A,B]_{\text{PB}}$. Thus, for any finite Poisson bracket, the commutator remains zero in the classical $\hbar \to 0$ limit. This will be useful to keep in mind for the discussion in Sec.~\ref{sec:Discussion}.}.

An alternative, but essentially equivalent, viewpoint is that any correlation function that depends critically on operator ordering does not have a good classical limit. This implicates correlation functions such as out-of-time-ordered correlators (OTOCs)~\cite{LarkinOvchinnikov, ButterflyShenkerStanford, StringyShenkerStanford, MSSotocBound, GJWOTOCReview} as suitable candidates that meet our criterion. The details of how OTOCs constrain thermalization without averages is developed in the bulk of this paper, with the presentation giving precedence to some mathematical notions to emphasize technical intuition and exactness. In this section, we will provide an intuitive summary of our results in the familiar setting of a many-body system in which local observables are of interest.

As our focus is not on the structure of energy eigenstates with respect to observables (or any other properties of energy eigenstates or eigenvalues), we do not discuss in any particular detail the extensive recent work on the connection between higher order correlators such as OTOCs and the matrix elements of observables in the energy eigenbasis, stemming from the observations in Refs.~\cite{FoiniKurchanETH, PappalardiFoiniKurchanETH}, whose developments are generally orthogonal to the goal of this work. For example, our present results apply equally well to time-dependent unitary dynamics as to Hamiltonian dynamics, and do not require the existence of (or access to the) energy eigenstates or eigenvalues.

As in Sec.~\ref{sec:intro}, our primary setting of convenience is that of a core subsystem $\sigma$ of dimension $D_{\sigma}$ (e.g. with $N_{\sigma}$ qubits, $D_{\sigma} = 2^{N_{\sigma}}$), in contact with a thermal bath $\eta$ of dimension $D_{\eta}$ (e.g. with $N_{\eta}$ qubits, $D_{\eta} = 2^{N_{\eta}}$). The joint system evolves under a unitary $\uh(t)$, that is now allowed to be time-dependent (usually with $\uh(0) = \idop$). We are also interested in a projective few-body observable $\proj_R$ (satisfying $\proj_R^\dagger = \proj_R$, $\proj_R^2 = \proj_R$), that projects onto a pure state within an ``observed'' subsystem (core, bath, or partially on both) of dimension $D_S$ (of $N_S$ qubits, $D_S = 2^{N_S}$). For example, $\proj_R$ may represent the probability that all of the $N_S$ qubits within the observed subsystem are each in the basis state $\lvert 0\rangle$. We will assume that the core is small, e.g., $D_{\sigma} = \Theta(1)$ in a thermodynamic limit, as is the observed subsystem $D_S = \Theta(1)$, while the bath can be thermodynamically large $D_{\eta} \to \infty$. 

Statements \ref{intro1:statement1} and \ref{intro1:statement2} in Sec.~\ref{sec:intro} discuss the thermalization of $\proj_R$ for almost all pure states $\lvert \phi\rangle_{\eta}$ in $\eta$ if averaged over a large range of times, and for almost all times if averaged over all $\lvert \phi\rangle_{\eta}$. To eliminate both averages, we consider a projector onto the ensemble of all bath states for a fixed core state $\lvert \psi\rangle_{\sigma}$:
\begin{equation}
    \proj_{\psi} \equiv \lvert \psi\rangle_{\sigma}\langle \psi\rvert \otimes \idop_{\eta}.
    \label{eq:synopsis_projectorpsi}
\end{equation}
Importantly, this is also a few body observable: its measurement involves a projection onto $\lvert \psi\rangle_{\sigma}$ in the core with absolutely no action on the bath. Two dynamical correlation functions involving these projectors are of interest. The first is the conventional 2nd order correlator of both observables at $t$:
\begin{equation}
    G_{R\psi}^{(2)}(t) = \frac{D_{\sigma}}{D} \Tr[\proj_{\psi}(t)\proj_R]
\end{equation}
where $\proj_{\psi}(t) \equiv \uh(t) \proj_{\psi} \uh^\dagger(t)$, and the multiplication by the $\Theta(1)$ factor $D_{\sigma}$ is for formal convenience. In general, this is an $O(1)$ nonnegative quantity that is constrained to lie in $[0,1]$, whose typical values are finite.
The second is the out-of-time-ordered correlator at time $t$ may be defined as:
\begin{equation}
    G_{R\psi}^{(4)}(t) \equiv \frac{D_{\sigma}}{D}\Tr\left[\proj_R\proj_\psi(t)\proj_R\proj_{\psi}(t)\right],
\end{equation}
Their difference equals the (rescaled) Hilbert-Schmidt or Frobenius norm~\cite{HaarBook} of the commutator of the two projectors:
\begin{equation}
    \frac{1}{2D}\Tr\left\lbrace [\proj_R, \proj_{\psi}(t)]^\dagger [\proj_R, \proj_{\psi}(t)]\right\rbrace = G_{R\psi}^{(2)}(t)-G_{R\psi}^{(4)}(t),
\end{equation}
relating nonclassical values of the OTOC to a nonzero commutator~\cite{LarkinOvchinnikov, GJWOTOCReview} in line with our motivating arguments above.

Given these correlators, our main result is, informally:
\begin{statement}[Thermalization without averages from correlators, informal]
\label{statement:informalresult_synopsis}
    At each time $t$ and for a given few-body core state $\lvert \psi\rangle_{\sigma}$, the few-body observable $\proj_R$ thermalizes to the second-order correlator $G_{R\psi}^{(2)}(t)$ in almost all states $\lvert \phi\rangle_{\eta}$ in any orthonormal basis for the bath $\eta$,
    \begin{equation}
        {_{\sigma}}\langle \psi\rvert \otimes {_{\eta}}\langle \phi\rvert \uh^\dagger(t)\ \proj_R\ \uh(t) \lvert \psi\rangle_{\sigma} \otimes \lvert \phi\rangle_{\eta} \approx G_{R\psi}^{(2)}(t)
        \label{eq:synopsis_thermalization_statement_eq1}
    \end{equation}
    ``if and only if'' the out-of-time-ordered correlator nearly factorizes into the square of the second-order correlator:
    \begin{equation}
        G_{R\psi}^{(4)}(t) \approx \left[ G_{R\psi}^{(2)}(t)\right]^2.
    \end{equation}
\end{statement}
\begin{proof}[Justification]
    This is an informal statement of Corollary~\ref{cor:manybodypurestatethermalization}, which is an application of Theorem~\ref{thm:quantumthermalization_subspaces} and Proposition~\ref{prop:thermalizationsubspacesconverse} to the many-body setting. Its key emphasis is the mapping of thermalization in almost all complex many-body pure states to a condition on relatively accessible few-body correlators.
\end{proof}

We note that unlike Eq.~\eqref{eq:synopsis_autocorrelator_to_correlator} and its consequences for thermalization with statistical averages, Statement~\ref{statement:informalresult_synopsis} is an instantaneous relation that holds at any given time $t$ rather than a prediction that extrapolates from a short range of times to arbitrarily long ones. This allows the latter to be more general in applying to \textit{time-dependent} unitary dynamics, but even in the Hamiltonian case, we have not found an extrapolation procedure in time that works. The main strength of this result is then its ability to use a finite amount of information from few-body dynamics to predict \textit{pure state} thermalization without averages in a thermodynamically large number of states \textit{at the same instant of time}. By way of comparison, the dynamical connection to OTOCs without a classical counterpart in Statement~\ref{statement:informalresult_synopsis} surpasses these previous results in terms of being able to access pure state quantum thermalization without averages, but the full formalism for extrapolating from finite-time computations to arbitrary-time predictions for Hamiltonian systems remains to be established in the fully quantum case. We will return to this point in Sec.~\ref{sec:restoredynamics} and Sec.~\ref{sec:Discussion}. 

In almost all cases, we expect that $D_{\sigma} \gg D_S$ is a necessary condition for Statement~\ref{statement:informalresult_synopsis} to be useful, i.e. the core should be larger than the observed subsystem by a ``significant'' amount. This corresponds to ``controllably nonlocal'' OTOCs, as these quantities were termed in Sec.~\ref{sec:intro}. Measuring these may be expensive; it is shown in Sec.~\ref{sec:manybodysystems} that to establish the thermalization of a single qubit, a nonlocal OTOC with $N_{\sigma} \gtrsim 24$ should be measured to an accuracy of $10^{-7}$ for the approximate equality in Eq.~\eqref{eq:synopsis_thermalization_statement_eq1} to hold to an accuracy of $10\%$ for $90\%$ of bath basis states under e.g. random unitary dynamics. A more near-term computational or experimental goal may be to rigorously establish thermalization in \textit{at least} a non-negligible fraction of states in \textit{every} orthonormal basis, which still constitutes thermodynamically many states and is a stronger thermalization diagnosis than previously possible; e.g. single-qubit thermalization to an accuracy of $0.1$ for $10\%$ of states in every bath basis requires a more accessible resolution of $10^{-4}$ with $N_{\sigma} \gtrsim 12$.

From a theoretical standpoint, the calculation of these correlators also presents a higher challenge than the calculation of their single-site counterparts~\cite{NahumOTOCCircuits, ChanScrambling, ClaeysLamacraftMaxVelCkt, BertiniPiroliScrambling, XuSwingleScramblingTutorial}, but again requires handling only a finite amount of tedium for finite resolution and therefore should be solvable in principle. For the theoretically ideal scenario of showing thermalization in the core to arbitrary resolution, we do typically require a weaker limit of $N_{\sigma} \to \infty$ (after taking $N \to \infty$ with fixed $N_S$ and $N_{\sigma}$) so that the approximate equalities above approach exactness. As this also requires the core to become infinitely large (but still a vanishingly small fraction of the total system), such theoretical calculations are most useful if they can show the eventual smallness of the OTOC for (at least) an overwhelming fraction of a set of initial states of interest (e.g., a specific orthonormal basis) in the core, so that one does not require the state of an essentially infinite number of core qubits to be fine-tuned to a specific $\lvert \psi\rangle_{\sigma}$.

We have expressed Statement~\ref{statement:informalresult_synopsis} in terms of the thermalization of $\proj_R$ to $G_{R\psi}^{(2)}(t)$ in pure states at time $t$, to formally allow the thermal value to depend on the initial state of the core and on time. In many-body systems, this kind of generality is realistically not necessary in most cases. In practice (e.g. in the absence of any locally accessible conserved quantities), we usually want the thermal value to be constant and independent of the initial state of the core i.e. ``global thermalization'', which necessarily imposes $G_{R\psi}^{(2)}(t) \to 1/D_S$ (as obtained by averaging over an orthonormal basis of core states) which precisely corresponds to the ``mixing'' value of these functions (i.e. thermal value in mixed rather than pure states). In such cases, global quantum thermalization requires that mixing takes place, whose presence can in turn be determined via the techniques described in Sec.~\ref{sec:synopsis_prevwork}.

More generally, in the presence of (accessible) conserved charges $\hat{Q}_k$, the thermal value can depend on the distribution of the initial state across the eigenspaces of each $\hat{Q}_k$. In general, highly mixed product states such as in Eq.~\eqref{eq:synopsis_projectorpsi} may be distributed over a broad range of such eigenspaces, and cannot be restricted to a single eigenspace of any charge. Here, we expect that the out-of-time-ordered correlator can be filtered to a narrow range of eigenvalues of accessible conserved charges to determine thermalization in each subspace using more general interference/echo measures as described in Ref.~\cite{dynamicalqthermalization} for second order correlators. In our case, by applying additional unitary transformations generated by these conserved charges (including, possibly, the Hamiltonian $\hat{Q}_0 = \hat{H}$), i.e. generalizing to OTOCs of the form:
\begin{equation}
    G_{R\psi}^{(4)}\left(t; \lbrace s_{k1}\rbrace,\lbrace s_{k2}\rbrace,\lbrace s_{k3}\rbrace,\lbrace s_{k4}\rbrace\right) \equiv \frac{D_{\sigma}}{D}\Tr\left[\proj_R e^{-i\sum_k\hat{Q}_k s_{k1}}\proj_\psi(t)e^{i\sum_k\hat{Q}_k s_{k2}}\proj_R e^{-i\sum_k\hat{Q}_k s_{k3}} \proj_{\psi}(t) e^{i\sum_k\hat{Q}_k s_{k4}}\right],
\end{equation}
and averaging suitably over the $\lbrace s_{kj}\rbrace$, we can restrict our subspace of initial states to projectors that have a large overlap with the filtered distributions represented by the density operators:
\begin{equation}
    \hat{\rho}_{\psi, v_{kj}}(t) \equiv \int\diff s_{01} v_{01}(s_{01})\ldots \int\diff s_{K2} v_{K2}(s_{K2})\ e^{-i\sum_{k=0}^{K}\hat{Q}_k s_{k1}}\ \proj_\psi(t)\ e^{i\sum_{k=0}^{K}\hat{Q}_k s_{k2}},
\end{equation}
where we require that $\int\diff s\ v_{kj}(s) e^{-i s \tau} \geq 0$. This may allow deriving pure state thermalization in subspaces with fixed values (or narrow ranges) of each conserved charge, which may then imply thermalization of more general initial states to a value that depends on their distribution across conserved charge eigenspaces.
We do not explicitly consider this more general case any further here for simplicity, and point to the earlier work~\cite{dynamicalqthermalization} for more technical details in the context of regular $2$-point functions and thermalization with averages, which may be possible to adapt to this context in future work. The rest of this paper will therefore focus entirely on global thermalization, so that we may specifically highlight what additional ingredients are required to establish quantum thermalization without averages in the simplest setting. How these results fit into a broader framework of ``operational quantum statistical mechanics'' is described in Ref.~\cite{dynamicalqaperiodicity}.

\section{Geometry: Quantum thermalization and aligned subspaces}
\label{sec:subspacethermalization}

In this section, our goal is to identify if there is a way to predict quantum thermalization in almost all pure states of interest (according to some sufficiently strong, nontrivial notion of ``almost all''), in terms of quantities that do not \textit{a priori} have sufficient resolution to probe pure states in the Hilbert space. In particular, we are only allowed to start with quantities that may be built out of \textit{high-dimensional} projectors $\proj$ in the Hilbert space~\cite{vonNeumannThermalization}, where $\Tr[\proj]/\dim \mathcal{H} > \epsilon_{\text{hd}} = \Theta(1)$ [perhaps with logarithmic factors in $D$] as $\dim D \to \infty$. For example, in a many-body system, few-body observables can be constructed out of such high-dimensional projectors, which corresponds to a standard notion of experimentally accessible observables.

To address this problem in a form that is pruned down to its most essential details, we will begin with its mathematical core: how can one derive constraints on individual pure states from simpler constraints on projectors that contain a large number of linearly independent pure states? Physically, this can be thought of as addressing the problem of how to constrain microscopic thermalization in pure states from correlators of macroscopic projective observables at a single instant of time. Here, we will find that this can be done using certain geometric properties of high-dimensional subspaces in a Hilbert space captured by higher order correlators involving their operator products, described in Sec.~\ref{sec:correlatorgeometry}. Subsequently, we will derive rigorous constrains on thermalization in Sec.~\ref{sec:thermalsubspacetheorem} from these correlators, namely Theorem~\ref{thm:quantumthermalization_subspaces} and Proposition~\ref{prop:thermalizationsubspacesconverse}.  In Sec.~\ref{sec:dynamics}, we will make the identification of these geometric properties with OTOCs more explicit, although this should already be apparent to a large extent from the correlation functions used in this section.

\subsection{Operator products as probes of subspace geometry}
\label{sec:correlatorgeometry}

Our starting point is a geometric picture of correlation functions constructed out of operator product strings of two projectors, based on a Venn-diagram-like decomposition of a Hilbert space relative to any two given subspaces~\cite{Halmos2sub, Intro2sub}, in terms of which subspaces lie entirely inside or outside the given subspaces and which subspaces account for their relative orientation. Such a picture was used in \cite[Appendix C]{qergthesis} to derive a quantum speed limit on the growth of commutators; here, we will expand on this picture to set it up as a key ingredient for predicting quantum thermalization without reference to energy eigenstates, the latter being handled in Sec.~\ref{sec:thermalsubspacetheorem}.

For definiteness, consider two projectors $\proj_R$ and $\proj_{\rho}$ in a Hilbert space $\mathcal{H}$ of dimension $D$. Each projector represents a subspace $\mathcal{H}_R$ or $\mathcal{H}_{\rho}$ of $\mathcal{H}$ (on which the corresponding projector acts as identity, e.g. loosely\footnote{More precisely, $\proj_R \lvert \psi_R\rangle = \lvert \psi_R\rangle$ for all $\lvert \psi_R\rangle \in \mathcal{H}_R$ and $\proj_R \lvert \phi_{\overline{R}}\rangle = 0$ for all $\lvert \phi_{\overline{R}}\rangle \in \mathcal{H} \setminus \mathcal{H}_R$, and likewise for $B$.} $\proj_R \mathcal{H}_R = \mathcal{H}_R$), of respective dimensions $\dim \mathcal{H}_R = \Tr[\proj_R] = D_R$ and $\dim \mathcal{H}_{\rho} = \Tr[\proj_{\rho}] = D_{\rho}$. Without loss of (mathematical) generality, we take $D_{\rho} \leq D_R$. From an intuitive standpoint, we will consider $\proj_R$ (or more precisely, $\mathcal{H}_R$) to be an eigenspace of some local observable $\hat{A}$, while $\proj_\rho$ (i.e., $\mathcal{H}_{\rho}$) is a subspace within which we may choose a pure state $\lvert \psi\rangle_{\rho}$ of interest. 

In Euclidean geometry in $D$ dimensions, these subspaces are analogous to flat planar surfaces of respective dimensions $D_R$ and $D_{\rho}$. The relative orientation of these planar surfaces can be completely described by a set of $D_{\rho}$ ``principal angles of overlap'' $\theta_k \in [0,\pi/2]$ that capture overlaps between different directions (or ``principal axes of overlap'') within the two surfaces. To motivate this terminology, note that a unit sphere in surface $O$, when projected onto surface $\rho$, takes the form of an ellipsoid with principal axes $\lvert \widetilde{w}_k\rangle \in \mathcal{H}_{\rho}$ of respective lengths $\cos\theta_k$.

An analogous geometric decomposition is possible in our complex Hilbert space $\mathcal{H}$, in which case it is sometimes called Halmos' decomposition~\cite{Halmos2sub, Intro2sub}. For formal reasons (simplifying notation), it is convenient to introduce an auxiliary Hilbert space $\mathcal{H}_{\aux}$ (in which no state of interest lives) such that $\dim(\mathcal{H} \oplus \mathcal{H}_{\aux}) \geq D_R+D_{\rho}$. Then, 
\begin{lemma}[Halmos' decomposition of two subspaces~\cite{Halmos2sub, Intro2sub}]
\label{lem:Halmos}
There exist $(D_R - D_{\rho})$ vectors $\lvert a_j\rangle \in \mathcal{H}$, $D_{\rho}$ vectors $\lvert u_k\rangle \in \mathcal{H}$ and $D_{\rho}$ vectors $\lvert v_k\rangle \in \mathcal{H} \oplus \mathcal{H}_{\aux}$, all mutually orthonormal:
\begin{alignat}{3}
    &\langle a_j\vert a_m\rangle = \delta_{jm},&\ &\langle u_k\vert u_{\ell}\rangle = \delta_{k\ell},&\ &\langle v_k\vert v_{\ell}\rangle = \delta_{k\ell}, \nonumber \\
    &\langle a_j\vert u_{\ell}\rangle = 0,&\ &\langle u_k\vert v_{\ell}\rangle = 0,&\ &\langle v_k\vert a_m\rangle = 0,
\end{alignat}
with $D_{\rho}$ corresponding principal angles $\theta_k \in [0,\pi/2]$, such that
\begin{align}
    \proj_R &= \sum_{j = 1}^{D_R-D_{\rho}} \lvert a_j\rangle\langle a_j\rvert + \sum_{k = 1}^{D_{\rho}} \lvert u_k\rangle \langle u_k\rvert, \\
    \proj_{\rho} &= \sum_{k=1}^{D_{\rho}}\left(\cos\theta_k \lvert u_k\rangle + \sin\theta_k \lvert v_k\rangle\right) \left(\cos\theta_k \langle u_k\rvert + \sin\theta_k \langle v_k\rvert\right).
\end{align}    
\end{lemma}
\begin{proof}
    This is the main content of Refs.~\cite{Halmos2sub, Intro2sub}, with the adaptation to the form used here with an auxiliary Hilbert space $\mathcal{H}_{\aux}$ described in Ref.~\cite[Appendix C]{qergthesis}.
\end{proof}

Intuitively, this states that a ``preferred'' set of eigenvectors of $\proj_{\rho}$ (determined by $\proj_R$),
\begin{equation}
    \lvert w_k\rangle = \cos\theta_k \lvert u_k\rangle + \sin\theta_k \lvert v_k\rangle,
    \label{eq:principalaxesdef}
\end{equation}
are rotated by $\theta_k$ away from the corresponding eigenvectors $\lvert u_k\rangle$ of $\proj_R$, which form the respective ``principal axes'' in each subspace. Note that if $\theta_k \neq 0$, then it is necessary that $\lvert v_k\rangle \in \mathcal{H}$ as neither projector lives in $\mathcal{H}_{\aux}$ (in this context, we add that $\mathcal{H}_{\aux}$ is merely introduced to harbor unnecessary but convenient orthogonal vectors $\lvert v_k\rangle$ for directions in $\mathcal{H}_{\rho}$ that also lie in $\mathcal{H}_R$, i.e. with $\theta_k = 0$, so that the number of principal angles always corresponds to the dimension $D_{\rho}$ of $\mathcal{H}_{\rho}$). To directly draw a parallel to the Euclidean surfaces analogue above, the decomposition of $A$ represents some set of principal axes of a sphere in $\mathcal{H}_R$, a subset $\lbrace \lvert u_k\rangle \rbrace$ of which can be projected onto the $\proj_{\rho}$-eigenvectors $\lvert w_k\rangle$
in $\mathcal{H}_{\rho}$ to give an ellipsoid with principal axes of length $\cos\theta_k$:
\begin{equation}
    \lvert \widetilde{w}_k\rangle \equiv \lvert w_k\rangle \langle w_k\vert u_k\rangle = \cos\theta_k \lvert w_k\rangle.
\end{equation}

As $\proj^2 = \proj$ for a projector, the only nontrivially independent correlation functions are traces of powers of $\proj_R\proj_{\rho}$ i.e. must involve a product string of operators that necessarily alternates between $\proj_R$ and $\proj_{\rho}$. The correlator containing $(2n)$ operators can be written in terms of the principal angles (with an additional normalization by $D_{\rho}$ for later convenience) as:
\begin{equation}
    G_{R{\rho}}^{(2n)} \equiv\frac{1}{D_{\rho}}\Tr[(\proj_R \proj_{\rho})^n] = \frac{1}{D_{\rho}}\sum_k \left(\cos\theta_k\right)^{2n}.
\end{equation}
In this sense, these correlation functions capture the distribution of the principal angles $\theta_k$, via the $n$-th moments of $\cos\theta$. It is useful to emphasize that these remain $2$-point correlation functions, in the sense of fundamentally involving only $2$ operators $\proj_R$ and $\proj_{\rho}$, but are of order $(2n)$ in the sense of involving a product of $(2n)$ instances of these operators. The $2$-point nature of these functions is convenient for some tangential considerations in Appendix~\ref{app:autocorrelator_method}.

\subsection{Instantaneous quantum thermalization in pure states}
\label{sec:thermalsubspacetheorem}

Let us return to the intuitive interpretation of the subspace $\mathcal{H}_{\rho}$ as one that contains our pure states $\lvert \psi\rangle_{\rho}$ of interest. This subspace may be represented as an ensemble of states with a density operator:
\begin{equation}
    \hat{\rho} \equiv \frac{1}{D_{\rho}}\proj_{\rho}.
\end{equation}
The simplest correlation function is the second order correlation function:
\begin{equation}
    G_{R\rho}^{(2)} \equiv \frac{1}{D_{\rho}} \Tr[\proj_R\proj_{\rho}] = \Tr[\hat{\rho} \proj_R] = \frac{1}{D_{\rho}}\sum_k\cos^2\theta_k,
    \label{eq:G2anglesexpr}
\end{equation}
which can be interpreted as the mean value of the $\cos^2\theta$. In previous work~\cite{dynamicalqthermalization}, it has been possible to constrain the thermalization of such correlation functions \textit{at a given instant of time} with accessible resources, with the caveat that these always correspond to a statistical average over an ensemble of states; no accessible way of constraining pure states appeared to be possible within that framework. Separately, time-averages could be constrained probabilistically for individual states, and it appeared that one needed to have either the state average or the time average for accessible results (see Sec.~\ref{sec:intro} and Sec.~\ref{sec:synopsis} for details).

Here, in light of the geometric picture of correlation functions in terms of principal angles, we note that the 4th order correlation function
\begin{equation}
    G_{R\rho}^{(4)} \equiv \frac{1}{D_{\rho}}\Tr[\proj_R \proj_{\rho}\proj_R \proj_{\rho}] = \frac{1}{D_{\rho}} \sum_k \cos^4\theta_k
    \label{eq:G4anglesexpr}
\end{equation}
contains information about the distribution of the $\theta_k$, and therefore can constrain the expectation value of $\proj_R$ in \textit{individual} pure states at a given instant of time, thereby eliminating both statistical averages. In particular, we can form the variance of principal angles,
\begin{equation}
    \sigma_{R\rho}^2 \equiv G_{R\rho}^{(4)} - \left[G_{R\rho}^{(2)}\right]^2  = \frac{1}{D_\rho}\sum_k\left[\cos^2\theta_k - G_{R\rho}^{(2)}\right]^2.
    \label{eq:anglevariance_def}
\end{equation}
Intuitively, if this variance is small, then almost all of the $\cos^2\theta_k$ are close to $G_{R\rho}^{(2)}$, which means that any linear combination of the corresponding vectors $\lvert w_k\rangle = \cos\theta_k\lvert u_k\rangle + \sin\theta_k \lvert v_k\rangle$ in $\mathcal{H}_{\rho}$ has a $\proj_R$-expectation value close to $G_{R\rho}^{(2)}$. In other words, $\proj_R$ thermalizes in almost all pure states in $\mathcal{H}_{\rho}$. Geometrically, it is then convenient to call $\proj_{\rho}$ ``(nearly) aligned'' with $\proj_R$, as all the principal angles are (nearly) equal\footnote{This is not to be confused with $\proj_\rho$ being nearly contained in $\proj_R$, which would require that all $\theta_k \approx 0$.}.

Formalizing the above argument gives us our first main result:
\begin{theorem}[Quantum thermalization in pure states in nearly aligned subspaces]
\label{thm:quantumthermalization_subspaces}
Let $\proj_R$, $\proj_{\rho}$ be projectors in the Hilbert space $\mathcal{H}$ representing subspaces $\mathcal{H}_R$, $\mathcal{H}_{\rho}$, with $\sigma_{R\rho}^2$ defined as in Eq.~\eqref{eq:anglevariance_def}. Then, for every $\lambda \geq 0$, the following statements hold:
\begin{enumerate}
    \item There exists a subspace $\mathcal{H}_{\therm}(\lambda) \subseteq \mathcal{H}_{\rho}$ with a dimension of at least
    \begin{equation}
        D_{\therm}(\lambda) \equiv \dim \mathcal{H}_{\therm}(\lambda) \geq D_{\rho}\left(1-\frac{\sigma_{R\rho}^2}{\lambda^2}\right),
        \label{eq:dimrestriction1}
    \end{equation}
    such that every pure state $\lvert \psi\rangle \in \mathcal{H}_{\therm}(\lambda)$ attains an expectation value of $\proj_R$ that is close to $G_{R\rho}^{(2)}$ to a resolution $\lambda$:
    \begin{equation}
        \left\lvert\frac{\langle \psi\vert \proj_R\lvert \psi\rangle}{\langle \psi\vert \psi\rangle} - G_{R\rho}^{(2)} \right\rvert \leq \lambda.
        \label{eq:subspacethermalizationcondition}
    \end{equation}
    \item In any orthonormal basis $\mathcal{B}_{\rho} = \lbrace \lvert b_k\rangle_{\rho}\rbrace \subset \mathcal{H}_{\rho}$ for $\mathcal{H}_{\rho}$, the fraction of basis states in which the expectation value of $\proj_R$ deviates from $G_{R\rho}^{(2)}$ by more than $\lambda$ [where $\ifc(\mathcal{C}) = 1$ if the conditional statement $\mathcal{C}$ is true and $\ifc(\mathcal{C}) = 0$ if $\mathcal{C}$ is false],
    \begin{equation}
        f_{\lambda} \equiv \frac{1}{D_{\rho}} \sum_{\lvert b_k\rangle \in \mathcal{B}_{\rho}} \ifc\left(\left\lvert\langle b_k\vert \proj_R\lvert b_k\rangle - G_{R\rho}^{(2)} \right\rvert > \lambda\right),
        \label{eq:nonthermalbasisstates_fraction}
    \end{equation}
    is constrained by $\sigma_{R\rho}^2$ as follows:
    \begin{equation}
        f_{\lambda} \leq \frac{3}{\lambda}\left(\frac{\sigma_{R\rho}^2}{4}\right)^{1/3}.
        \label{eq:nonthermalbasisstatesconstraint}
    \end{equation}
\end{enumerate}
\end{theorem}
\begin{proof}
    This follows from a rigorous version of the argument after Eq.~\eqref{eq:anglevariance_def}, see Appendix~\ref{proof:quantumthermalization_subspaces}.
\end{proof}

While these are general bounds, the physical content of these bounds is that $D_{\rho}(\lambda)$ is nearly as large as $D_{\rho}$ and $f_{\lambda}$ is small when $\sigma_{R\rho}^2$ is small. This allows one to predict quantum thermalization to accuracy $\lambda \gg \sigma_{R\rho}^{2/3}$ for an overwhelming fraction of pure states in the subspace $\mathcal{H}_{\rho}$, if one is given the (e.g. analytically computable and experimentally measurable) $2$-nd and $4$-th order correlation functions of the high-dimensional projectors $\proj_R$ and $\proj_{\rho}$.

We note that higher-order correlation functions ($G_{R\rho}^{(2n)}$ for $2n>4$) are usually not \textit{necessary} to constrain thermalization. However, they may allow for more refined bounds e.g. on $D_{\rho}(\lambda)$ or $f_{\lambda}$ via the distributional information captured in the higher moments of the $\cos^2\theta_k$, which may be relevant especially to improve resolution in finite-size systems.

Another important question concerns the reverse implication: does pure state thermalization in $\proj_{\rho}$ imply a low variance $\sigma_{R\rho}^2$? Only were this the case can we regard the correlation functions that occur in $\sigma_{R\rho}^2$ as providing a robust signature of pure state thermalization or its absence. As it turns out, the answer to this question is a straightforward but qualified yes:
\begin{proposition}[Pure state thermalization in all bases implies a near-alignment of subspaces]
\label{prop:thermalizationsubspacesconverse}
Given two subspaces $\mathcal{H}_R$, $\mathcal{H}_{\rho}$ with the associated projectors $\proj_R$, $\proj_{\rho}$ and a constant $\lambda \geq 0$ with $\lambda < 1$, if for \textbf{every} orthonormal basis $\mathcal{B}_{\rho} = \lbrace \lvert b_k\rangle_{\rho}\rbrace \subset \mathcal{H}_{\rho}$ for $\mathcal{H}_{\rho}$, the fraction of basis states in which the expectation value of $\proj_R$ deviates from $G_{R\rho}^{(2)}$ by more than $\lambda$ is bounded by $f_{\lambda}^{\max}$, i.e.
\begin{equation}
    f_{\lambda} \equiv \frac{1}{D_{\rho}} \sum_{\lvert b_k\rangle \in \mathcal{B}_{\rho}} \ifc\left(\left\lvert\langle b_k\vert \proj_R\lvert b_k\rangle - G_{R\rho}^{(2)} \right\rvert > \lambda\right) \leq f_{\lambda}^{\max},
    \label{eq:basisfractionbound_converse}
\end{equation}
then the variance $\sigma_{R\rho}^2$, given in terms of correlators by Eq.~\eqref{eq:anglevariance_def}, is bounded by:
\begin{equation}
    \sigma_{R\rho}^2 \leq \lambda^2+(1-\lambda^2)f_{\lambda}^{\max}.
    \label{eq:conversevariancebound}
\end{equation}
\end{proposition}
\begin{proof}
    In Eq.~\eqref{eq:basisfractionbound_converse}, choose $\lvert b_k\rangle = \lvert w_k\rangle$ (the principal axes in $\proj_{\rho}$), in which case a fraction $(1-f_{\lambda})$ of terms in Eq.~\eqref{eq:anglevariance_def} contributes at most $\lambda^2$ and the remaining fraction $f_{\lambda}$ of terms contributes at most $1$, giving Eq.~\eqref{eq:conversevariancebound} on using $f_{\lambda} \leq f_{\lambda}^{\max}$ as $(1-\lambda^2) > 0$.
\end{proof}

The combination of Theorem~\ref{thm:quantumthermalization_subspaces} and Proposition~\ref{prop:thermalizationsubspacesconverse} shows that the smallness of $\sigma_{R\rho}^2$ is a faithful indicator of thermalization in pure states, specifically almost all pure states in \textit{every} orthonormal basis in $\mathcal{H}_{\rho}$, both implying and being implied by the latter (with appropriate quantitative restrictions). But the qualifier indicated above is as follows: if we relax the requirement of thermalization to apply to only \textit{some} (not all) orthonormal bases in $\proj_{\rho}$, it is no longer necessary that $\sigma_{R\rho}^2$ is small. Indeed, by standard typicality arguments~\cite{vonNeumannThermalization, tumulka_CT, CanonicalTypicalityPSW, NormalTypicality}, almost all pure states (in $\mathcal{H}_{\rho}$, where any finite set of states has measure zero) have a thermal expectation value of any given observable $\proj_R$ irrespective of any details of the system (if $D_{\rho}$ is sufficiently large). What $\sigma_{R\rho}^2$ brings to the table is a stronger constraint on almost all pure states in every (finite) orthonormal basis in $\mathcal{H}_{\rho}$, where resolving a finite set of states is beyond the reach of typicality arguments. Such a stronger constraint is useful for providing guarantees of thermalization when the initial states of interest belong to atypical bases that are easier to access in practice, such as the computational basis of a many-qubit system\footnote{At the same time, one can also wonder in parallel if there is an intermediate degree of rigorous predictability between the previously known universal typicality results for ``almost all'' orthonormal bases in a subspace, and strong thermalization results for almost all states in \textit{every} orthonormal basis in a space as obtained here. Such an intermediate approach may be helpful in systems where certain bases of interest thermalize, but not every conceivable basis.}.

\section{Dynamics: Subspace alignment in many-body systems}
\label{sec:dynamics}

While Sec.~\ref{sec:subspacethermalization} contains our main technical results, we will now undertake the task of completing the translation of these results to the language of few-body correlation functions in many-body systems. This will bridge these technical statements to the physical picture provided in the synopsis in Sec.~\ref{sec:synopsis}.

Sec.~\ref{sec:typicality} will consider the typical behavior of these correlation functions, obtaining constrains on the sizes of the projectors for the results of Sec.~\ref{sec:subspacethermalization} to be useful in typical cases. Sec.~\ref{sec:restoredynamics} will consider the time evolution of these correlation functions, and point to certain difficulties in generalizing our instantaneous results to carry implications outside the times of observation for Hamiltonian systems. This makes the restriction to Hamiltonian dynamics unnecessary for our present purposes, and allows us to consider general time-dependent unitary dynamics. Sec.~\ref{sec:manybodysystems} will specifically formulate Theorem~\ref{thm:quantumthermalization_subspaces} and Proposition~\ref{prop:thermalizationsubspacesconverse} in a form compatible with the many-body setting of Sec.~\ref{sec:synopsis}, namely Corollary~\ref{cor:manybodypurestatethermalization}.

\subsection{Typicality of instantaneous subspace alignment}
\label{sec:typicality}

Theorem~\ref{thm:quantumthermalization_subspaces} is a general result that constrains instantaneous thermalization in pure states given certain subspace alignment properties implied by correlation functions, with Proposition~\ref{prop:thermalizationsubspacesconverse} handling its converse. It is also instructive to consider when such an alignment may take place in physical systems. Our key observation will be that this \textit{typically} requires $D_{\rho} \ll D_{R}$, i.e., the space of states represented by $\proj_{\rho}$ is much smaller than a degenerate eigenspace of the observable $\proj_R$.

Generic behavior is obtained when $\proj_R$ and $\proj_{\rho}$ are randomly chosen projectors~\cite{vonNeumannThermalization, NormalTypicality} acting on $\mathcal{H}$. This can be effectively modeled by taking $\proj_{R} = \hat{V}_R \proj_{R1} \hat{V}_R^\dagger$ and $\proj_{\rho} = \hat{V}_{\rho} \proj_{\rho1} \hat{V}_{\rho}^\dagger$, where $\proj_{R1}$ and $\proj_{\rho 1}$ are some fixed projectors of dimensions $D_R$ and $D_{\rho}$, while $\hat{V}_R$ and $\hat{V}_{\rho}$ are chosen to be random unitaries, specifically according to the Haar measure~\cite{HaarBook} on the unitary group $\unitarygroup(D)$, corresponding to the Circular Unitary Ensemble (CUE)~\cite{Mehta, Haake} acting on the $D$-dimensional Hilbert space $\mathcal{H}$. However, in all correlation functions of interest, only the combination $\hat{V} = \hat{V}_R^\dagger \hat{V}_{\rho}$ occurs:
\begin{equation}
    G_{R\rho}^{(2n)} = \frac{1}{D_{\rho}}\Tr\left(\left[\proj_{R1}\hat{V} \proj_{\rho 1} \hat{V}^\dagger\right]^n\right);
    \label{eq:correlator_randomization_def}
\end{equation}
by standard properties of the Haar measure (invariance under multiplication by a unitary), it is sufficient to consider $\hat{V}$ itself to be randomly chosen according to the Haar measure of CUE.

Averaging $\hat{V}$ over CUE, we are interested in the mean values of correlators,
\begin{equation}
    \left\langle G_{R\rho}^{(2n)} \right\rangle_{\hat{V},\text{CUE}} = \int_{\unitarygroup(D)}\diff \hat{V}\ \frac{1}{D_{\rho}}\Tr\left(\left[\proj_{R 1} \hat{V} \proj_{\rho 1} \hat{V}^\dagger\right]^n\right),
    \label{eq:Haar_mean_G}
\end{equation}
as well as the variances $\left\langle [G_{R\rho}^{(2n)}]^2 \right\rangle_{\hat{V},\text{CUE}} - \left\langle G_{R\rho}^{(2n)} \right\rangle_{\hat{V},\text{CUE}}^2$, which determine if a typical member of CUE has $G_{R\rho}^{(2n)}$ close to the Haar average. However, most of these are tedious to compute: the mean values are (conveniently) straightforward for $2n \leq 4$, while the variances are (inconveniently) manageable only for $2n = 2$; higher orders require a careful consideration of multiple terms associated with permutations of degree $n$ as per Weingarten calculus~\cite{Weingarten1, Weingarten2} (see e.g. \cite{MertWg} for relevant symbolic manipulation resources for these terms). Recall that $\sigma_{R\rho}^2$ involves both $2n = 2$ and $2n=4$. Most crucially for our purposes, the $2n = 4$ variance appears to be quite formidable, which we will circumvent using stronger concentration inequalities.

Evaluating the Haar averages for the easily tractable cases is straightforward using standard Haar integral relations~\cite{Weingarten1, Weingarten2, ChaosComplexityDesign, ChaosComplexityRMT, CotlerHunterJones2}, as described in Appendix~\ref{app:Weingarten}, We obtain, with $D_S = D/D_R$ and $D_{\sigma} = D/D_{\rho}$:
\begin{align}
    \left\langle G_{R\rho}^{(2)} \right\rangle_{\hat{V},\text{CUE}} &= \frac{1}{D_S}, \\
    \left\langle [G_{R\rho}^{(2)}]^2 \right\rangle_{\hat{V},\text{CUE}} - \left\langle G_{R\rho}^{(2)} \right\rangle_{\hat{V},\text{CUE}}^2 &= \frac{(D_{\sigma}-1)(D_S-1)}{(D^2-1)D_S^2} \\
    \left\langle G_{R\rho}^{(4)} \right\rangle_{\hat{V},\text{CUE}} &= \frac{D^2}{D^2-1}\left\lbrace \frac{1}{D_S^2} + \frac{1}{D_S D_{\sigma}} - \frac{1}{D_S^2 D_{\sigma}} - \frac{1}{D D_S D_{\sigma}}\right\rbrace.
\end{align}
We observe that the variance of $G_{R\rho}^{(2)}$ becomes negligible compared to $\left\langle G_{R\rho}^{(2)} \right\rangle_{\hat{V},\text{CUE}}$ in the large-$D$ limit with fixed $D_S$, $D_{\sigma}$. This limit is justified in Sec.~\ref{sec:manybodysystems} (and anticipated in Sec.~\ref{sec:intro} and Sec.~\ref{sec:synopsis}), where we want $D_S$ and $D_{\sigma}$ to capture the dimensionality of few-body systems in a thermodynamically large system of dimension $D$.

Now, we turn to the size of fluctuations for $2n=4$. An easier and more ``natural'' way of estimating the size of deviations, without evaluating the Haar averages of moments, is to use concentration of measure results for the unitary group~\cite{HaarBook} as is standard in typicality arguments~\cite{vonNeumannThermalization, NormalTypicality} (which is similar to the use of L\'{e}vy's lemma for Haar random vectors on a sphere~\cite{HaarBook} in state typicality arguments~\cite{CanonicalTypicalityPSW, CanonicalTypicalityPSW2}). Concentration of measure loosely states that the size of deviations from the Haar average for any sufficiently ``coarse-grained'' $O(1)$ quantity in a typical member of CUE is vanishingly small as $D\to\infty$. The precise size of these deviations for the above correlation functions is rigorously estimated in Appendix~\ref{app:concentration}; we note that other concentration bounds for closely related correlation functions have been derived in Refs.~\cite{SAZ, DKM}, which our concentration results are similar to with some technical differences.

Qualitatively, concentration of measure implies that with a probability of almost $1$, the size of fluctuations in $G_{R\rho}^{(2n)}$ is of the order (see Appendix~\ref{app:concentration} for a rigorous quantitative statement)
\begin{equation}
    \left\lvert G_{R \rho}^{(2n)} - \left\langle G_{R\rho}^{(2n)} \right\rangle_{\hat{V},\text{CUE}}\right\rvert \lesssim \frac{2n \sqrt{2D_{\sigma}}}{D} \kappa,
    \label{eq:generalcorrelators_concentration}
\end{equation}
for some large constant $\kappa \gg 1$.
These fluctuations may therefore be neglected as $D\to\infty$ e.g. for fixed $D_{\sigma}$. For near-term practical computations in small systems, we note that these fluctuations are negligible even if $D \sim D_{\sigma}$ i.e. the ``bath'' subsystem need not be large compared to the core.

Given the smallness of fluctuations, we have the following leading order behavior of the relevant correlators for ``typical'' choices of $\proj_R$ and $\proj_{\rho}$ (taking $D\to\infty$ while keeping $D_S$, $D_{\sigma}$ finite):
\begin{align}
    G_{R\rho}^{(2)} &\simeq \frac{1}{D_S}, \label{eq:G2Haar} \\
    G_{R\rho}^{(4)} &\simeq \frac{1}{D_S^2}+ \frac{1}{D_{\sigma} D_S} - \frac{1}{D_{\sigma} D_S^2}. \label{eq:G4Haar}
\end{align}
Therefore, one should typically see that the variance of principal angles saturates to (here $D_{\sigma} \geq D_S$ by assumption, as $D_R \geq D_{\rho}$):
\begin{equation}
    \sigma_{R\rho}^2 \simeq \frac{1}{D_\sigma D_S}\left(1-\frac{1}{D_S}\right).
    \label{eq:sigmaHaar}
\end{equation}
For this variance to be small, i.e. $\sigma_{R\rho}^2/[G_{R\rho}^{(2)}]^2 \ll 1$, we require that $D_\sigma \gg D_S$. This is a very important point in this context: we note that even $\textit{typically}$, subspaces with $D_{\sigma} \sim D_S$ have large variances and their principal axes fail to thermalize. In fact, choosing $\lambda = \epsilon D_S^{-1}$ (with $\epsilon> 0$ expected to be small but finite), we obtain from Theorem~\ref{thm:quantumthermalization_subspaces}:
\begin{align}
    D_{\therm}(\epsilon D_S^{-1}) &\gtrsim D_{\rho}\left(1-\frac{D_S-1}{\epsilon^2 D_{\sigma}}\right), \nonumber \\
    f_{\epsilon D_S^{-1}} &\lesssim 4\left[\frac{D_S^2-D_S}{4\epsilon^3 D_{\sigma}}\right]^{1/3},
\end{align}
which requires $D_{\sigma} \gg D_S^2/\epsilon^3$ (or $D_{\rho} \ll \epsilon^3(D_R/D) D_R$) to give good bounds. As $\epsilon$ is a finite constant and $D_S$ is also assumed to be finite, this still allows $\proj_{\rho}$ to be a high-dimensional projector in the sense that $D_{\rho} > \epsilon_{\text{hd}} D$ for some sufficiently small $\epsilon_{\text{hd}} = \Theta(1)$. Moreover, from Proposition~\ref{prop:thermalizationsubspacesconverse}, we also have
\begin{equation}
    f_{\epsilon D_S^{-1}}^{\max} \geq \frac{D_S^2 \sigma_{R\rho}^2 -\epsilon^2}{D_S^2-\epsilon^2} \geq \sigma_{R\rho}^2 - \frac{\epsilon^2}{D_S^2} \simeq \frac{1}{D_S D_{\sigma}}\left(1-\frac{1+\epsilon^2 D_{\sigma}}{D_S}\right), 
\end{equation}
which shows a clear nonthermal fraction of basis states in $\proj_{\rho}$ for $\epsilon \ll 1$ if $1/(D_S D_{\sigma})$ is appreciable, e.g. for $D_{\sigma} \sim D_S$. Therefore, large subspaces typically almost never thermalize in all their constituent pure states, and there necessarily exist noticeable nonthermal bases. It becomes \textit{necessary} to choose smaller subspaces $\mathcal{H}_{\rho}$ with $D_\rho \ll D_R$ to guarantee thermalization to within a finite resolution for almost all states (in every orthonormal basis) within the subspace.

Some intuition for this $D_{\rho}$-dependence may be acquired as follows: it is guaranteed that any set of eigenstates of $\proj_R$ (which is highly degenerate) do not thermalize, having expectation values equal to the eigenvalues $\lbrace 0,1\rbrace$. These nonthermal eigenstates are visible to any typical subspace $\proj_{\rho}$ of sufficiently large dimension (e.g., trivially, if $D_{\rho} = D$, $\mathcal{H}_{\rho} = \mathcal{H}$ contains these nonthermal states). Consequently, a typical large $\proj_{\rho}$ has principal axes that are sufficiently sensitive to these nonthermal eigenstates to not allow thermalization within $\mathcal{H}_{\rho}$. This can only be mitigated by making $D/D_{\rho}$ large, whose implications we will examine more closely in Sec.~\ref{sec:manybodysystems}.

Having assessed the feasibility of finding a sufficiently small variance to make Theorem~\ref{thm:quantumthermalization_subspaces} useful, i.e., finding that small variance is typical provided $\proj_{\rho}$ though large is parametrically smaller in dimension that $\proj_R$, we briefly comment on what appears to be a counter-intuitive (or at least, to us, not immediately intuitive) geometry of this situation. It is well known that two random pure states are typically nearly orthogonal, i.e. $\lvert \langle \psi\vert \phi\rangle\rvert \sim O(1/\sqrt{D})$ with leading order $O(1/\sqrt{D})$ fluctuations, if $\lvert \psi\rangle, \lvert \phi\rangle \in \mathcal{H}$ are independently chosen according to the Haar measure of states (the measure on $\mathcal{H}$ that is invariant under unitary transformations). However, the above estimates show that the behavior of two random high-dimensional surfaces is quite different: most of the $D_{\rho}$ principal angles cluster around a finite value $\theta = \arccos D_S^{-1}$ determined by the larger subspace, with fluctuations of size $\sim \sqrt{D_S/D_{\sigma}}$, which becomes negligible as $D_{\sigma} \gg D_S$. What results is a phenomenon of \textit{uniform alignment} of most of the corresponding principal axes of the two subspaces with nearly the same angle $\theta$, rather than strong random fluctuations or near-orthonormality as with pure states (note, however, that this random-orthonormal pure state behavior is recovered above for $D_R = D_{\rho} = 1$). However, uniform alignment is only seen for $D_{\sigma} \gg D_S$, and completely disappears if $D_{\sigma} \sim D_S$. From our present standpoint, this uniform alignment effect is the geometric phenomenon that is at the root of pure state quantum thermalization, its typicality directly translating into the ubiquity of the latter in sufficiently restricted subspaces.

\subsection{Dynamics of subspace alignment}
\label{sec:restoredynamics}

We had deliberately ignored time evolution in Sec.~\ref{sec:subspacethermalization} section to focus on the core mathematical content of our results, which expresses pure state thermalization in terms of a simple geometric picture of alignment of subspaces. Here, we will start addressing the direct implications of Sec.~\ref{sec:subspacethermalization} for thermalization dynamics.

For this purpose, let us now formally introduce dynamics into our system (in the Schr\"{o}dinger picture~\cite{ShankarQM}): we will take states to have a time dependence $t$ generated by some (in general, time-dependent) unitary operator $\uh(t)$:
\begin{equation}
    \lvert \psi(t)\rangle = \uh(t)\lvert \psi\rangle.
    \label{eq:timeevolution}
\end{equation}
Note that we do not require that this unitary is autonomous, i.e.,  generated by a Hamiltonian; as our thermalization results apply at each instant of time, non-autonomous unitary dynamics is in its domain of applicability just as well as autonomous unitary dynamics.

As we have previously associated $\mathcal{H}_{\rho}$ with our subspace of states, we will require all states in this subspace to evolve according to Eq.~\eqref{eq:timeevolution}. Consequently, other quantities associated with this subspace will also acquire some time-dependence, e.g., $\proj_{\rho}(t) = \uh(t)\proj_{\rho}\uh^\dagger(t)$. As $\proj_R$ was associated with observables, this subspace has no dynamics of its own; however, the principal axes and angles in Lemma~\ref{lem:Halmos} depend on $\proj_{\rho}(t)$ and will therefore become time-dependent, e.g. $\lvert u_k(t)\rangle$, $\lvert v_k(t)\rangle$, $\lvert w_k(t)\rangle$, and $\theta_k(t)$. Where specification is necessary, by $\mathcal{H}_{\rho}$ and $\mathcal{H}_{R}$ we will mean the subspaces $\mathcal{H}_{\rho}(t)$ and $\mathcal{H}_{R}(t)$ corresponding to these projectors at $t=0$.

All our considerations in Sec.~\ref{sec:subspacethermalization} now carry over to the corresponding time-dependent correlation functions:
\begin{equation}
    G_{R\rho}^{(2n)}(t) \equiv \frac{1}{D_{\rho}}\Tr\left(\left[\proj_R \proj_{\rho}(t)\right]^n\right).
    \label{eq:timedependentcorrelatorsdef}
\end{equation}
The $2n = 2$ case is the standard $2$-point correlation function of $\proj_R$ and $\proj_{\rho}$. The $2n = 4$ case corresponds to a slightly more nontrivial object: $G_{R\rho}^{(4)}(t)$ may be recognized as the out-of-time-ordered correlator (OTOC) of $\proj_R$ and $\proj_{\rho}(t)$, whose role was qualitatively delineated in Sec.~\ref{sec:synopsis}.

The behavior of these correlation functions can be characterized, for example, for Haar random unitaries. To be specific, let us pick some time $t=t_0$ at which $\uh(t_0)$ is sampled from the circular unitary ensemble (CUE) of random unitary matrices~\cite{Haake, Mehta}. This is a reasonable model for generic quantum dynamics: even unitaries with a much simpler (than Haar random) structure tend to show leading-order correlations (e.g. for operators of high rank) comparable to Haar random unitaries (motivating several facets of the field of ``quantum chaos''~\cite{Haake, ChaosComplexityRMT} that are based on comparisons of different quantities in a system with ``sufficiently complex'' dynamics to Haar random behavior).

Typical systems may evolve to show leading-order Haar behavior in a specific subset of observables of interest after some ``Haar typicality time''\footnote{Such a time scale goes by different names; one common choice is ``scrambling time''~\cite{ChaosComplexityRMT}, which is used in several different (but somewhat related) ways in other contexts e.g. for entanglement generation~\cite{HaydenPreskill, LashkariFastScrambling, dynamicalqspeedlimit, dynamicalqfastscrambling}. We will avoid using such an umbrella term with a view to distinguish various timescales with more specific names.} $t_s$ that potentially depends on the chosen set of observables. In many cases, the relevant Haar ensemble is CUE, in which case by Eqs.~\eqref{eq:G2Haar}, \eqref{eq:G4Haar} and \eqref{eq:sigmaHaar} one can expect (for $t$ implicitly smaller than the scale of quantum recurrences~\cite{QuantumRecurrences, BrownSusskind2}):
\begin{align}
    G_{R\rho}^{(2)}(t > t_s) &\simeq \frac{1}{D_S}, \\
    G_{R\rho}^{(4)}(t > t_s) &\simeq \frac{1}{D_S^2}+ \frac{1}{D_{\sigma} D_S} - \frac{1}{D_{\sigma} D_S^2}, \\
    \sigma_{R\rho}^2(t > t_s) &\simeq \frac{1}{D_\sigma D_S}\left(1-\frac{1}{D_S}\right).
\end{align}
If it is known that such ``typical'' behavior lasts for almost all times, then thermalization in pure states follows at almost all times as in Sec.~\ref{sec:typicality}. Especially in experimentally relevant situations, however, it may only be possible to establish such behavior for some finite set of times. It is worth exploring to what extent one can extrapolate from a finite set of times to almost all times.

Such an extrapolation is possible for $2$nd order functions in autonomous (Hamiltonian) systems as follows. Let $\mathcal{T}_{R\rho}^{(2)}(\epsilon)$ be the time interval in which $G_{R\rho}^{(2)}(t)$ is close to its Haar-averaged value within some resolution $\epsilon > 0$:
\begin{equation}
    \mathcal{T}_{R\rho}^{(2)}(\epsilon) \equiv \left\lbrace t: \left\lvert G_{R\rho}^{(2)}(t) - \frac{1}{D_S}\right\rvert \leq \epsilon\right\rbrace.
\end{equation}
If the dynamics is autonomous (generated by a Hamiltonian), then the results of Ref.~\cite{dynamicalqthermalization} allow us to provide powerful constraints on $\mathcal{T}_{R\rho}^{(2)}(\epsilon)$, given only the dynamics of the autocorrelator $G_{RR}^{(2)}(t)$ over a finite time interval. For example, given any \textit{symmetric} and contiguous time interval $\mathcal{T}_{\obs} = [-T_{\obs},T_{\obs}]$ of finite duration, if it is possible to determine that the autocorrelator is nonthermal only for a small fraction $\kappa_{RR} > 0$ of this interval,
\begin{equation}
    1-\frac{\left\lvert \mathcal{T}_{RR}^{(2)}(\varepsilon) \cap \mathcal{T}_{\obs}\right\rvert}{\left\lvert \mathcal{T}_{\obs}\right\rvert} < \kappa_{RR},
    \label{eq:autocorrelatortimeintervalconstraint}
\end{equation}
then it follows that for any contiguous time interval $\mathcal{T} = [t_0, t_0+T]$ of finite length $\lvert \mathcal{T}\rvert = T$, and for any chosen constant $\xi > 0$,
\begin{equation}
    1-\frac{\left\lvert \mathcal{T}_{R\rho}^{(2)}(\lambda) \cap \mathcal{T}\right\rvert}{\left\lvert \mathcal{T}\right\rvert} \leq \frac{D_{\sigma}}{\lambda^2D_S}\left(\frac{\xi}{|\sin\xi|}\sqrt{\varepsilon^2 + 2\kappa_{RR}} + \frac{T_{\obs}}{\xi T}\right).
    \label{eq:secondordertimeintervalconstraint}
\end{equation}
This is derived in Appendix~\ref{app:secondordercorrelators} [footnote~\ref{footnote:timeintervalconstraintderivation}] as a consequence of (a variant of) Eq.~\eqref{eq:synopsis_autocorrelator_to_correlator} in Sec.~\ref{sec:synopsis} (and is in principle a slightly weaker inequality than the optimal inequality that would follow from the results therein). Thus, for a given small resolution $\lambda$ with autocorrelator observations that satisfy $\epsilon^2, \kappa_{RR} \ll \lambda^4$, the nonthermal set of times $\mathcal{T}_{R\rho}^{(2)}(\lambda)$ of the correlator occupies a negligibly small fraction of any interval $\mathcal{T}$ of sufficiently large duration $T \gg T_{\obs}$, as long as $D_S$ and $D_\sigma$ are finite. In this way, the thermalization of the autocorrelator over a certain time interval implies the thermalization of all other second order two point functions over almost all times in long intervals.

When it comes to fourth order functions, we do not appear to be so fortunate. For example, it is argued in Appendix~\ref{app:fourthordercorrelators} that the method of Ref.~\cite{dynamicalqthermalization}, suitably generalized as in Appendix~\ref{app:autocorrelators_ips}, does not yield a useful bound for any fourth-order correlator from observations over a finite time interval with finite resolution. From a knowledge of the second order correlators alone, the best we can seem to do is to obtain the following straightforward inequalities:
\begin{equation}
    G_{R\rho}^{(2)}(t) \geq G_{R\rho}^{(4)}(t) \geq \left[G_{R\rho}^{(2)}(t)\right]^2.
    \label{eq:G2G4bound}
\end{equation}
Here, the first inequality follows from  e.g. $\cos^4\theta_k \leq \cos^2\theta_k$ in Eqs.~\eqref{eq:G2anglesexpr} and \eqref{eq:G4anglesexpr}, while the second follows from $\sigma_{R\rho}^2(t) \geq 0$ in Eq.~\eqref{eq:anglevariance_def}. Which of these two extremes $G_{R\rho}^{(4)}(t)$ is closer to at any given time determines the closeness of the subspace $\mathcal{H}_{\rho}(t)$ to being strongly thermal for $\proj_R$. When $G_{R\rho}^{(4)}(t)$ nearly saturates the second inequality, i.e. virtually attains its minimum value $[G_{R\rho}^{(2)}(t)]^2$, then we have thermalization in almost all pure states in $\mathcal{H}_{\rho}(t)$; however, when it is comparable to $G_{R\rho}^{(2)}(t)$, the principal axes at $t$ strongly fail to thermalize, e.g. may have extreme principal angles $\theta_k \in \lbrace 0, \pi/2\rbrace$ if the first inequality is saturated. Consequently, even in $\mathcal{T}_{R\rho}^{(2)}(\lambda)$, we can only constrain:
\begin{equation}
     \frac{1}{D_S} + \lambda \geq \sigma_{R\rho}^2\left(t \in \mathcal{T}_{R\rho}^{(2)}(\lambda)\right) \geq 0,
\end{equation}
which is not sufficient to conclude quantum thermalization in pure states.

To see why such forecasts from the behavior of a given set of correlators $G_{R\rho}^{(2n)}(t)$ alone may be (at least) difficult, let us return to the typicality results of Sec.~\ref{sec:typicality}. In typical systems, second order autocorrelators $G_{RR}(t) \sim 1/D_S$ saturate to the same value required to establish the thermalization $G_{R\rho}^{(2)}(t) \sim1/D_S$ of other second order correlators. Therefore, if we choose $\mathcal{H}_{\rho} \subset \mathcal{H}_R$, there is no obstacle to the thermalization of $\proj_{\rho}(t)$.

This is not the case for fourth-order correlators. In particular, let $\mathcal{H}_A$ be a subspace such that $\dim\mathcal{H}_A = \dim \mathcal{H}_R$. Then, even typically,
\begin{align}
    \left.G_{RA}^{(4)}(t)\right\rvert_{\text{typical}} &\simeq \frac{2}{D_S^2}
-\frac{1}{D_S^3} \nonumber \\
\implies \left.\sigma_{RA}^2(t)\right\rvert_{\text{typical}} &\simeq \frac{1}{D_S^2}
\end{align}
which, by Proposition~\ref{prop:thermalizationsubspacesconverse} implies that in the basis of principal axes within $\proj_{A}(t)$, the fraction $f_{\lambda}^{\max}$ that fails to thermalize to resolution $\lambda$ is at least, for $\lambda \leq \epsilon/D_S$ for $\epsilon < 1$,
\begin{equation}
    f_{\lambda}^{\max}(t) \geq \frac{\sigma_{RA}^2(t)-\lambda^2}{1-\lambda^2} \geq \frac{1-\epsilon^2}{D_S^2-1},
\end{equation}
which is a $\Theta(1)$ fraction of states. Therefore, the principal axes of $\proj_A(t)$ are even \textit{typically} nonthermal, consistent with our observation after Eq.~\eqref{eq:sigmaHaar} that large subspaces $\mathcal{H}_{\rho}$ generally have a large variance of principal angles. Given any $\mathcal{H}_{\rho} \subset \mathcal{H}_A$, it is possible that the principal axes of $\proj_{\rho}(t)$ [noting that $\mathcal{H}_{\rho}(t) \subseteq \mathcal{H}_{A}(t)$] have a large overlap with the nonthermal principal axes of $\proj_A(t)$ with respect to $\proj_R$.

This means that there can be no general statement that \textit{every} high-dimensional subspace $\mathcal{H}_{\rho}$ remains thermal \textit{at all times} with respect to $\proj_R$. It may still be possible to show that every high-dimensional subspace remains thermal at almost all times\footnote{Intuitively, that such long-time predictions are even possible for $2$nd order functions as in \cite{dynamicalqthermalization} was not particularly clear within the standard framework of quantum thermalization, but for the existence of semiclassical examples that hinted at such a possibility~\cite{KhinchinStatMech, Shnirelman, CdV, ZelditchOG, ZelditchMixing, Sunada, Zelditch, Anantharaman} as discussed there. We would therefore not like to dismiss the potential existence of some general prediction strategy for $4$th order functions on mere intuitive grounds.}, i.e., maintains its overlap with the principal axes of any corresponding $\proj_A(t)$ only briefly. In Sec.~\ref{sec:Discussion_forecasts}, we will develop an argument that shows that thermalization at almost all times can be extrapolated from finite-time data (usually) involving an additional set of correlators, without fully establishing its computational accessibility.

We will therefore not specialize to autonomous dynamics and continue to express our results for general unitary dynamics, unlike Ref.~\cite{dynamicalqthermalization}, as gains in predictive power for the autonomous case remain to be worked out. For this more general form of unitary dynamics, that a direct determination of $\sigma_{R\rho}^2(t)$ over a given set of times allows constraining pure state thermalization \textit{at the same times} is the best prediction we have to offer in this work. This may be compared to a simpler variant of Eq.~\eqref{eq:synopsis_autocorrelator_to_correlator} that can be derived for general (not necessarily autonomous) dynamics, discussed in Appendix~\ref{app:autocorrelators_ips} [Eq.~\eqref{eq:autocorrelator_to_correlator_generaldynamics}]. We will return to this issue qualitatively in Sec.~\ref{sec:Discussion}, but now turn to showing that the determination of $\sigma_{R\rho}^2(t)$ at a given time $t$ is significantly easier from a computational or experimental perspective than a direct determination of thermalization in pure states in many-body systems.

\subsection{Application to many-body systems}
\label{sec:manybodysystems}

The high-dimensional projectors in our results have a natural realization in many-body systems: as anticipated in Sec.~\ref{sec:intro} and Sec.~\ref{sec:synopsis}, they can be thought of (for example) as projectors onto specific pure states in a subsystem, with the state of the complementary subsystem left completely undetermined. In this setting, consider two tensor product decompositions of the Hilbert space, $\mathcal{H} = \mathcal{H}_S \otimes \mathcal{H}_E$ and $\mathcal{H} = \mathcal{H}_{\sigma} \otimes \mathcal{H}_{\eta}$, such that $\dim H_S = D_S$, $\dim \mathcal{H}_E = D_E$, $\dim \mathcal{H}_\sigma = D_{\sigma}$, and $\dim \mathcal{H}_{\eta} = D_{\eta}$, with $D = D_S D_E = D_\sigma D_{\eta}$. For a system of qubits, we also have the particle numbers $N_k = \log_2 D_k$ for each of these subsystems. Then, for any pair of states $\lvert \chi\rangle_S \in \mathcal{H}_S$ and $\lvert \psi\rangle_{\sigma} \in \mathcal{H}_{\sigma}$, it is natural to choose
\begin{align}
    \proj_R &= \lvert \chi\rangle_S\langle \chi\rvert \otimes \idop_E, \label{eq:projRdef} \\
    \proj_{\rho} &= \lvert \psi\rangle_\sigma\langle \psi\rvert \otimes \idop_{\eta}. \label{eq:projrhodef}
\end{align}
Note that $D_R \equiv \Tr[\proj_R] = D_E$ and $D_{\rho} \equiv \Tr[\proj_{\rho}] = D_{\eta}$. In such a realization, $\proj_R$ may represent an eigenstate of a few-body observable on $S$, say $\hat{A}_S = a \lvert \chi\rangle_S\langle \chi\rvert + \ldots$; more generally, one can write $\hat{A}_S = \sum_{k=1}^{D_S} c_k \proj_{Rk}$ for an orthogonal family of projectors $\proj_{Rk}$ of the above form. Concluding thermalization for all such projectors of interest can subsequently imply the thermalization of more general few-body observables.

Then, Theorem~\ref{thm:quantumthermalization_subspaces} and Proposition~\ref{prop:thermalizationsubspacesconverse} address the thermalization of $\proj_R$ in initial states of the form
\begin{equation}
    \lvert \Psi_k\rangle = \lvert \psi\rangle_{\sigma} \otimes \lvert \phi_k\rangle_{\eta},
\end{equation}
where $\lvert \phi_k\rangle_{\eta}$ forms an orthonormal basis for $\mathcal{H}_{\eta}$, with (say) $k \in \mathbb{Z}_{D_{\eta}}$, according to:
\begin{corollary}[Quantum thermalization from correlators for pure states of a bath]
\label{cor:manybodypurestatethermalization}
For a given initial state $\lvert \psi\rangle_{\sigma}$ of the ``core'' subsystem $\sigma$, define the ``core'' projector $\proj_{\psi} \equiv \lvert \psi\rangle_\sigma\langle \psi\rvert \otimes \idop_{\eta}$. Further, for any orthonormal basis $\mathcal{B}_{\eta} \equiv \left\lbrace \lvert \phi_k\rangle_{\eta}\right\rbrace_{k\in \mathbb{Z}_{D_{\eta}}}$ of the complementary ``bath'' subsystem $\eta$, define the set of normalized initial states corresponding to the given core state and all bath states in the basis:
\begin{equation}
    \mathscr{I}[\mathcal{B}_{\eta}] = \left\lbrace \lvert \Psi_k(0)\rangle = \lvert \psi\rangle_{\sigma} \otimes \lvert \phi_k\rangle_{\eta} :\ \lvert \phi_k\rangle_{\eta} \in \mathcal{B}_{\eta}\right\rbrace.
\end{equation}
Then, for any projector observable $\proj_R$ (satisfying $\proj_R^2 =\proj_R$, $\proj_R = \proj_R^\dagger$)
\begin{enumerate}
    \item If the correlation functions of $\proj_{\psi}(t) = \uh(t) \proj_{\psi} \hat{U}^\dagger(t)$ and $\proj_R$ show a ``variance'' at least as small as some $\epsilon > 0$ squared:
\begin{equation}
    \frac{1}{D_{\eta}}\Tr\left[\proj_{\psi}(t)\proj_R \proj_{\psi}(t) \proj_R\right]-\left(\frac{1}{D_{\eta}}\Tr\left[\proj_{\psi}(t)\proj_R\right]\right)^2 \leq \epsilon^2
    \label{eq:epsilonincorollary}
\end{equation}
at a given time $t$, then for \textit{any} choice of $\mathcal{B}_{\eta}$ the fraction of initial states in $\mathscr{I}[\mathcal{B}_{\eta}]$ that are nonthermal at time $t$ to any given resolution $\lambda > 0$,
\begin{equation}
    f_{\lambda}[\mathcal{B}_{\eta}](t) \equiv \frac{1}{D_{\eta}}\sum_{\lvert \Psi_k(t)\rangle \in \mathscr{I}[\mathcal{B}_{\eta}]} \ifc\left(\left\lvert \langle \Psi_k(t)\rvert \proj_R\lvert \Psi_k(t)\rangle - \frac{1}{D_{\eta}}\Tr\left[\proj_{\psi}(t)\proj_R\right]\right\rvert > \lambda\right)
\end{equation}
is constrained by:
\begin{equation}
    f_{\lambda}[\mathcal{B}_{\eta}](t) \leq \frac{3\epsilon^{2/3}}{4^{1/3}\lambda}.
    \label{eq:flambdaconstraintincorollary}
\end{equation}
\item If, on the other hand, the ``variance'' is at least as large as some $\gamma>0$ squared,
\begin{equation}
    \frac{1}{D_{\eta}}\Tr\left[\proj_{\psi}(t)\proj_R \proj_{\psi}(t) \proj_R\right]-\left(\frac{1}{D_{\eta}}\Tr\left[\proj_{\psi}(t)\proj_R\right]\right)^2 \geq \gamma^2,
\end{equation}
at a given time $t$, then there exists an orthonormal basis $\widetilde{B}_{\eta,t}$ for $\mathcal{H}_{\eta}$ corresponding to each such time $t$ such that a certain  minimum fraction of initial states in $\mathscr{I}[\widetilde{B}_{\eta,t}]$ are nonthermal to a resolution $\lambda > 0$ at the time $t$:
\begin{equation}
    f_{\lambda}[\widetilde{B}_{\eta,t}](t) \geq \frac{\gamma^2-\lambda^2}{1-\lambda^2}.
\end{equation}
\end{enumerate}
\end{corollary}
\begin{proof}
Each statement respectively follows from applying Theorem~\ref{thm:quantumthermalization_subspaces} and Proposition~\ref{prop:thermalizationsubspacesconverse} to the stated setting.
\end{proof}

The content of the above corollary is that it is necessary and sufficient for the ``variance'' measure of $\proj_{\psi}(t)$ and $\proj_R$, constructed out of their $2$nd and $4$th order [out-of-time-ordered] correlators, to be small at time $t$ for $\proj_R$ to thermalize in an overwhelming fraction of every orthonormal basis of bath states at the time $t$ for the given core state. The time $t$ here is somewhat of a formality, given that our results do not relate quantities at different times. However, our results isolate the variance of such few body correlators as the determiners of thermalization that is sensitive to pure states in the bath. This provides a rigorous statement of our main result concerning quantum thermalization in many-body systems, on which the summary in Statement~\ref{statement:informalresult_synopsis} of Sec.~\ref{sec:synopsis} is based.

Typically, as per the arguments of Sec.~\ref{sec:typicality}, we require that $D_{\sigma} \gg D_S = \Tr[\proj_R]/D$ for the variance to be small. This makes $\proj_{\psi}$ a slightly nonlocal object, in the sense that the ``core'' subsystem $\sigma$ necessarily has to be much larger than the observed subsystem $S$. However, it can still be a finite subsystem, $D_{\sigma} = O(1)$, in the thermodynamic limit $D\to\infty$. In this sense, $\proj_{\psi}$ remains an extended but local operator relative to the system size, which we term ``controllably nonlocal''. Quantitatively, the required degree of nonlocality depends very sensitivity on the desired resolution. For example, using the typical behavior of the variance in Eq.~\eqref{eq:sigmaHaar}, we find that the dimension of the core subsystem must satisfy:
\begin{equation}
    D_{\sigma} \gtrsim \frac{D_S-1}{4D_S^2}\left(\frac{3}{\lambda f_{\lambda}}\right)^3 = \frac{D_S(D_S-1)}{4}\left(\frac{3}{\lambda_{\text{rel}} f_{\lambda}}\right)^3,
\end{equation}
where $\lambda_{\text{rel}} = \lambda D_S$ represents the resolution relative to the Haar typical thermal value $1/D_S$ of $\proj_R$. The cubic scaling of $D_{\sigma}$ with $\lambda_{\text{rel}} f_{\lambda}$ is the main practical challenge: for $N_S = 1$, with a crude resolution of $\lambda_{\text{rel}} \leq 0.1$ and $f_{\lambda} \leq 0.1$ (i.e. at least $90\%$ of bath states in every basis thermalize to $\pm 10\%$ resolution), we require a measurement of $\sigma^2 \lesssim 10^{-7}$ with $N_{\sigma} \gtrsim 24$, which requires high experimental fidelity. A more near-term experimental goal may be to establish thermalization to low resolution in a non-negligible fraction of bath states (which is still a thermodynamically large number): for $\lambda \leq 0.1$, $f_{\lambda} \leq 0.9$ (at least $10\%$ of bath states in every basis thermalize to $\pm 0.1$ resolution), and $N_S = 1$, it should suffice to have $\sigma^2 \lesssim 10^{-4}$ with $N_{\sigma} \gtrsim 12$.

Even when restricting ourselves to a specific instant of time $t$, one key qualitative difference between Corollary~\ref{cor:manybodypurestatethermalization} and Statements~\ref{intro1:statement1} and \ref{intro1:statement2} in Sec.~\ref{sec:intro} is that the latter pair only require as dynamical input the autocorrelator of a given observable, but the former involves correlators that require choosing both an observable and a core initial state. The need for fixing a core initial state stems from the fact that not all bases of initial states can be thermal at time $t$; notably, a basis that evolves to coincide at time $t$ with an eigenbasis of the (nontrivial projector) observable is necessarily non-thermal. This problem does not arise in Statements~\ref{intro1:statement1} and \ref{intro1:statement2} of Sec.~\ref{sec:intro} as they are not statements about specific instants of time: different initial bases may only briefly have significant overlap with an eigenbasis of the observable at different times, allowing each to thermalize either on average over time or at almost all times. Again from the discussion in Sec.~\ref{sec:restoredynamics}, we emphasize that thermalization at a specific time $t$ typically \textit{does not occur} for every possible choice of core initial state $\lvert \psi\rangle_{\sigma}$, particularly those that correspond to nonthermal bath states for a smaller core $\sigma'$ within $\sigma$ (e.g. with $N_{\sigma'} \sim N_S$). It is likely that thermalization does occur at \textit{almost all} times for every choice of core initial state in generic systems of interest, but this is beyond the rigorous reach of our present framework (possibly being tied to the problem of extrapolating to other times discussed in Sec.~\ref{sec:restoredynamics} and later in Sec.~\ref{sec:Discussion}), and is an important avenue for potential future extensions e.g. as discussed in Sec.~\ref{sec:Discussion_forecasts}.

Finally, let us consider how we might choose the core subsystem in practice. With the observable $\proj_R$ initially localized to the $N_S$-qubit subsystem $S$, any projector $\proj_{\psi}$ localized to a subsystem $\sigma$ within its complement $E$ does not connect different eigenspaces of $\proj_R$ (i.e., $\proj_R$ and $\proj_{\psi}$ commute). This means that one can identify states within the corresponding subspace $\mathcal{H}_{\psi}$ that lie entirely within each eigenspace of $\proj_R$. In other words, for such a choice of core subsystem, bath states that are entirely non-thermal always exist. Furthermore, this situation of nonthermality would persist until $\proj_{\psi}$ spreads into the subsystem $S$ (i.e. the commutator remains zero, or the OTOC remains large, until the operator spreads from $\sigma$ to $S$). It is therefore apparent that a practically convenient choice of core subsystem $\sigma$ that may show OTOC decay and pure state thermalization at early times, without requiring a mandatory wait~\cite{LucasReview} for operator spreading between subsystems, is one that already contains the subsystem $S$.

\section{Discussion: Experiments, timescales, and forecasts}
\label{sec:Discussion}

Having set up a formalism for predicting quantum thermalization in pure states of the bath given a knowledge of the dynamics of few-body correlators, we will now discuss some of its implications for theoretical and experimental studies of thermalization.

\subsection{Experimental accessibility}

First, let us consider the experimental accessibility of these measurements in an $N$-particle system, which we will take to be qubits for specificity. Given the core subsystem of $N_{\sigma}$ qubits in a state $\lvert \psi\rangle_{\sigma}$ with a bath of $N_\eta = N-N_{\sigma}$ qubits, and a few body observable $\proj_{R}$ supported on a subsystem of $N_S \ll N_{\sigma}$ qubits [see Eqs.~\eqref{eq:projRdef}, \eqref{eq:projrhodef} for formal expressions with $\proj_{\rho}$ identified with $\proj_{\psi}$], we need to measure the $2$nd order correlator and the $4$th order OTOC (reproduced here for convenience):
\begin{align}
    G_{R\psi}^{(2)}(t) &= \frac{1}{D_\eta}\Tr[\proj_{\psi}(t)\proj_R], \\
    G_{R\psi}^{(4)}(t) &= \frac{1}{D_\eta}\Tr[\proj_{\psi}(t)\proj_R \proj_{\psi}(t)\proj_R].
\end{align}
We recall that a finite value of $N_{\sigma}$ is sufficient to establish thermalization to within a finite resolution, with a weaker $N_{\sigma} \to \infty$ limit guaranteeing infinite resolution. From an experimental standpoint, we will assume that a finite resolution is inescapable, and therefore take $N_{\sigma}$ to be formally finite and independent of $N$.

The measurement of $G_{R\psi}^{(2)} (t)$ is straightforward, and can be done by preparing the bath subsystem $\eta$ in a maximally mixed state (which is easy to create assuming that some basic quantum gates and qubit-level control are available, e.g. by preparing each qubit in $\eta$ in a maximally entangled state with another external qubit~\cite{NielsenChuang}, see also a similar discussion in this context in \cite{dynamicalqthermalization}), evolving for the time $t$, and measuring $\proj_R$ projectively~\cite{NielsenChuang}. This corresponds to dynamics in the mixed initial state $\proj_\psi/D_{\eta}$, which we can write as:
\begin{equation}
    \hat{\rho}_{\psi} = \lvert \psi\rangle_{\sigma}\langle \psi\rvert \otimes \frac{\idop_{\eta}}{D_{\eta}}.
\end{equation}
Moreover, as $G_{R\psi}^{(2)}(t) \sim 1/D_S$ is typically finite for thermal dynamics, these measurements are successful with a finite probability\footnote{Recall that $p_{\proj}(\hat{\rho}) = \Tr[\hat{\rho}\proj]$ is the postselection probability of a projective measurement, which is $G_{R\psi}^{(2)}(t)$ here essentially by definition, whose finiteness implies that only a finite number of measurements are necessary for the measured expectation value of $\proj$ to approach its ideal expectation. In particular, the number of measurements required scales as $N_M \sim p_{\proj}(\hat{\rho})^2$.} in the thermodynamic limit $D\to\infty$.

Our main concern here is more the measurement of $G_{R\psi}^{(4)}(t)$ with the specific normalization required for our results. In particular, we want to be able to measure \textit{finite} (accessible) values of $G_{R\psi}^{(4)}(t)$, recalling that $[G_{R\psi}^{(2)}(t)]^2 \leq G_{R\psi}^{(4)}(t) \leq 1$ with both bounds being finite in our setting. The easiest way to measure an OTOC with the natural dynamics $\uh(t)$ available in a platform (i.e. no control qubits or time reversal) is to interpret it~\cite{YaoPurityOTOC} as the purity of $\hat{\rho}_{\psi}$ under the completely positive quantum operation~\cite{NielsenChuang}:
\begin{equation}
    \mathcal{M}_t(\hat{\rho}_{\psi}) = \proj_R \uh(t) \hat{\rho}_{\psi} \hat{U}^\dagger(t)\proj_R,
\end{equation}
which is straightforward to implement by a projective measurement of $\proj_R$ after evolving $\hat{\rho}_{\psi}$ for the time $t$. The purity of $\mathcal{M}(\hat{\rho})$, which can be measured at time $t$ e.g. with local-qubit control~\cite{AndreasPurityRM} is then related to the OTOC by the simple (but divergent) scale factor $D_{\eta}$:
\begin{equation}
    \mathcal{P}(t) \equiv \Tr\left(\mathcal{M}_t(\hat{\rho})^2\right) = \frac{1}{D_{\eta}}G_{R\psi}^{(4)}(t),
\end{equation}
wherein lies the obstacle for scalable measurements. For $G_{R\psi}^{(4)}(t)\leq 1$, we have $\mathcal{P}(t) \leq 1/D_{\eta} = 2^{-(N-N_{\sigma})}$, which is exponentially small in $N$ for finite $\sigma$, and therefore would typically require an inaccessible number of measurements. More specifically, a standard purity measurement protocol based on randomized measurements~\cite{AndreasPurityRM} even appears to require exponentially many measurements (in the subsystem size, which here is the full system) for any value of purity. Further, direct randomized measurements of OTOCs~\cite{Vermersch2019, Joshi2020} again appear to require exponentially many measurements in general, with workarounds possible at short times with limited operator spreading in local systems.

The only available measurement techniques that are fully scalable in the thermodynamic limit appear to be those that interpret $G_{R\psi}^{(4)}(t)$ as a normalized trace of a Kraus operator~\cite{NielsenChuang} (here, a combination of unitary, reversed unitary, and projective evolution) such as
\begin{equation}
    \hat{\mathcal{K}}_{R\psi}(t) \equiv \uh^\dagger(t) \proj_R \uh(t) \proj_{\psi}\uh^\dagger(t)\proj_R\uh(t),
\end{equation}
in terms of which one can write
\begin{equation}
    G_{R\psi}^{(4)}(t) = \Tr\left[\hat{\rho}_{\psi}\hat{\mathcal{K}}_{R\psi}(t)\right]
\end{equation}
The quantity on the right hand side can be measured to $\Theta(1)$ resolution with an accessible number of measurements, most directly~\cite[Appendix F]{SFFmeas} by initializing the system to the state $\hat{\rho}_{\psi}$ and using a control qubit\footnote{To summarize this strategy for completeness, one initializes the control qubit to an eigenstate of $\sigma_x$, applies the evolution $\lvert 0\rangle\langle 0\rvert \otimes \idop + \lvert 1\rangle\langle 1\rvert \otimes \hat{\mathcal{K}}_{R\psi}(t)$ where $\lvert 0\rangle, \lvert 1\rangle$ are $\sigma_z$ eigenstates, and measures $(\sigma_x + i\sigma_y)$ whose expectation value gives this trace.} to apply each step in $\hat{\mathcal{K}}_{R\psi}(t)$. This allows a scalable measurement of finite values of $G_{R\psi}^{(4)}(t)$ e.g. when $N_{\sigma}$ is finite, as $G_{R\psi}^{(4)}(t)$ is directly encoded in spin expectation values of the control qubit without any divergent scale factors. Moreover, for Hamiltonian evolution, using an additional control qubit to switch the effective sign of $\hat{H}$ can implement the reversal of $\uh(t)$, as in Ref.~\cite{HafeziOTOCprotocol}. With these protocols already in place, the main experimental challenge is the magnitude of the ``finite'' values to be measured. For example, we require (1) an accuracy of $10^{-7}$ with $N_{\sigma} \sim 24$ for the thermalization of a single qubit to $10\%$ resolution in $90\%$ of states in every bases, (2) an accuracy of $10^{-4}$ with $N_{\sigma} \gtrsim 12$ for single-qubit thermalization to an absolute resolution of $0.1$ in $10\%$ of states in every basis. The latter still constitutes a nontrivial rigorous prediction that may be more accessible in the near term, as discussed in Sec.~\ref{sec:synopsis_thiswork} and Sec.~\ref{sec:manybodysystems}.

\subsection{Timescales of mixing and thermalization}

Now, we pivot to some conceptual implications. The role of chaos in thermalization (classical or quantum) is a prominent subject in everyday folklore~\cite{KhinchinStatMech, DAlessio2016, deutsch2018eth}. Yet, care must be taken in discerning its precise physical implications.

In ergodic theory, chaos is most closely associated with a nonzero invariant (Kolmogorov-Sinai) dynamical entropy~\cite{SinaiCornfeld}. In systems with ``smooth'' dynamics, this entropy can be related to the Lyapunov exponents, associated with the exponential divergence of nearby trajectories, via Pesin's theorem~\cite{Pesin, RuelleEntropyLyapunov}. It is not necessarily the case that the exponential divergence represented by a nonzero dynamical entropy implies an exponentially fast rate of mixing, though the converse does hold~\cite{ExpMixingImpliesBernoulliK}; instead, the mixing of conventional two-point correlation functions can be arbitrarily slow or even absent with nonzero dynamical entropy~\cite{PlatoErgodic}.

In a classical limit (for systems that admit them), a rescaled OTOC can be reduced to a Poisson bracket with a phase space derivative structure, which is sensitive to Lyapunov exponents~\cite{LarkinOvchinnikov, MSSotocBound, XuSwingleScramblingTutorial, GJWOTOCReview}. To this extent, an OTOC can be identified as a quantum dynamical quantity that captures some aspects of classical dynamical entropy generation\footnote{Of course, with quantum dynamics being fundamentally different from classical dynamics, there is usually no such thing as \textit{the} unique quantum counterpart to classical dynamical entropy. In this regard, we would also like to entertain the notion that the spectral form factor~\cite{Haake, BerrySpectralRigidity} is a true quantum dynamical invariant that measures one sense of ``entropy'' generation (or ``information loss'') without requiring any choice of basis or observables. This is partly supported by the results of \cite{dynamicalqspeedlimit, dynamicalqfastscrambling}, by which the spectral form factor rigorously constrains the rate of ``information loss'' in projection valued measures as well as of entanglement entropy generation, and in fact specifies the optimal rate at which the former may occur.}. Here, a point worth emphasizing is that nothing about classical dynamical entropy in classical ergodic theory would suggest that the OTOC is sensitive to mixing or thermalization per se, beyond non-informative quantum bounds such as Eq.~\eqref{eq:G2G4bound}. Indeed, it has explicitly been suggested in the context of quantum information scrambling~\cite{Susskind_dS_scrambling} that OTOCs are sensitive to features beyond thermalization, the latter being presumed to belong to the domain of two-point functions (as with the mixing of observables in classical ergodic theory~\cite{KhinchinStatMech}).

From this point of view, we find that our results provide some food for thought. In particular, we have rigorously associated (controllably nonlocal) OTOCs with quantum thermalization. In systems with a classical limit subject to Pesin's theorem~\cite{Pesin, RuelleEntropyLyapunov}, this suggests that the classical dynamical entropy, which does not appear to have an obviously close connection to classical thermalization with statistical averages (in the sense of few-body ergodicity or mixing), may contain information about the \textit{initial} approach of the quantum system towards thermalization (if it occurs) in the least thermal quantum states at any given instant (instantaneous principal axes). In other words, the classical dynamical entropy may track the initial approach towards quantum thermalization in bases that may not even be visible to classical dynamics, even as commutators or OTOCs change negligibly in the classical limit (see also footnote~\ref{footnote:commutatorPB}). It would be interesting to consider if this is mere speculation or can be justified by a robust analysis of the classical limit.

In purely quantum systems, we obtain a more concrete picture of the different timescales of thermalization available in the full Hilbert space. The timescale of autocorrelator decay corresponds to the timescale of quantum mixing, which in the setting of Sec.~\ref{sec:intro} and Sec.~\ref{sec:synopsis} is the timescale of thermalization of \textit{typical} states in the bath. But it is the timescale of OTOC saturation that determines the largest timescale of thermalization in the \textit{collection} of every basis in the bath (i.e. the timescale required for the almost all states in the ``most'' nonthermal basis of bath states at each instant to thermalize). In other words, an autocorrelator captures the \textit{typical} rate of thermalization of an observable, while the set of (controllably nonlocal) OTOCs measures its \textit{slowest} collective rate across different bases.
Additionally, the spectral form factor~\cite{BerrySpectralRigidity, Haake} measures the \textit{fastest} possible averaged rate of thermalization of a complete set of projective observables~\cite{dynamicalqspeedlimit, dynamicalqfastscrambling} e.g. within a subsystem, across its measurable $\Theta(1)$ values. Between these three measures, we get a fairly detailed and rigorous picture of the dynamical thermalization profiles available in a system from theoretical calculations or experimentally feasible measurements of accessible quantities.



\subsection{Can we foresee thermalization in the future?}
\label{sec:Discussion_forecasts}

A notable tradeoff is worth emphasizing: with the loss of statistical averages, we also seem to have lost the ability to make model-independent forecasts about thermalization even for autonomous unitary or Hamiltonian dynamics\footnote{This is a caveat relative to the expectations set by methods with statistical averages, not an implicit shortcoming in our present method. While such computational forecasts are not possible for \textit{all} allowed initial states due to results based on Turing machines~\cite{ThermalizationUndecidability1, ThermalizationUndecidability2}, we will still argue here that they may be possible for \textit{almost all states} in the Hilbert space provided that a single correlator per observable is computationally accessible.}. In both the time-averaged and state-averaged statements of Ref.~\cite{dynamicalqthermalization}, it turned out to be sufficient to track the dynamics of a single observable in a single mixed state with finite resolution over a finite interval of time to predict the thermalization of the observable over all states and all times (with either of the latter coming with a statistical average), just within the fully general framework of Hamiltonian quantum dynamics. This was due to Eq.~\eqref{eq:synopsis_autocorrelator_to_correlator} [corresponding to Theorem~\ref{thm:autocorrelatorimpliescorrelator_Gen}] having only $O(1)$ prefactors in that case. When applied to OTOCs, it is not clear that such forecasts are generally possible using a similar method (see Appendix~\ref{app:correlators_examples} for details), although we will argue below that a different method may allow such forecasts from finite-time data, whose computational accessibility however remains to be established.

In any case, we can still point to the fact that the OTOCs of few body observables have been derived with quantitative exactness in a variety of physical models, such as ensembles of random quantum circuits~\cite{NahumOTOCCircuits, ChanScrambling, ClaeysLamacraftMaxVelCkt, BertiniPiroliScrambling, RQCreview_2023, XuSwingleScramblingTutorial}, notably for arbitrarily long but finite (i.e. asymptotically long) times in the $N\to\infty$ limit. While these are usually single-site OTOCs rather than the more general ``controllably nonlocal'' few-body OTOCs of interest here, it appears possible that such techniques may generalize to the latter, especially if one is only interested in showing thermalization to finite resolution (i.e., a finite core subsystem). In particular, the minimum number of sites required in the core subsystem for a given resolution decreases with increasing local Hilbert space dimension at each site. The question of whether qualitatively new techniques are necessary for controllably nonlocal OTOCs with infinite resolution (i.e. a core subsystem whose size grows without bound, but yet is not thermodynamically large) may be interesting to explore. Such an exact derivation of these OTOCs for an observable over a long time interval in any system, if it shows their decay even to a finite or at best $o(1)$ resolution in the thermodynamic limit, would then rigorously imply the thermal equilibration of the observable in an overwhelming fraction of thermodynamically many pure states over the corresponding times.

But for the forecast caveat (partially addressed below), which can be circumvented in such specific classes of models that allow a determination of controllably nonlocal OTOCs, our results here (without statistical averages but restricted to times of observation) and in Ref.~\cite{dynamicalqthermalization} (with at least one statistical average but with the ability to make forecasts for infinite times) bring us tantalizingly close to a complete rigorous formalism for quantum statistical mechanics that allows one to predict thermalization from (finite-time) observations of few-body observables rather than their inaccessible (infinite-time) properties in energy eigenstates or any other detailed information about the dynamics. Given that the formalism of Ref.~\cite{dynamicalqthermalization} can already handle nontrivial thermalization forecasts in \textit{typical} pure states (of a large subsystem) in place of state averages, the kind of universal predictive power we are seeking with this present formalism is notably stronger in providing thermalization guarantees for \textit{every} pure state basis in the Hilbert space. Indeed, we seem to have already surpassed the predictive power of classical statistical mechanics in this regard, which has no obvious analogue of pure state thermalization. It would nevertheless be interesting to test the limits of this approach and complete this formalism (or prove its incompleteness) by explicitly identifying an accessible method to make long-time forecasts of OTOC behavior from finite-time observations (or prove that this cannot be done in general), or at least by enumerating classes of systems where this may be possible.

\paragraph*{Argument for forecasting eternal quantum thermalization from finite-time data} With a view towards resolving the above question, we will now describe an argument that shows that forecasts of long-time quantum thermalization from finite-time data are \textit{in principle} possible for Hamiltonian or time-translation invariant dynamics, somewhat inspired by the state-dependent forecasting method of Ref.~\cite{dynamicalqaperiodicity} in the separate context of macroscopic thermalization. In particular, we will show that for a given observable $\proj_R$ and a suitable set of times $\mathcal{T}$ of finite size or range $\lvert \mathcal{T}\rvert < \infty$, for any $D_{\rho}$ there exists a \textit{single} subspace $\mathcal{H}_{\rho}^{\max}$ of dimension $D_{\rho}$ with projector $\proj_{\rho}^{\max}$, such that the factorization of its controllably nonlocal OTOC at almost all times in $\mathcal{T}$ implies the thermalization of $\proj_R$ in almost all states of \textit{any} subspace $\mathcal{H}_{\rho}$ of dimension $D_{\rho}$ in the Hilbert space, at \textit{almost all} times inside or outside $\mathcal{T}$. Thus, a single controllably nonlocal OTOC of the observable over finite times contains enough information to control its thermalization in almost all states over almost all times, which can allow computational forecasts provided that this OTOC can be computed.

We begin our formal argument with the observable $\proj_{R}$ and an unspecified ensemble of states $\proj_{\rho}$, in the notation of Sec.~\ref{sec:dynamics}. As per Statement~\ref{statement:informalresult_synopsis} or Corollary~\ref{cor:manybodypurestatethermalization}, the quantum thermalization of $\proj_R$ in almost all states in $\mathcal{H}_{\rho}$ at a specified time $t$ is determined (given the thermal behavior of $G_{R\rho}^{(2)}(t)$ as can be established at almost all times using finite-time autocorrelators \cite{dynamicalqthermalization}) by the smallness of
\begin{equation}
    G_{R\rho}^{(4)}(t) - [G_{R\rho}^{(2)}(t)]^2.
\end{equation}
By the bounds in Eq.~\eqref{eq:G2G4bound} and the fact that each correlator itself satisfies $0\leq G_{R\rho}^{(2n)}(t) \leq 1$, the time average of this quantity over the times in any set $\mathcal{T}$ is bounded (and is a continuous function of $\proj_{\rho}$):
\begin{equation}
    \Delta_4 G\left(\proj_{\rho}; \proj_{R}, \mathcal{T}\right) \equiv \frac{1}{\lvert \mathcal{T}\rvert}\int_{\mathcal{T}}\diff t\ \left\lbrace G_{R\rho}^{(4)}(t) - [G_{R\rho}^{(2)}(t)]^2\right\rbrace\ \in\ [0,1].
\end{equation}
Every projector of dimension $D_{\rho}$ can be written as $\proj_{\rho}(\hat{W}) = \hat{W} \proj_{\rho 0}\hat{W}^{\dagger}$ for a fixed $\proj_{\rho 0}$, where $\hat{W}$ is unitary. Therefore, $\Delta_4 G\left(\proj_{\rho}; \proj_{R}, \mathcal{T}\right)$ attains a maximum as a function of $\proj_{\rho}(\hat{W})$ for at least one choice of $\hat{W} = \hat{W}_{\max}$. Let $\proj_{\rho}^{\max} \equiv \proj_{\rho}(\hat{W}_{\max})$. We can formally write the existence of this maximum as:
\begin{equation}
    \Delta_4 G\left(\proj_{\rho}(\hat{W}); \proj_{R}, \mathcal{T}\right) \leq \Delta_4 G\left(\proj_{\rho}^{\max}; \proj_{R}, \mathcal{T}\right),\;\;\ \text{ for  every unitary } \hat{W}.
    \label{eq:forecastinequalitystatic}
\end{equation}
This  inequality is specified over a given set of times $\mathcal{T}$. Now, we can leverage the time-translation invariance of Hamiltonian dynamics: if we translate every time in $\mathcal{T}$ by a fixed amount $t_0$, resulting in the shifted set of times (note that even if $\mathcal{T}$ is discrete, such shifts can collectively span $\mathbb{R}$)
\begin{equation}
    \mathcal{T}(t_0) \equiv \lbrace t+t_0: \text{ for all } t \in \mathcal{T}\rbrace, 
\end{equation}
then a unitary transformation $\proj_{\rho}(\hat{W}) \to e^{i\hat{H} t_0} \proj_{\rho}(\hat{W}) e^{-i \hat{H} t_0}$ leaves the correlators invariant, by Eqs.~\eqref{eq:timeevolution} and \eqref{eq:timedependentcorrelatorsdef} with $\uh(t) = e^{-i\hat{H} t}$. Therefore,
\begin{equation}
    \Delta_4 G\left(\proj_{\rho}(e^{i\hat{H} t_0}\hat{W}); \proj_{R}, \mathcal{T}(t_0)\right) = \Delta_4 G\left(\proj_{\rho}(\hat{W}); \proj_{R}, \mathcal{T}\right)\;\;\ \text{ for every } t_0 \in \mathbb{R}.
\end{equation}
Noting that $e^{i\hat{H} t_0} \hat{W}$ is itself unitary, we can replace $\hat{W} \to e^{i\hat{H} t_0} \hat{W}$ in the condition in Eq.~\eqref{eq:forecastinequalitystatic}, obtaining
\begin{equation}
    \Delta_4 G\left(\proj_{\rho}(\hat{W}); \proj_{R}, \mathcal{T}(t_0)\right) \leq \Delta_4 G\left(\proj_{\rho}^{\max}; \proj_{R}, \mathcal{T}\right),\;\;\ \text{ for  every unitary } \hat{W} \text{ and every } t_0 \in \mathbb{R}.
\end{equation}
This relates the computation of the right hand side in $\mathcal{T}$ to the behavior of the left hand side for arbitrary $\mathcal{T}(t_0)$, including $t_0 \to \pm \infty$.
By Markov's inequality~\cite{RossProbability}, if the left hand side is small, then $G_{R\rho}^{(4)}(t) - [G_{R\rho}^{(2)}(t)]^2$ is small at almost all times in $\mathcal{T}(t_0)$ for every $t_0$, similar to the considerations in Sec.~\ref{sec:restoredynamics}. Thus, the smallness of $\Delta_4 G\left(\proj_{\rho}^{\max}; \proj_{R}, \mathcal{T}\right)$ over the finite set $\mathcal{T}$ is sufficient to establish the requisite behavior of OTOCs, for \textit{every} initial $\proj_{\rho}$, for the thermalization of $\proj_{R}$ at almost all times in $\mathbb{R}$ (or $\mathbb{Z}$). The primary open questions for this argument are whether (a suitable approximation to) $\proj_{\rho}^{\max}$ for a given observable $\proj_R$ and set of times $\mathcal{T}$ can be accessed computationally (which appears plausible given that it is only a finite-time correlator that is being maximized)
and if the corresponding correlator is indeed small in enough cases of interest, which we hope to characterize in future work.

\subsection*{Acknowledgments}

We thank Laura Shou for interesting discussions on the differences between measure concentration on the sphere and on the unitary group~\cite{HaarBook}, and the one-way implication from exponentially fast second-order mixing to K-mixing in classical ergodic theory~\cite{ExpMixingImpliesBernoulliK}.
This work was supported by NIST, by the National Science Foundation under Grant Number 1734006 (Physics Frontier Center), and by the Heising-Simons Foundation under Grant 2024-4848.

\appendix

\section{Quantum thermalization from autocorrelators}

\label{app:autocorrelator_method}

In this section, we will review some key results of Ref.~\cite{dynamicalqthermalization} for standard two-point correlation functions and the thermalization of observables with statistical averages. Where these results were partly expressed using an intermediate notion of energy-band thermalization in \cite{dynamicalqthermalization}, here we will avoid any reference to energy levels as far as possible in our main statements, further highlighting the main power of this formalism as its ability to predict thermalization without relying on energy eigenstates. We will also consider the application of this technique to OTOCs, and show that it does not generalize to these quantities, preventing an easy prediction of pure state thermalization (without statistical averages) over infinite time scales from mere finite time observations.

\subsection{Autocorrelators constrain other correlators: general inner products}
\label{app:autocorrelators_ips}

Let us return to the setting in Sec.~\ref{sec:synopsis}: given an inner product $\langle \cdot,\cdot\rangle$ of time-dependent operators $\hat{A}(t)$ and $\hat{B}(t)$, we would like to constrain the behavior of correlation functions $\langle \hat{B}(t'),\hat{A}(t)\rangle$ in terms of autocorrelation functions $\langle \hat{A}(t'), \hat{A}(t)\rangle$. This appendix is dedicated to this mathematical problem. The larger motivation for constraining correlators in terms of autocorrelators stems partly from Ref.~\cite{KhinchinStatMech}, where autocorrelators enable a formulation of classical statistical mechanics without reference to thermodynamically inaccessible ergodic properties, and more directly from Ref.~\cite{dynamicalqthermalization}, where they play a similar role in formulating quantum statistical mechanics without relying on thermodynamically inaccessible properties of energy levels and eigenstates (such as the eigenstate thermalization hypothesis).

Taking $\hat{A}$, $\hat{B}$ to be \textit{vectors} in a complex finite-dimensional vector space\footnote{In all our applications, this vector space will usually comprise of tensor products of a number of copies of the space of linear operators on our Hilbert space of states $\mathcal{H}$, motivating the hat notation $\hat{A}$, $\hat{B}$ (associated with operators) for these vectors.}, we will require all the standard properties of inner products~\cite{ByronFuller} to enable our results:
\begin{enumerate}
    \item Symmetry under complex conjugation:
    \begin{equation}
        \langle \hat{B}, \hat{A}\rangle^\ast = \langle \hat{A}, \hat{B}\rangle.
        \label{eq:ip_symmetry}
    \end{equation}
    \item Linearity in the second argument:
    \begin{equation}
        \langle \hat{B}, c_1 \hat{A}_1 + c_2 \hat{A}_2\rangle = c_1 \langle \hat{B}, \hat{A}_1 \rangle+c_2\langle \hat{B}, \hat{A}_2\rangle.
        \label{eq:ip_linearity}
    \end{equation}
    Anti-linearity (linearity with complex conjugation) in the first argument follows from here and Eq.~\eqref{eq:ip_symmetry}.
    \item Positive-definiteness:,
    \begin{equation}
        \langle \hat{A}, \hat{A}\rangle \geq 0,
        \label{eq:ip_postivity}
    \end{equation}
    and $\langle \hat{A}, \hat{A}\rangle = 0$ if and only if $\hat{A} = 0$.
\end{enumerate}

We will take the time evolution of operators to be generated by an \textit{autonomous} unitary map $\mathcal{U}_t$ parametrized by $t$:
\begin{equation}
    \hat{A}(t) = \mathcal{U}_t(\hat{A}),
\end{equation}
such that\footnote{Here, $\circ$ denotes function composition, as in $\mathcal{U}\circ\mathcal{V}(\hat{A}) = \mathcal{U}(\mathcal{V}(\hat{A}))$.} $\mathcal{U}_{t_1} \circ \mathcal{U}_{t_2} = \mathcal{U}_{t_1+t_2}$. This corresponds to ``time-independent'' dynamics. Unitarity is defined with respect to the inner product $\langle \cdot, \cdot\rangle$, namely $\langle \mathcal{U}_t(\hat{B}),\mathcal{U}_t(\hat{A})\rangle = \langle \hat{B}, \hat{A}\rangle$ with $\mathcal{U}_t$ being an invertible map, and allows an eigendecomposition of the map in terms of ``energy-levels'' $\mathcal{E}_k \in \mathbb{R}$ and projectors $\mathcal{P}_k$ (satisfying $\mathcal{P}_k \circ \mathcal{P}_k = \mathcal{P}_k$) to ``energy-eigenspaces'':
\begin{equation}
    \mathcal{U}_t(\hat{A}) = \sum_k e^{-i\mathcal{E}_k t} \mathcal{P}_k(\hat{A}).
    \label{eq:genunitary_eigendecomposition}
\end{equation}
These eigenspaces are orthonormal with respect to the inner product,
\begin{align}
    \langle \mathcal{P}_j(\hat{B}), \mathcal{P}_k(\hat{A})\rangle &= 0\;\;\ \text{ if }\;\;\  j \neq k, \nonumber \\
    \langle \mathcal{P}_k(\hat{B}), \mathcal{P}_k(\hat{A})\rangle &= \langle \hat{B}, \mathcal{P}_k(\hat{A})\rangle = \langle \mathcal{P}_k(\hat{B}), \hat{A}\rangle,
    \label{eq:eigenspaceorthonormality}
\end{align}
and satisfy the completeness relation:
\begin{equation}
    \sum_k \langle \hat{B},\mathcal{P}_k(\hat{A})\rangle = \langle \hat{B}, \hat{A}\rangle.
    \label{eq:eigenspacecompleteness}
\end{equation}

With this setup, our main concern will be to constrain the weighted time-averages of inner products,
\begin{equation}
    \int\diff t\ w(t) \langle \hat{B}, \mathcal{U}_t(\hat{A})\rangle
\end{equation}
in which $w(t)$ is a ``weighting function'' that satisfies $w(t) \geq 0$ and $\int\diff t\ w(t) = 1$. For example, the standard time average over the interval $[0,T]$ corresponds to $w(t \in [0,T]) = 1/T$ and $w(t \notin [0,T]) = 0$. Another subcategory of such weighting functions is that of ``completely positive'' weighting functions~\cite{dynamicalqthermalization}, denoted with a $+$ subscript (as in $w_+(t)$) that in addition have a non-negative Fourier transform:
\begin{equation}
    \widetilde{w}_+(E) \equiv \int\diff t\ w(t) e^{-i E t} \geq 0.
\end{equation}
An example of such a function is the ``tent'' average in $[-T, T]$, i.e. $w_+(t\in [-T,T]) = (1-|t|/T)/T$ and $w_+(t\notin[-T,T]) = 0$.

Such a constraint is given by the following theorem, which mildly generalizes and combines some results in Ref.~\cite{dynamicalqthermalization}, and essentially states that if a completely-positive time average of an autocorrelator is small, then a long time average of any correlator is small\footnote{\label{footnote:autocorrelatortheoremspecifics} To obtain Eq.~\eqref{eq:synopsis_autocorrelator_to_correlator} from Theorem~~\ref{thm:autocorrelatorimpliescorrelator_Gen}, we pick $\mathcal{U}_t(\hat{A}) = \uh(t) \hat{A} \uh^\dagger(t)$, $w(t) = T^{-1}\ifc(t\in[t_0,t_0+T])$ for any $t_0$, $w_+(t) = T_{\obs}^{-1}(1-|t|/T_{\obs})\ifc(t\in [-T_{\obs},T_{\obs}])$ and $\Delta E = 2\xi/T_{\obs}$, from which it follows that $W = \sinc^2 \xi$, $w_0 = T_{\obs}/(\xi T)$, and rearrange some factors; see also Sec.~\ref{app:secondordercorrelators}.} (intuitively, we have in mind $w_0 \approx 0$ and $W \approx 1$ below):
\begin{theorem}[Autocorrelator smallness predicts correlator smallness]
\label{thm:autocorrelatorimpliescorrelator_Gen}
Let $w(t)$ be a weighting function such that for a given ``Fourier window'' $\Delta E > 0$, $\widetilde{w}(E)$ is at most some $w_0: 0 < w_0 < 1$ for $\lvert E\rvert \geq \Delta E$ (implicitly, within $E \in \lbrace \mathcal{E}_k\rbrace \subset \mathbb{R}$, the range of energy levels of $\mathcal{U}_t$)
\begin{equation}
    \left\lvert\widetilde{w}(E: \lvert E\rvert \geq \Delta E)\right\rvert \leq w_0.
    \label{eq:w_w0_constraint}
\end{equation}
Further, let $w_{+}(t)$ be a completely positive weighting function that is larger than some $W: 0< W < 1$ inside the same Fourier window:
\begin{equation}
    \widetilde{w}_+(E: \lvert E\rvert < \Delta E) > W.
    \label{eq:wplus_W_constraint}
\end{equation}
Then, for any $\hat{A}$ and $\hat{B}$, the completely positive $w_+$-time-average of the autocorrelator $\langle \hat{A}, \mathcal{U}_t(\hat{A})\rangle$ constrains the $w$-time-average of the correlator $\langle \hat{B}, \mathcal{U}_t(\hat{A})\rangle$ according to the inequality:
\begin{equation}
    \left\lvert \int\diff t\ w(t) \langle \hat{B}, \mathcal{U}_t(\hat{A})\rangle\right\rvert \leq \left(\sqrt{\frac{1}{W} \int\diff t\ w_+(t) \langle \hat{A}, \mathcal{U}_t(\hat{A})\rangle} + w_0\sqrt{\langle \hat{A},\hat{A}\rangle}\right)\sqrt{\langle \hat{B},\hat{B}\rangle}.
    \label{eq:auto_to_cross_correlator}
\end{equation}
\end{theorem}
\begin{proof}
    This follows according to a proof strategy similar to the combination of Theorem 5.2 and Proposition 6.1 in Ref.~\cite{dynamicalqthermalization}, adapted to this general setting; see Appendix~\ref{proof:autocorrelatorimpliescorrelator_Gen}.
\end{proof}

For intuition, we note that Eq.~\eqref{eq:w_w0_constraint} with $0 < w_0 \ll 1$ corresponds to stating that $w(t)$ is a time average of the correlator $\langle \hat{B}(0), \hat{A}(t)\rangle$ over some long time scale $T \gg 2\pi/\Delta E$; similarly Eq.~\eqref{eq:wplus_W_constraint} with $0 < W \approx 1$ amounts to a time-average of the autocorrelator $\langle \hat{A}(0),\hat{A}(t)\rangle$ over some short time scale $T \ll 2\pi/\Delta E$. In practice, $2\pi/\Delta E$ may be chosen to be slightly larger than the decay timescale of the autocorrelator, and therefore essentially represents the latter up to a large prefactor. Overall, Theorem~\ref{thm:autocorrelatorimpliescorrelator_Gen} allows short-time autocorrelator decay to imply long-time correlator decay for a general inner product with autonomous unitary dynamics.

As an aside, we note that a similar inequality that applies to more general forms of dynamics (not restricted to autonomous unitary dynamics) can be derived, but comes with the loss of predictive power for time intervals notably longer than that over which the autocorrelator is sampled. In particular, a direct application of the Cauchy-Schwarz inequality~\cite{ByronFuller} gives:
\begin{align}
    \left\lvert \int\diff t\ w(t) \langle \hat{B}, \mathcal{U}_t(\hat{A})\rangle\right\rvert &= \left\lvert  \left\langle \hat{B}, \int\diff t\ w(t)\mathcal{U}_t(\hat{A})\right\rangle\right\rvert \nonumber \\
    &\leq \sqrt{\int\diff t\ \diff t'\ w(t)w(t') \langle \mathcal{U}_{t'}(\hat{A}), \mathcal{U}_t(\hat{A})\rangle} \sqrt{\langle \hat{B},\hat{B}\rangle}.
    \label{eq:autocorrelator_to_correlator_generaldynamics}
\end{align}
In this case, however, given the smallness of the right hand side for a given $w(t)$, it does not follow that the left hand side is small for other choices of $w(t)$ (in particular, with significant support over longer time intervals), in contrast to the ability of Theorem~\ref{thm:autocorrelatorimpliescorrelator_Gen} to make long-time forecasts from short-time data. Nevertheless, Eq.~\eqref{eq:autocorrelator_to_correlator_generaldynamics} still shows that autocorrelators constrain other correlators at least over a comparable time range for more general forms of dynamics, providing a comparative standard for what we have for thermalization without averages in terms of OTOCs in the main text.

\subsection{Application to different correlators}
\label{app:correlators_examples}

Given the general result of Theorem~\ref{thm:autocorrelatorimpliescorrelator_Gen}, our task is now to apply it to all manner of correlation functions to obtain rigorous statements about thermalization dynamics. As we will use a number of different inner products that are all consistent with the template of Sec.~\ref{app:autocorrelators_ips}, we will use subscripts to distinguish these inner products, e.g. $\langle \cdot ,\cdot\rangle_{4c}$. It will be convenient to denote the space of linear operators on $\mathcal{H}$ by $L(\mathcal{H})$; these inner products will act on different tensor products of $L(\mathcal{H})$.

\subsubsection{Review: Second order correlators}
\label{app:secondordercorrelators}

First, let us consider the simplest version of this problem: time-averaged thermalization in second-order two-point functions. In this case, our approach will recover some key results of Ref.~\cite{dynamicalqthermalization} (here, with incrementally sharper constants). To begin with, for two operators $\proj_R$ representing a projector observable with $\Tr[\proj_R]/D = D_R/D = 1/D_S$ and $\proj_{\rho}$ representing a subspace of states with $\Tr[\proj_{\rho}]/D = D_{\rho}/D =  1/D_{\sigma}$, we may choose a general structure with the standard trace inner product of two operators:
\begin{align}
    \hat{A}_2, \hat{B}_2 &\in L(\mathcal{H}), \\
    \langle \hat{B}_2, \hat{A}_2\rangle_2 &\equiv \frac{1}{D}\Tr_{\mathcal{H}}[\hat{B}_2^\dagger \hat{A}_2], \\
    \mathcal{U}_t(\hat{A}_2) &= e^{-i\hat{H}t} \hat{A}_2 e^{i\hat{H} t}.
\end{align}
Explicitly, if we make the identifications
\begin{align}
    \hat{A}_2 &= \proj_R - \frac{1}{D_S}\idop, \nonumber \\
    \hat{B}_2 &= D_{\sigma}\proj_{\rho} - \idop,
\end{align}
then for the density operator
\begin{equation}
    \hat{\rho} = \frac{D_{\sigma}}{D}\proj_{\rho},
\end{equation}
Eq.~\eqref{eq:auto_to_cross_correlator} takes the form:
\begin{equation}
    \left\lvert \int\diff t\ w(t) \Tr[\hat{\rho}(t) \proj_R]-\frac{1}{D_S}\right\rvert \leq \sqrt{\frac{D_{\sigma}}{W} \int\diff t\ w_+(t) \left(\frac{1}{D}\Tr[\proj_R(t)\proj_R] - \frac{1}{D_S^2}\right)} + \frac{w_0\sqrt{D_{\sigma}(D_S-1)}}{D_S}.
    \label{eq:timeaveragedbound_app}
\end{equation}
The content of this equation is that for finite (or accessible) $D_S$ and $D_{\sigma}$, the time-averaged saturation of the autocorrelator of $\proj_R(t)$ and $\proj_R$ to nearly $1/D_S^2$ implies the time-averaged thermalization of the expectation value of $\proj_{R}$ in $\hat{\rho}(t)$ to the infinite temperature value $1/D_S = \Tr[\proj_R]/D$. As this holds for all choices of $\proj_{\rho}$, this can be shown to further imply the time-averaged thermalization of the expectation value of $\proj_R$ in almost all \textit{pure states} in every orthonormal basis, as in Theorem 5.3 of Ref.~\cite{dynamicalqthermalization}. We stress that the reduction to pure states from large subspaces $\proj_{\rho}$ even for finite $D_{\sigma}$ (large $D_{\rho}$) is possible only because of the time average for reasons discussed in that work, and a na\"{i}ve choice of $D_{\rho} = 1$ $\implies$ $D_{\sigma} = D$ above does not work directly in Eq.~\eqref{eq:timeaveragedbound_app}; indeed, we do not expect that time-averaged thermalization can be predicted in every pure state, but only an overwhelming fraction of them.

It is also possible to (effectively) eliminate the time average via the technique of ``cloning'' these operators~\cite{dynamicalqthermalization}, which we will indicate by tacking on the subscript $c$. Our setting with cloned operators is:
\begin{align}
    \hat{A}_{2c}, \hat{B}_{2c} &\in L(\mathcal{H})\otimes L(\mathcal{H}), \\
    \langle \hat{B}_{2c}, \hat{A}_{2c}\rangle_{2c} &\equiv \frac{1}{D^2} \Tr_{\mathcal{H} \otimes \mathcal{H}}[\hat{B}_{2c}^\dagger \hat{A}_{2c}] \label{eq:clonedinnerproduct} \\
    \mathcal{U}_t(\hat{A}_{2c}) &= e^{-i\hat{H}t} \otimes e^{-i\hat{H} t} \hat{A}_{2c} e^{i\hat{H}t} \otimes e^{i\hat{H} t}.
\end{align}
Now, if we identify
\begin{equation}
    \hat{A}_{2c} = \hat{A}_2 \otimes \hat{A}_2,\ \hat{B}_{2c} = \hat{B}_2\otimes \hat{B}_2,
\end{equation}
which effectively ``clones'' the operators above, the inner product effectively becomes:
\begin{equation}
    \langle \hat{B}_{2c}, \hat{A}_{2c}\rangle_{2c} = \left(\Tr_{\mathcal{H}}[\hat{\rho}(t)\proj_R] - \frac{1}{D_S}\right)^2.
\end{equation}
In this case, Eq.~\eqref{eq:auto_to_cross_correlator} takes the form
\begin{equation}
    \int\diff t\ w(t) \left(\Tr[\hat{\rho}(t) \proj_R]-\frac{1}{D_S}\right)^2 \leq \sqrt{\frac{D_{\sigma}^2}{W} \int\diff t\ w_+(t) \left(\frac{1}{D}\Tr[\proj_R(t)\proj_R] - \frac{1}{D_S^2}\right)^2} + \frac{w_0 D_{\sigma}(D_S-1)}{D_S^2}.
    \label{eq:clonedthermalequilibriumZ}
\end{equation}
This constrains thermal equilibrium at almost all times as follows. If the autocorrelator $D^{-1}\Tr[\proj_R(t)\proj_R]$ deviates from $1/D_S^2$ by more than $\epsilon$ only in a set of times ${\texcept}_+$ of $+$-weighted duration $0 \leq \kappa_+ \leq 1$ (intuitively, $\kappa_+ \ll 1$), i.e.
\begin{equation}
    \left\lvert D^{-1}\Tr[\proj_R(t)\proj_R] - \frac{1}{D_S^2}\right\rvert > \epsilon \implies t \in {\texcept}_+ : \int_{{\texcept}_+}\diff t\ w_+(t) = \kappa_+,
\end{equation}
then we have (on account of $\lvert D^{-1}\Tr[\proj_R(t)\proj_R] - D_S^{-2}\rvert \leq D_S^{-1}$)
\begin{equation}
    \int\diff t\ w_+(t) \left(\frac{1}{D}\Tr[\proj_R(t)\proj_R] - \frac{1}{D_S^2}\right)^2 \leq \epsilon^2 (1-\kappa_+) + \frac{\kappa_+}{D_S^2}.
\end{equation}
By Eq.~\eqref{eq:clonedthermalequilibriumZ}, we have
\begin{equation}
     \int\diff t\ w(t) \left(\Tr[\hat{\rho}(t) \proj_R]-\frac{1}{D_S}\right)^2 \leq \epsilon_c \equiv D_{\sigma}\sqrt{\frac{\epsilon^2 (1-\kappa_+)}{W} + \frac{\kappa_+}{W D_S^2}} + \frac{w_0 D_{\sigma}(D_S-1)}{D_S^2}.
     \label{eq:time_averaged_thermal_equilibirum_expr}
\end{equation}
Note that $\epsilon_C$ is a small parameter when $\epsilon$, $\kappa_+$ and $w_0$ are sufficiently small. It follows that the set of times $t$ at which $\Tr[\hat{\rho}(t) \proj_R]$ deviates from $D_S^{-1}$ by more than some  $\lambda$ has small $w$-weighted duration: if
\begin{equation}
    \left\lvert \Tr[\hat{\rho}(t) \proj_R]-\frac{1}{D_S}\right\rvert > \lambda \implies t\in \texcept: \int_{\texcept}\diff t\ w(t) = \kappa,
\end{equation}
then by Eq.~\eqref{eq:time_averaged_thermal_equilibirum_expr}, as the integrand (modulo the weight $w(t)$) there is at least $\lambda^2$ in $t \in \texcept$, we have\footnote{\label{footnote:timeintervalconstraintderivation}To obtain Eq.~\eqref{eq:secondordertimeintervalconstraint} from here, we choose $w(t) = T^{-1}\ifc(t\in[t_0,t_0+T])$ for any $t_0$, $w_+(t) = T_{\obs}^{-1}(1-|t|/T_{\obs})\ifc(t\in [-T_{\obs},T_{\obs}])$ and $\Delta E = 2\xi/T_{\obs}$, from which it follows that $\kappa_+ \leq 2 \kappa_{RR}$ [using Eq.~\eqref{eq:autocorrelatortimeintervalconstraint}], $W = \sinc^2 \xi$, $w_0 = T_{\obs}/(\xi T)$, and $\epsilon = \varepsilon/D_S$, together with some additional mild inequalities such as $(1-\varepsilon^2) \leq 1$ and $D_S-1 \leq D_S$.}
\begin{equation}
    \kappa \leq \frac{\epsilon_C}{\lambda^2},
\end{equation}
which is small if $\epsilon_C \ll \lambda^2$. This shows that $\Tr[\hat{\rho}(t) \proj_R]$ equilibrates to $D_S^{-1}$ to within an accuracy $\lambda$ at all times except a set of small $w$-weighted length $\kappa$, if the autocorrelator is close to thermal equilibrium for a sufficiently long time. Here, however, we do not have a way to go from the large subspace $\proj_{\rho}$ to individual pure states for finite values of parameters such as $\epsilon$, $\kappa_+$, $w_0$, and $D_S$. Thus, the time average has, almost inevitably, been replaced by a state average over the different states in $\proj_{\rho}$.

There are several ways to argue that this is necessarily the case, such as the argument based on the classical limit in Sec.~\ref{sec:intro}; here we provide a different argument from the ones in Ref.~\cite{dynamicalqthermalization} that reveals comparable difficulties even in generic quantum systems. For pure state dynamics, we want $D_{\sigma} = D$, so that $D_{\rho} = \Tr[\proj_{\rho}] = 1$. By Eq.~\eqref{eq:time_averaged_thermal_equilibirum_expr}, we see that for $\epsilon_C \ll 1$, we will require (taking $W \approx 1$)
\begin{equation}
    \epsilon \ll \frac{1}{D},\ \kappa_+ \ll \frac{D_S^2}{D^2},\ w_0 \ll \frac{D_S}{D},
\end{equation}
none of which scale accessibly with $D$ for fixed $D_S$. In fact, for a typical unitary, even with $w_+(t)$ significantly supported over a very long time that scales with $D$ or even an infinitely long time, we have for most times (see Appendix~\ref{app:Haar}, and here we note that during typical dynamics, the size of fluctuations can increase relative to the Haar value due to spectral rigidity effects such as a random Hamiltonian-like ramp in fluctuations~\cite{ChaosComplexityRMT, refRampPlateau2, CotlerHunterJones2})
\begin{equation}
    \left\lvert G_{RR}^{(2)}(t) - \frac{1}{D_S}\right\rvert \gtrsim \Omega(D^{-1}) \implies \epsilon \gtrsim \Omega(D^{-1}).
\end{equation}
Therefore, $\epsilon \ll 1/D$ cannot typically be satisfied, and we cannot get a useful bound for these large values of $D_{\sigma}$.

To get some more direct eigenstate-structure based intuition for why this might be the case, consider the expectation value of $\proj_R$ in a pure state $\lvert \psi(t)\rangle$ at time $t$. For simplicity, we will take the pure state to be unbiased in the energy eigenstates $\lvert E_n\rangle$ at $t=0$, i.e.
\begin{equation}
    \lvert \psi\rangle = \frac{1}{\sqrt{D}}\sum_n e^{i\varphi_n}\lvert E_n\rangle.
\end{equation}
This is a crude but reasonable model of typical states that factors out some technical complications for the sake of better intuition; e.g. a Haar random pure state looks equidistributed ``on average''~\cite{vonNeumannThermalization}, though not at the level of individual eigenstates. Moreover, the phases $\varphi_n$ can be absorbed into the eigenstates (for a given initial state) via the redefinition $\lvert \widetilde{E_n}\rangle = e^{-i\varphi_n}\lvert E_n\rangle$. The expectation value of $\proj_R$ in this state at time $t$ is given by:
\begin{equation}
\langle \psi(t)\vert \proj_R\lvert \psi(t)\rangle = \frac{1}{D}\Tr[\proj_R]+ \frac{1}{D}\sum_{n \neq m} \langle \widetilde{E_m}\rvert \proj_R\lvert \widetilde{E_n}\rangle e^{-i(E_n-E_m)t}.
\label{eq:thermalizationenergyeigenbasisexpr}
\end{equation}
The second term on the right hand side can be viewed as the Fourier transform at frequency $t$ of a complex-valued distribution $\Pi_R(\mathcal{E}) = D^{-1}\sum_{n\neq m} \langle \widetilde{E_m}\rvert \proj_R\lvert \widetilde{E_n}\rangle \delta(E_n-E_m-\mathcal{E})$ in $\mathcal{E} = E_n-E_m$.
Thermal equilibrium corresponds to an approximate vanishing of this Fourier transform, or a degree of uniformity of $\Pi_R(\mathcal{E})$ when aggregated (summed) over intervals spanning energy scales comparable to $\Delta \mathcal{E} \sim 2\pi/t$. When such uniformity is observed for $t \lesssim T_{\obs}$ for some finite scale $T_{\obs}$, it does not generally imply any strong degree uniformity over scales $\Delta \mathcal{E} \ll 2\pi/T_{\obs}$ (corresponding to thermalization for $t \gg T_{\obs}$), where the distribution may still be widely fluctuating such that only a coarse-grained version for $\Delta \mathcal{E} \gtrsim 2\pi/T_{\obs}$ appears uniform. A limited degree of uniformity does follow~\cite{dynamicalqthermalization}, but this is not sufficient to limit the magnitude of the second term in Eq.~\eqref{eq:thermalizationenergyeigenbasisexpr} to smaller than $\Theta(1)$ values, unless its prefactor is $O(1/D^{3/2})$ rather than $1/D$ --- corresponding to highly mixed initial states rather than pure initial states, or to $D_{\sigma} = O(1)$ in Eq.~\eqref{eq:time_averaged_thermal_equilibirum_expr} (in addition to applying only to ``almost all times'').

This necessitates the OTOC-based considerations of this manuscript, as described in the main text, to obtain reasonable constraints on pure state thermalization at a given instant of time. However, one can also ask if OTOCs themselves can be constrained in this manner, particularly from the viewpoint of predicting the long-time behavior of OTOCs from finite-time observations of dynamics, which we will explore next for completeness without much in the way of gains.


\subsubsection{Fourth order correlators: independence from autocorrelators?}
\label{app:fourthordercorrelators}

A more nontrivial inner product is necessary to handle OTOCs. In fact, we will need to use an indefinite metric in the inner product. We switch to the following setting:
\begin{align}
    \hat{A}_{4}, \hat{B}_{4} &\in L(\mathcal{H})\otimes L(\mathcal{H}), \\
    \langle \hat{B}_4,\hat{A}_4\rangle_4 &= \frac{1}{D}\Tr_{\mathcal{H}\otimes\mathcal{H}}[\hat{B}_4^\dagger \hat{A}_4] \label{eq:OTOCinnerproduct}\\
    \mathcal{U}_t(\hat{A}_4) &= e^{-i\hat{H}t} \otimes e^{-i\hat{H} t} \hat{A}_{2c} e^{i\hat{H}t} \otimes e^{i\hat{H} t},
\end{align}
which is similar to the cloned case ($2c$), except that the inner product is now normalized by $1/D$ instead of $1/D^2$. This difference in normalization will prove very crucial.

Now, consider the superoperator
\begin{equation}
    \eta(\hat{A}_4) = \hat{\swap}\ \hat{A}_4,
\end{equation}
where $\hat{\swap} \in L(\mathcal{H})\otimes L(\mathcal{H})$ is the standard swap operator~\cite{NielsenChuang}, which acts as $\swap \lvert \psi\rangle \otimes \lvert \chi\rangle = \lvert \chi\rangle \otimes \lvert \psi\rangle$ for $\lvert \psi\rangle, \lvert \chi\rangle \in \mathcal{H}$. Note that $\eta^2 \equiv \eta \circ \eta = 1$, and
\begin{equation}
    \langle \eta(\hat{B}_4), \hat{A}_4\rangle = \langle \hat{B}_4, \eta(\hat{A}_4)\rangle,
\end{equation}
making $\eta$ self-adjoint with respect to this inner product. Therefore, $\eta$ is a Hermitian superoperator with eigenvalues $\pm 1$, and can function as an indefinite metric that takes our inner product space to a Krein space~\cite{KreinSpaces}.

This is of relevance because we can now write the OTOC, $G_{R\rho}^{(4)}(t)$, as follows:
\begin{equation}
    G_{R\rho}^{(4)} = \frac{1}{D_{\rho}}\Tr[\proj_R \otimes \proj_R \hat{\swap} \proj_{\rho}(t) \otimes \proj_{\rho}(t)] = D_{\sigma} \langle \eta(\proj_R \otimes \proj_R), \mathcal{U}_t(\proj_{\rho} \otimes \proj_{\rho})\rangle.
    \label{eq:OTOCinnerproductfullexpr}
\end{equation}
In this way, Theorem~\ref{thm:autocorrelatorimpliescorrelator_Gen} can also be formally applied to OTOCs. In practice, for the kind of application we have in mind, it is sufficient to consider the time average of this OTOC. This is because $G_{R\rho}^{(4)} \geq [G_{R\rho}^{(2)}]^2$, and this lower bound is also (nearly) the saturation value we are interested in; recall that we want to operate in a regime with small $\sigma_{R\rho}^2$, typically associated with large $D_{\sigma}$. In other words, if the time average of $G_{R\rho}^{(4)}(t)$ is seen to be nearly $D_S^{-2}$, then that is sufficient to argue that $G_{R\rho}^{(4)}(t) \simeq D_S^{-2}$ at almost all times, as will be formalized below.

Based on the above intuition, we pick our operators to be:
\begin{align}
    \hat{A}_4 &= \proj_{\rho} \otimes \proj_{\rho} - \frac{1}{D_\sigma^2} \idop\otimes \idop, \nonumber \\
    \hat{B}_4 &= \eta\left(\proj_{R} \otimes \proj_{R} - \frac{1}{D_S^2} \idop\otimes \idop\right).
\end{align}
Here, we note that at any given time,
\begin{equation}
    \langle \eta(\hat{B}_4), \mathcal{U}_t(\hat{A}_4)\rangle_4 = \frac{1}{D_{\sigma}}\left[G_{R\rho}^{(4)}(t)-\frac{1}{D_S^2}-\frac{1}{D_S D_{\sigma}}+\frac{1}{D_S^2 D_{\sigma}}\right],
\end{equation}
which precisely measures the deviation of the OTOC from its saturation value. Moreover, we have
\begin{align}
    \langle \hat{A}_4, \mathcal{U}_t(\hat{A}_4)\rangle_4 &= \frac{1}{D}\left(\Tr[\proj_{\rho}(t)\proj_{\rho}]\right)^2-\frac{2D_{\rho}^2}{D D_{\sigma}^2}+\frac{D}{D_{\sigma}^4} \nonumber \\
    &= \frac{D_{\rho}}{D_{\sigma}}\left[\left(\frac{1}{D_{\rho}}\Tr[\proj_{\rho}(t)\proj_{\rho}]\right)^2-\frac{1}{D_{\sigma}^2}\right]
\end{align}
and
\begin{align}
    \langle \eta(\hat{B}_4), \eta(\hat{B_4})\rangle_4 = \frac{D_R(D_S^2-1)}{D_S^3}.
\end{align}

Applying Theorem~\ref{thm:autocorrelatorimpliescorrelator_Gen} to this set of choices of operators, inner product, and dynamics, we obtain
\begin{align}
    \frac{1}{D_{\sigma}} &\int\diff t\ w(t)  \left[G_{R\rho}^{(4)}(t)-\frac{1}{D_S^2}-\frac{1}{D_S D_{\sigma}}+\frac{1}{D_S^2 D_{\sigma}}\right] \nonumber \\
    &\leq \left[\sqrt{\frac{D_{\rho}}{D_{\sigma}W}\int\diff t\ w_+(t)\left(\left[G_{\rho \rho}^{(2)}(t)\right]^2-\frac{1}{D_{\sigma}^2}\right)} + w_0\sqrt{\frac{D_{\rho}(D_{\sigma}^2-1)}{D_{\sigma}^3}}\right]\sqrt{\frac{D_R(D_S^2-1)}{D_S^3}}.
\end{align}
Rearranging factors, we get
\begin{align}
    &\int\diff t\ w(t)  \left[G_{R\rho}^{(4)}(t)-\frac{1}{D_S^2}-\frac{1}{D_S D_{\sigma}}+\frac{1}{D_S^2 D_{\sigma}}\right] \nonumber \\
    &\leq D_R\sqrt{1-\frac{1}{D_S^2}}\left[\sqrt{\frac{1}{W}\int\diff t\ w_+(t)\left(\left[G_{\rho \rho}^{(2)}(t)\right]^2-\frac{1}{D_{\sigma}^2}\right)} + w_0\sqrt{1-\frac{1}{D_{\sigma}^2}}\right].
\end{align}
For this to provide a useful bound on the time-averaged OTOC, we must be able to constrain
\begin{equation}
    \left\lvert \left[G_{\rho \rho}^{(2)}(t)\right]^2-\frac{1}{D_{\sigma}^2}\right\rvert \ll \frac{1}{D_R^2},
\end{equation}
However, the size of these fluctuations are estimated at $D_{\sigma}/D^2$ even for Haar unitaries in Appendix~\ref{app:Haar}, and therefore the above constraint cannot generically be satisfied (note that we would generally like $D_\sigma \gg D_S$ as discussed in Sec.~\ref{sec:manybodysystems}, while the above condition requires $D_{\sigma} \ll D_S^2$, which can be compatible only if $D_S \gg 1$), much less measured in any case. Moreover, as $G_{\rho\rho}^{(2)}(0) = 1$, the time interval of observation must scale with $D$ for the time-averaged autocorrelator to be as small as $D_R^2$. This supports the expectation that OTOCs (at least those with $D_S \sim 1$) are genuinely independent of second order functions (especially over finite times) in generic systems (mathematically, this is due to the difference in normalization in Eq.~\eqref{eq:OTOCinnerproduct} as opposed to Eq.~\eqref{eq:clonedinnerproduct}, which reflects the appropriate accessible ranges of these quantities).
For predicting the behavior of OTOCs at long times from finite time observations, this suggests a need for entirely different (and perhaps model-dependent) accessible techniques, rather than Theorem~\ref{thm:autocorrelatorimpliescorrelator_Gen} that works best for second order correlators of high-dimensional subspaces.

\section{Haar averages and fluctuations}
\label{app:Haar}

\subsection{Low-order Haar averages}
\label{app:Weingarten}

The Haar averages in Eq.~\eqref{eq:Haar_mean_G} are straightforward to evaluate for $n=1,2$ using standard Haar integrals~\cite{Weingarten1, Weingarten2, ChaosComplexityDesign, ChaosComplexityRMT, CotlerHunterJones2}. With $\int_{U(D)}\diff \hat{U}$ representing an integration using the Haar measure on $U(D)$, we have
\begin{align}
    \int_{U(D)}\diff \hat{U}\ (\hat{U})_{a_1 b_1} (\hat{U}^\dagger)_{b_1' a_1'} =&\ \frac{1}{D} \delta_{a_1 a_1'} \delta_{b_1 b_1'}, \label{eq:1design}  \\
    \int_{U(D)}\diff \hat{U}\ (\hat{U})_{a_1 b_1} (\hat{U})_{a_2 b_2} (\hat{U}^\dagger)_{b_1' a_1'} (\hat{U}^\dagger)_{b_2' a_2'} =&\ \frac{1}{D^2-1}\left[\vphantom{{1}{D}}\delta_{a_1 a_1'} \delta_{a_2 a_2'} \delta_{b_1 b_1'} \delta_{b_2 b_2'} + \delta_{a_1 a_2'} \delta_{a_2 a_1'} \delta_{b_1 b_2'} \delta_{b_2 b_1'}\right. \nonumber \\
    &\left.- \frac{1}{D}\delta_{a_1 a_2'} \delta_{a_2 a_1'} \delta_{b_1 b_1'} \delta_{b_2 b_2'} -\frac{1}{D} \delta_{a_1 a_1'} \delta_{a_2 a_2'} \delta_{b_1 b_2'} \delta_{b_2 b_1'}\right]. \label{eq:2design}
\end{align}
In the latter expression \eqref{eq:2design}, we note that the two terms on the second line cannot be dismissed as being subleading just because of the $1/D$ prefactor; one of them will in fact end up contributing to the $n=2$ average to leading order. Before proceeding with the integrals, we recall that $\Tr\left[\proj_{R1}\right] = D_R$, $\Tr\left[\proj_{\rho 1}\right] = D_{\rho}$, $D_S = D/D_R$, $D_{\sigma} = D/D_{\rho}$.

For $n=1$, performing the Haar integral over $\hat{V} \in U(D)$, we have
\begin{align}
    \left\langle G_{R\rho}^{(2)}\right\rangle_{\hat{V},\text{CUE}} &= \frac{1}{D_{\rho}} \int_{U(D)}\diff \hat{V}\ \Tr\left[\proj_{R 1} \hat{V} \proj_{\rho 1} \hat{V}^\dagger\right] \nonumber \\
    &= \frac{1}{D_{\rho} D} \Tr\left[\proj_{R1}\right] \Tr\left[\proj_{\rho 1}\right] \nonumber \\
    &= \frac{1}{D_S}.
    \label{eq:G2average}
\end{align}

The Haar integral for $n=2$ gives, using $\proj^2 = \proj$ for both projectors:
\begin{align}
    \left\langle G_{R\rho}^{(4)} \right\rangle_{\hat{V},\text{CUE}} =&\ \frac{1}{D_{\rho}}\int_{\unitarygroup(D)}\diff \hat{V}\ \Tr\left[\proj_{R 1} \hat{V} \proj_{\rho 1} \hat{V}^\dagger\proj_{R 1} \hat{V} \proj_{\rho 1} \hat{V}^\dagger\right] \nonumber \\
    =&\ \frac{1}{D_{\rho}(D^2-1)}\left\lbrace \vphantom{\frac{1}{D}} \Tr\left[\proj_{R1}\right] \left(\Tr\left[\proj_{\rho 1}\right]\right)^2 + \left(\Tr\left[\proj_{R1}\right]\right)^2 \Tr\left[\proj_{\rho 1}\right]  \right. \nonumber \\ &\left. -\frac{1}{D}\left(\Tr\left[\proj_{R1}\right]\right)^2\left(\Tr\left[\proj_{\rho 1}\right]\right)^2 - \frac{1}{D}\Tr\left[\proj_{R1}\right]\Tr\left[\proj_{\rho 1}\right]\right\rbrace \nonumber \\
    =&\ \frac{D^2}{D^2-1}\left\lbrace \frac{1}{D_S D_{\sigma}} + \frac{1}{D_S^2}-\frac{1}{D_S^2 D_{\sigma}} - \frac{1}{D D_S D_{\sigma}}\right\rbrace.
    \label{eqs:G4average}
\end{align}
In the large $D$ limit with fixed $D_S, D_{\sigma}$, we note that (the first) three of the four terms in braces are finite, including one of the terms with a $1/D$ coefficient in Eq.~\eqref{eq:2design}.

We can also evaluate the fluctuations for $n=1$, for which we get
\begin{align}
    \left\langle \left(G_{R\rho}^{(2)}\right)^2 \right\rangle_{\hat{V},\text{CUE}} =&\ \frac{1}{D_{\rho}^2}\int_{\unitarygroup(D)}\diff \hat{V}\ \Tr\left[\proj_{R 1} \hat{V} \proj_{\rho 1} \hat{V}^\dagger\right] \Tr\left[\proj_{R 1} \hat{V} \proj_{\rho 1} \hat{V}^\dagger\right] \nonumber \\
    =&\ \frac{1}{D_{\rho}^2(D^2-1)}\left\lbrace \vphantom{\frac{1}{D}} \left(\Tr\left[\proj_{R1}\right]\right)^2 \left(\Tr\left[\proj_{\rho 1}\right]\right)^2 + \Tr\left[\proj_{R1}\right] \Tr\left[\proj_{\rho 1}\right]  \right. \nonumber \\ &\left. -\frac{1}{D} \Tr\left[\proj_{R1}\right] \left(\Tr\left[\proj_{\rho 1}\right]\right)^2 - \frac{1}{D}\left(\Tr\left[\proj_{R1}\right]\right)^2\Tr\left[\proj_{\rho 1}\right]\right\rbrace \nonumber \\
    =&\ \frac{D^2}{D^2-1}\left\lbrace \frac{1}{D_S^2} + \frac{D_{\sigma}}{D^2 D_S}-\frac{1}{D^2 D_S} - \frac{D_\sigma}{D^2 D_S^2}\right\rbrace.
\end{align}
Consequently, the variance between unitaries is given by:
\begin{equation}
    \left\langle \left(G_{R\rho}^{(2)}\right)^2\right\rangle_{\hat{V},\text{CUE}} - \left\langle G_{R\rho}^{(2)}\right\rangle_{\hat{V},\text{CUE}}^2 = \frac{(D_{\sigma}-1)(D_S-1)}{(D^2-1)D_S^2},
    \label{eq:WeingartenG2fluctuationsapp}
\end{equation}
which corresponds to a standard deviation of $\sim O(1/D)$ in $G_{R{\rho}}^{(2)}$ for fixed $D_S$, $D_{\sigma}$, which vanishes as $D\to\infty$. A comparable estimate of fluctuations for $G_{R\rho}^{(4)}$ is tedious due to involving a tensor product of $8$ Haar unitaries, and we will therefore constrain these fluctuations through concentration-of-measure bounds for the Haar measure.

\subsection{Constraining fluctuations via concentration of measure inequalities}
\label{app:concentration}

The ``universality'' properties of the Haar measure, by which virtually any ``sufficiently smooth'' observable in a typical quantum system (i.e. with some property varied by the application of CUE unitaries) agrees with the Haar average, is due to its concentration property\footnote{This concentration property is more powerful than the Haar measure. From the standpoint of ``quantum chaos'', even measure zero regions with respect to the Haar measure tend to agree with Haar averages, i.e., such properties tend to be universal even in smaller families of systems than the Haar measure is sensitive to. Even in less typical systems, many features are compatible with a simple modification (a modulation by some macroscopic quantities) of Haar behavior~\cite{DAlessio2016}. The power of (modulated) typicality has sometimes led to the attribution of the occurrence of phenomena associated with complexity to the general emergence of Haar random behavior~\cite{deutsch1991eth, srednicki1994eth, DAlessio2016, ChaosComplexityRMT}. In this context, we prefer to emphasize a logical separation between the (modulated) typicality phenomenon (which may yet determine which forms of complexity are ubiquitous and therefore of more practical interest) and direct mechanisms that link different forms of complex quantum dynamics and other characteristic features of complexity in an individual system (which can inform rigorous computational modes of \textit{predicting} the behavior of these systems, as in this work).}. We will quantify this according to the formulation in Ref.~\cite{HaarBook}.

First, we will specialize to $SU(D)$, the subset of $U(D)$ with $\det \hat{U} = 1$ for $\hat{U} \in SU(D)$. The less restrictive condition $\left\lvert \det \hat{U}\right\rvert = 1$ satisfied by $\hat{U} \in U(D) = U(1) \times SU(D)$ allows an overall phase $\phi \in [0,2\pi)$ via $\hat{U} = e^{i\phi} \hat{U}_S$ for $\hat{U} \in U(D)$ and $\hat{U}_S \in SU(D)$; but our correlators $G_{R\rho}^{(2n)}(\hat{U}) = G_{R\rho}^{(2n)}(\hat{U}_S)$ (with $\hat{U}$ taking the place of $\hat{V}$ in Eq.~\eqref{eq:correlator_randomization_def}) do not depend on this phase, justifying this restriction. The practical reason for this restriction is that the concentration bound satisfied by $SU(D)$ is stronger than $U(D)$ by a small constant~\cite[Theorem 5.16]{HaarBook} (using concentration for $U(D)$ would lead to a weaker bound by a factor of $3$ in Eqs.~\eqref{eq:G2fluctuations} and \eqref{eq:G4fluctuations} below).

If $F: SU(D) \to \mathbb{R}$ is a real-valued function of a unitary, $F$ being ``sufficiently smooth'' is formalized by the criterion of Lipschitz continuity with Lipschitz constant $L \geq 0$:
\begin{equation}
    \left\lvert F(\hat{U})-F(\hat{V})\right\rvert \leq L \sqrt{\Tr\left[(\hat{U}-\hat{V})^\dagger (\hat{U}-\hat{V})\right]},\;\;\ \text{ for all } \hat{U},\hat{V} \in U(D).
    \label{eq:Lipschitzdef}
\end{equation}
Intuitively, the right hand side measures the Hilbert-Schmidt or Frobenius norm~\cite{HaarBook} of $\hat{U}-\hat{V}$ scaled by $L$, and therefore requires that the difference in $F$ between the two points is bounded by a specific scaled distance between the two unitaries.

In this case, the concentration property of the Haar measure on $SU(D)$ takes the following form:
\begin{lemma}[Concentration of Haar measure on $SU(D)$]
\label{lem:Haarconcentration}
If $F: SU(D) \to \mathbb{R}$ is Lipschitz continuous with Lipschitz constant $L \geq 0$ [i.e. $F$ satisfies Eq.~\eqref{eq:Lipschitzdef}], then the Haar-measure-induced probability\footnote{In our notation, given a conditional statement $\mathcal{C}[\hat{U}]$ on $\hat{U}$, the probability that this condition is satisfied according to the Haar measure on $SU(D)$ is given by \begin{equation*}
    \mathbb{P}_{\hat{U}}\left(\mathcal{C}[\hat{U}]\right) = \int_{SU(D)}\diff \hat{U}\ \ifc\left(\mathcal{C}[\hat{U}]\right),
\end{equation*}
where $\ifc(\mathcal{C}) = 1$ if $\mathcal{C}$ is true and $\ifc(\mathcal{C}) = 0$ if $\mathcal{C}$ is false.
} $\mathbb{P}_{\hat{U}}$ over $\hat{U} \in SU(D)$ that $F(\hat{U})$ differs from its Haar average by at least $\epsilon$ is constrained by:
\begin{equation}
        \mathbb{P}_{\hat{U}}\left(\left\lvert F(\hat{U}) - \int_{SU(D)}\diff \hat{V} F(\hat{V}) \right\rvert \geq \epsilon \right) \leq 2\exp\left(- \frac{D \epsilon^2}{4 L^2}\right).
        \label{eq:concentrationforSUn}
\end{equation}
\end{lemma}
\begin{proof}
    See, e.g., Ref.~\cite{HaarBook}; this follows specifically by using Proposition 5.13 in Theorem 5.12, and then using Theorem 5.5, all in that reference.
\end{proof}
For any (and all)
\begin{equation}
    \epsilon \gg L\sqrt{2/D},
    \label{eq:concentrationvarianceconstraint}
\end{equation}
the (cumulative) probability of larger fluctuations is therefore vanishingly small, with the right hand side above formally representing the standard deviation of $\epsilon$ in the Gaussian that appears in Eq.~\eqref{eq:concentrationforSUn}.

For $F = G_{R\rho}^{(2)}$, we have
\begin{align}
    \left\lvert G_{R\rho}^{(2)}(\hat{U})-G_{R\rho}^{(2)}(\hat{V})\right\rvert =&\ \frac{1}{D_{\rho}}\left\lvert \Tr\left[\proj_{R 1} \hat{U} \proj_{\rho 1} \hat{U}^\dagger\right] -  \Tr\left[\proj_{R 1} \hat{V} \proj_{\rho 1} \hat{V}^\dagger\right]\right\rvert \nonumber \\
    =&\ \frac{1}{D_{\rho}}\left\lvert \Tr\left[\proj_{R 1} (\hat{U}-\hat{V}) \proj_{\rho 1} \hat{U}^\dagger\right]+\Tr\left[\proj_{R 1} \hat{V} \proj_{\rho 1} (\hat{U}-\hat{V})^\dagger\right]\right\rvert \nonumber \\
    \leq&\ \frac{1}{D_{\rho}}\sqrt{\Tr\left[(\hat{U}-\hat{V})^\dagger (\hat{U}-\hat{V})\right]\Tr\left[\proj_{\rho 1} \hat{U}^\dagger\proj_{R1}\hat{U}\right]} \nonumber \\
     &+ \frac{1}{D_{\rho}}\sqrt{\Tr\left[(\hat{U}-\hat{V})^\dagger (\hat{U}-\hat{V})\right]\Tr\left[\proj_{\rho 1} \hat{V}^\dagger\proj_{R1}\hat{V}\right]} \nonumber \\
     \leq&\ \frac{2}{\sqrt{D_{\rho}}}\sqrt{\Tr\left[(\hat{U}-\hat{V})^\dagger (\hat{U}-\hat{V})\right]},
     \label{eq:G2Lipschitz}
\end{align}
where we have used the triangle and Cauchy-Schwarz inequalities, and that $0 \leq G_{R\rho}^{(2)} \leq 1$. For this function, we therefore obtain a Lipschitz constant of $L = 2/\sqrt{D_{\rho}}$, which by Eq.~\eqref{eq:concentrationvarianceconstraint} constrains fluctuations with probability nearly $1$ to
\begin{equation}
    \epsilon_2 \lesssim 2\sqrt{\frac{2}{D D_{\rho}}} = \frac{2\sqrt{2D_{\sigma}}}{D}.
    \label{eq:G2fluctuations}
\end{equation}
We note that this is comparable in scale to Eq.~\eqref{eq:WeingartenG2fluctuationsapp}, obtained by evaluating the Haar integral, but for the $D_S$-dependence (which typically makes Eq.~\eqref{eq:WeingartenG2fluctuationsapp} tighter, especially for large $D_S$).

Now, let us apply a similar strategy to to $G^{(4)}_{R\rho}$. Our derivation will largely parallel those of Refs.~\cite{SAZ, DKM} for other types of OTOCs, but with some quantitative and technical differences, and somewhat increased brevity. For $F = G_{R\rho}^{(4)}$, we can apply a similar strategy as in Eq.~\eqref{eq:G2Lipschitz}, i.e. introduce differences of $\hat{U}-\hat{V}$ for each factor and collect terms, and then apply the triangle and Cauchy-Schwarz inequalities to get:
\begin{align}
    \left\lvert G_{R\rho}^{(4)}(\hat{U})-G_{R\rho}^{(4)}(\hat{V})\right\rvert =&\ \frac{1}{D_{\rho}}\left\lvert \Tr\left[\proj_{R 1} (\hat{U}-\hat{V}) \proj_{\rho 1} \hat{U}^\dagger\proj_{R 1} \hat{U} \proj_{\rho 1} \hat{U}^\dagger\right] + \Tr\left[\proj_{R 1} \hat{V} \proj_{\rho 1} (\hat{U}-\hat{V})^\dagger\proj_{R 1} \hat{U} \proj_{\rho 1} \hat{U}^\dagger\right]\right. \nonumber \\
    &\left. + \Tr\left[\proj_{R 1} \hat{V} \proj_{\rho 1} \hat{V}^\dagger\proj_{R 1} (\hat{U}-\hat{V}) \proj_{\rho 1} \hat{U}^\dagger\right] + \Tr\left[\proj_{R 1} \hat{V} \proj_{\rho 1} \hat{V}^\dagger\proj_{R 1} \hat{V} \proj_{\rho 1} (\hat{U}-\hat{V})^\dagger\right]\right\rvert \nonumber \\
    \leq&\ \frac{4}{\sqrt{D_{\rho}}}\sqrt{\Tr\left[(\hat{U}-\hat{V})^\dagger (\hat{U}-\hat{V})\right]},
\end{align}
where we have also used bounds such as (here, for a factor that emerges subsequent to the use of the Cauchy-Schwarz inequality for the first term above)
\begin{equation}
    0\leq \frac{1}{D_{\rho}}\Tr\left[\proj_{R1} \hat{V} \proj_{\rho 1} \hat{V}^\dagger\proj_{R 1} \hat{V} \proj_{\rho 1} \hat{V}^\dagger\proj_{R 1} \hat{V} \proj_{\rho 1} \hat{V}^\dagger \proj_{R1}\right] \leq 1,
\end{equation}
as the trace can be viewed as that of $\proj_{\rho 1}$ subject to several unitary rotations and projections, which collectively form a trace-non-increasing quantum operation~\cite{NielsenChuang} with Kraus operator $\hat{\mathcal{K}} = \proj_{R1} \hat{V} \proj_{\rho 1} \hat{V}^\dagger\proj_{R 1} \hat{V}$. Consequently, $G_{R\rho}^{(4)}$ is Lipschitz continuous with Lipschitz constant $L = 4/\sqrt{D_{\rho}}$, which by Eq.~\eqref{eq:concentrationvarianceconstraint} restricts fluctuations to be of the scale
\begin{equation}
    \epsilon_4 \lesssim \frac{4\sqrt{2D_{\sigma}}}{D},
    \label{eq:G4fluctuations}
\end{equation}
with probability nearly $1$. As an aside, from the structure of the derivation of Lipschitz constants above (there are as many terms bounded by $1$ after the use of the triangle inequality as the number of Haar random unitaries in the correlator), we can anticipate that $L_{2n} = 2n/\sqrt{D_{\rho}}$ giving $\epsilon_{2n} \lesssim 2n\sqrt{2D_{\sigma}}/D$ for higher order correlators. In Eq.~\eqref{eq:generalcorrelators_concentration} in the main text, we quote this result with an additional large constant $\kappa$ to account for ``most of the probability'' of fluctuations to account for the $\gg$ sign in Eq.~\eqref{eq:concentrationvarianceconstraint}.

Thus, from Eqs.~\eqref{eq:G2average}, \eqref{eqs:G4average}, \eqref{eq:WeingartenG2fluctuationsapp} or \eqref{eq:G2fluctuations}, and \eqref{eq:G4fluctuations}, we are led to the estimate that typically,
\begin{equation}
    \sigma_{R\rho}^2 = G_{R\rho}^{(4)}-\left(G_{R\rho}^{(2)}\right)^2 \simeq \frac{D_S-1}{D_S^2 D_{\sigma}} + O(D^{-1}D_{\sigma}^{1/2}),
    \label{eq:variancefluctuationsapp}
\end{equation}
with the fluctuations indicated by the second term being strongly suppressed in the $D\to\infty$ limit, for fixed $D_S$ and $D_{\sigma}$. This ultimately justifies our procedure for controlling the size of this variance, in typical cases, by increasing $D_{\sigma}$ to reduce the fluctuations of the principal angles.

\section{Proofs}
\label{app:proofs}

\subsection{Theorem~\ref{thm:quantumthermalization_subspaces}: Quantum thermalization in pure states in nearly aligned subspaces}
\label{proof:quantumthermalization_subspaces}
From Eq.~\eqref{eq:anglevariance_def}, the fraction $\varphi_{\lambda}$ of the $D_{\rho}$ principal angles $\theta_k$ for which $\left\lvert \cos^2\theta_k-G_{R\rho}^{(2)}\right\rvert > \lambda$, i.e.
    \begin{equation}
        \varphi_{\lambda} \equiv \frac{1}{D_{\rho}}\sum_k \ifc\left(\left\lvert \cos^2\theta_k-G_{R\rho}^{(2)}\right\rvert > \lambda\right),
        \label{eq:fractiondef_inproof}
    \end{equation}
    satisfies, by Markov's inequality~\cite{RossProbability} (i.e. the angles in the sum above contribute at least $\lambda^2$ to $\sigma_{R\rho}^2$):
    \begin{equation}
        \varphi_{\lambda} \leq \frac{\sigma_{R\rho}^2}{\lambda^2}.
        \label{eq:fractionineq_inproof}
    \end{equation}
    Define the subspace $\mathcal{H}_{\therm}(\lambda) \subseteq \mathcal{H}_{\rho}$ spanned by the remaining principal axes $\lvert w_k\rangle$ [defined in Eq.~\eqref{eq:principalaxesdef}]:
    \begin{equation}
        \mathcal{H}_{\therm}(\lambda) \equiv \vectorspan\left\lbrace\lvert w_k\rangle:  \left\lvert \cos^2\theta_k-G_{R\rho}^{(2)}\right\rvert \leq \lambda\right\rbrace.
        \label{eq:Hthermdef_inproof}
    \end{equation}
    By Eq.~\eqref{eq:fractiondef_inproof}, $\dim \mathcal{H}_{\therm}(\lambda) = D_{\rho}(1-\varphi_{\lambda})$, which implies Eq.~\eqref{eq:dimrestriction1} of Theorem~\ref{thm:quantumthermalization_subspaces} by Eq.~\eqref{eq:fractionineq_inproof}.
    Further, any (non-null) vector $\lvert \psi\rangle \in \mathcal{H}_{\therm}(\lambda)$ can be expanded in the $\lvert w_k\rangle$ basis within this subspace:
    \begin{equation}
        \lvert \psi\rangle = \sum_{k: \lvert w_k\rangle \in \mathcal{H}_{\therm}(\lambda)} c_k \lvert w_k\rangle.
    \end{equation}
    We can therefore write
    \begin{align}
        \langle \psi\vert \proj_R\lvert \psi\rangle &= \sum_{\substack{k: \lvert w_k\rangle \in \mathcal{H}_{\therm}(\lambda),\\ \ell: \lvert w_\ell\rangle \in \mathcal{H}_{\therm}(\lambda)}} c_k^\ast c_{\ell} \langle w_k\vert \proj_R\lvert w_{\ell}\rangle \nonumber \\
        &= \sum_{\substack{k: \lvert w_k\rangle \in \mathcal{H}_{\therm}(\lambda),\\ \ell: \lvert w_\ell\rangle \in \mathcal{H}_{\therm}(\lambda)}} c_k^\ast c_{\ell} \sum_r \langle w_k\vert u_r\rangle\langle u_r\vert w_{\ell}\rangle \nonumber \\
        &= \sum_{k: \lvert w_k\rangle \in \mathcal{H}_{\therm}(\lambda)} \lvert c_k\rvert^2 \cos^2\theta_k,
    \end{align}
    where we have used Lemma~\ref{lem:Halmos} in the second line. Using the triangle inequality~\cite{ByronFuller} to obtain a sum of nonnegative terms and Eq.~\eqref{eq:Hthermdef_inproof}, i.e. that $\lvert \cos^2\theta_k - G_{R\rho}^{(2)}\rvert \leq \lambda$ in each term, we get
    \begin{equation}
        \left\lvert\langle \psi\vert \proj_R\lvert \psi\rangle - \sum_{k: \lvert w_k\rangle \in \mathcal{H}_{\therm}(\lambda)} \lvert c_k\rvert^2G_{R\rho}^{(2)}\right\rvert \leq \lambda \sum_{k: \lvert w_k\rangle \in \mathcal{H}_{\therm}(\lambda)} \lvert c_k\rvert^2,
        \label{eq:subspacethermalization_inproof}
    \end{equation}
    which, on identifying $\langle \psi\vert\psi\rangle = \sum_{k: \lvert w_k\rangle \in \mathcal{H}_{\therm}(\lambda)} \lvert c_k\rvert^2$, gives Eq.~\eqref{eq:subspacethermalizationcondition} of Theorem~\ref{thm:quantumthermalization_subspaces} and establishes the first part of the theorem.

    For the second part, it is convenient to first assume an unspecified parameter $\widetilde{\lambda} \in (0,1)$ in place of $\lambda$ above. Given any orthonormal basis $\mathcal{B}_{\rho} = \lbrace \lvert b_k\rangle_{\rho}\rbrace$ for $\mathcal{H}_{\rho}$ (with $\langle b_k\vert b_{\ell}\rangle = \delta_{k \ell}$), and the projector $\proj_{\therm}(\widetilde{\lambda}) = \sum_{k: \lvert w_k\rangle \in \mathcal{H}_{\therm}(\widetilde{\lambda})} \lvert w_k\rangle \langle w_k\rvert$ that projects onto $\mathcal{H}_{\therm}(\widetilde{\lambda})$, we can decompose each basis vector into a part inside $\mathcal{H}_{\therm}(\widetilde{\lambda})$ and an orthogonal part outside this subspace, while keeping all our considerations within $\mathcal{H}_{\rho}$ (with $\idop_{\rho}$ being the identity operator on $\mathcal{H}_{\rho}$, which may be replaced by $\proj_{\rho}$):
    \begin{equation}
        \lvert b_k\rangle = \proj_{\therm}(\widetilde{\lambda})\lvert b_k\rangle + (\idop_{\rho}-\proj_{\therm}(\widetilde{\lambda}))\lvert b_k\rangle.
    \end{equation}
    We will now show that due to Lemma~\ref{lem:Halmos}, this decomposition retains its orthogonal structure when projected onto $\mathcal{H}_R$, allowing us to separate the thermal and non-thermal contributions to the expectation value of $\proj_R$ from each vector, and constrain each such contribution. From the triangle inequality, we have
    \begin{align}
        \left\lvert\langle b_k\rvert \proj_R\lvert b_k\rangle-G_{R\rho}^{(2)}\right\rvert &\leq \left\lvert \langle b_k\rvert \proj_{\therm}(\widetilde{\lambda}) [\proj_R-G_{R\rho}^{(2)}\idop_{\rho}] \proj_{\therm}(\widetilde{\lambda})\lvert b_k\rangle\right\rvert \nonumber \\
        &+ 2 \left\lvert \langle b_k\rvert \proj_{\therm}(\widetilde{\lambda})[\proj_R-G_{R\rho}^{(2)}\idop_{\rho}][\idop_{\rho}-\proj_{\therm}(\widetilde{\lambda})]\lvert b_k\rangle\right\rvert \nonumber \\
        &+ \left\lvert\langle b_k\rvert[\idop_{\rho}-\proj_{\therm}(\widetilde{\lambda})] [\proj_R-G_{R\rho}^{(2)}\idop_{\rho}] [\idop_{\rho}-\proj_{\therm}(\widetilde{\lambda})]\lvert b_k\rangle\right\rvert.
    \end{align}
    It is convenient to consider each term separately:
    \begin{enumerate}
        \item For the first term, by Eq.~\eqref{eq:subspacethermalization_inproof}, we have (using $\proj_{\therm}(\widetilde{\lambda})^2 = \proj_{\therm}(\widetilde{\lambda})$)
        \begin{equation}
            \left\lvert \langle b_k\rvert \proj_{\therm}(\widetilde{\lambda}) [\proj_R-G_{R\rho}^{(2)}\idop_{\rho}] \proj_{\therm}(\widetilde{\lambda})\lvert b_k\rangle\right\rvert \leq \widetilde{\lambda} \langle b_k\rvert \proj_{\therm}(\widetilde{\lambda})\lvert b_k\rangle.
        \end{equation}
        \item For the second term, as $\proj_{\therm}(\widetilde{\lambda})$ and $\idop_{\rho}-\proj_{\therm}(\widetilde{\lambda})$ project onto orthogonal subspaces, the part inside the absolute value signs simplifies to:
    \begin{equation}
        \langle b_k\rvert \proj_{\therm}(\widetilde{\lambda})[\proj_R-G_{R\rho}^{(2)}\idop_{\rho}][\idop_{\rho}-\proj_{\therm}(\widetilde{\lambda})]\lvert b_k\rangle = \langle b_k\rvert \proj_{\therm}(\widetilde{\lambda})\proj_R [\idop_{\rho}-\proj_{\therm}(\widetilde{\lambda})]\lvert b_k\rangle.
    \end{equation}
    Using Lemma~\ref{lem:Halmos} and Eq.~\eqref{eq:principalaxesdef}, we can write (recalling that $\idop_{\rho}$ amounts to $\proj_{\rho}$)
    \begin{align}
       \proj_{\therm}(\widetilde{\lambda})\proj_R [\idop_{\rho}-\proj_{\therm}(\widetilde{\lambda})] &= \sum_{\substack{k: \lvert w_k\rangle \in \mathcal{H}_{\therm}(\widetilde{\lambda}),\\ \ell: \lvert w_\ell\rangle \notin \mathcal{H}_{\therm}(\widetilde{\lambda})}} \sum_r \lvert w_k\rangle\langle w_k\vert u_r\rangle \langle u_r\vert w_\ell\rangle \langle w_{\ell}\rvert \nonumber \\
       &= \sum_{\substack{k: \lvert w_k\rangle \in \mathcal{H}_{\therm}(\widetilde{\lambda}),\\ \ell: \lvert w_\ell\rangle \notin \mathcal{H}_{\therm}(\widetilde{\lambda})}} \delta_{k\ell} \cos^2\theta_k \lvert w_k\rangle \langle w_{\ell}\rvert = 0,
    \end{align}
    due to which the second term is identically zero.
    \item For the third term, we can use the fact that $[\proj_R-G_{R{\rho}}\idop_{\rho}]$ has eigenvalues in $[-1,1]$ (as $G_{R\rho}^{(2)} \in [0,1]$) to write, for any vector $\lvert v\rangle$:
    \begin{equation}
        \left\lvert \langle  v\rvert [\proj_R-G_{R{\rho}}\idop_{\rho}]\lvert v\rangle\right\rvert \leq \langle v\vert v\rangle.
    \end{equation}
    Choosing $\lvert v\rangle = [\idop_{\rho}-\proj_{\therm}(\widetilde{\lambda})]\lvert b_k\rangle$, this gives:
    \begin{equation}
        \left\lvert\langle b_k\rvert[\idop_{\rho}-\proj_{\therm}(\widetilde{\lambda})] [\proj_R-G_{R\rho}^{(2)}\idop_{\rho}] [\idop_{\rho}-\proj_{\therm}(\widetilde{\lambda})]\lvert b_k\rangle\right\rvert \leq \langle b_k\rvert (\idop_{\rho}-\proj_{\therm}(\widetilde{\lambda}))\lvert b_k\rangle,
    \end{equation}
    where we have also used $[\idop_{\rho}-\proj_{\therm}(\widetilde{\lambda})]^2 = [\idop_{\rho}-\proj_{\therm}(\widetilde{\lambda})]$.
    \end{enumerate}
    Combining these simplifications, we get
    \begin{equation}
        \left\lvert\langle b_k\rvert \proj_R\lvert b_k\rangle-G_{R\rho}^{(2)}\right\rvert \leq \widetilde{\lambda} \langle b_k\rvert \proj_{\therm}(\widetilde{\lambda})\lvert b_k\rangle + \langle b_k\rvert (\idop_{\rho}-\proj_{\therm}(\widetilde{\lambda}))\lvert b_k\rangle.
        \label{eq:projRinequalityintermediate_inproof}
    \end{equation}
    Now, if we sum over all the $\lvert b_k\rangle$ (amounting to a trace within $\mathcal{H}_{\rho}$), we obtain:
    \begin{align}
        \sum_k \langle b_k\rvert \proj_{\therm}(\widetilde{\lambda})\lvert b_k\rangle &= D_{\rho}(1-\varphi_{\widetilde{\lambda}}), \\
        \sum_k \langle b_k\rvert (\idop_{\rho}-\proj_{\therm}(\widetilde{\lambda}))\lvert b_k\rangle &= D_{\rho}\varphi_{\widetilde{\lambda}},
    \end{align}
    due to which we can write, from Eq.~\eqref{eq:projRinequalityintermediate_inproof} and then using Eq.~\eqref{eq:fractionineq_inproof} and \eqref{eq:Hthermdef_inproof},
    \begin{equation}
        \frac{1}{D_{\rho}}\sum_k \left\lvert\langle b_k\rvert \proj_R\lvert b_k\rangle-G_{R\rho}^{(2)}\right\rvert = \widetilde{\lambda}(1-\varphi_{\widetilde{\lambda}}) + \varphi_{\widetilde{\lambda}} \leq \widetilde{\lambda} + (1-\widetilde{\lambda})\frac{\sigma_{R\rho}^2}{\widetilde{\lambda}^2}.
    \end{equation}
    Note that this holds for any choice of $0 < \widetilde{\lambda} < 1$. The minimization of the right hand side with respect to $\widetilde{\lambda}$ to obtain the tightest such bound is straightforward but tedious (involving solutions to a cubic equation). To derive a more palatable final expression, at the expense of some tightness, we will simplify the above inequality by using $(1-\widetilde{\lambda}) \leq 1$, giving
    \begin{equation}
        \frac{1}{D_{\rho}}\sum_k \left\lvert\langle b_k\rvert \proj_R\lvert b_k\rangle-G_{R\rho}^{(2)}\right\rvert \leq \widetilde{\lambda} +\frac{\sigma_{R\rho}^2}{\widetilde{\lambda}^2}.
    \end{equation}
    This is minimized for $\widetilde{\lambda} = (2\sigma_{R\rho}^2)^{1/3}$, which implies that
    \begin{equation}
        \frac{1}{D_{\rho}}\sum_k \left\lvert\langle b_k\rvert \proj_R\lvert b_k\rangle-G_{R\rho}^{(2)}\right\rvert \leq 3 \left(\frac{\sigma_{R\rho}^2}{4}\right)^{1/3}.
    \end{equation}
    Now, defining the fraction $f_{\lambda}$ of basis states in which the expectation value of $\proj_R$ deviates from $G_{R\rho}^{(2)}$ more than $\lambda$ as in Eq.~\eqref{eq:nonthermalbasisstates_fraction}, we obtain using Markov's inequality~\cite{RossProbability}
    \begin{equation}
        f_{\lambda} \leq \frac{3}{\lambda}\left(\frac{\sigma_{R\rho}^2}{4}\right)^{1/3},
        \label{eq:nonthermalbasisstatesconstraint_inproof}
    \end{equation}
    which is Eq.~\eqref{eq:nonthermalbasisstatesconstraint} of Theorem~\ref{thm:quantumthermalization_subspaces}, completing the proof.
    
    As a sanity check, we note that Eq.~\eqref{eq:nonthermalbasisstatesconstraint_inproof} should not give a stronger bound than Eq.~\eqref{eq:fractionineq_inproof}, as we could have chosen $\lvert b_k\rangle = \lvert w_k\rangle$ and then Eq.~\eqref{eq:fractionineq_inproof} (which was used as an input in this bound) should apply, and presumably cannot by itself generate a stronger bound. We get that this holds as long as $\lambda \geq (2\sigma_{R\rho}^2)^{2/3}/3$, in which case the respective bounds are $f_{\lambda},\varphi_{\lambda} \leq (9/2) \sigma_{R\rho}^{-2/3}$; when $\lambda < (2\sigma_{R\rho}^2)^{2/3}/3$ which would na\"{i}vely violate this sanity condition, the above bounds on $f_{\lambda}, \varphi_{\lambda}$ are always weaker than the natural bounds $f_{\lambda}, \varphi_{\lambda} \leq 1$ (as $\sigma_{R\rho}^2 \leq 1$), preserving consistency.

\subsection{Theorem~\ref{thm:autocorrelatorimpliescorrelator_Gen}: Autocorrelator smallness predicts correlator smallness}
\label{proof:autocorrelatorimpliescorrelator_Gen}

 First, let us write the correlation function in terms of the Fourier transform of $w(t)$ using the eigendecomposition of $\mathcal{U}_t$ in Eq.~\eqref{eq:genunitary_eigendecomposition}:
    \begin{equation}
        \int\diff t\ w(t) \langle \hat{B}, \mathcal{U}_t(\hat{A})\rangle = \sum_k \widetilde{w}(\mathcal{E}_k) \langle \hat{B}, \mathcal{P}_k(\hat{A})\rangle.
    \end{equation}
    Let $\mathcal{K}_{\Delta E} = \lbrace k: \lvert \mathcal{E}_k\rvert < \Delta E\rbrace$ denote the set of energy levels inside the Fourier window $\Delta E$, and $\mathcal{K}_{\Delta E}^c = \lbrace k: \lvert \mathcal{E}_k\rvert \geq \Delta E\rbrace$ be its complementary set. Splitting the sum over $k$ on the right into these two sets and using the triangle inequality gives:
    \begin{equation}
        \left\lvert \int\diff t\ w(t) \langle \hat{B}, \mathcal{U}_t(\hat{A})\rangle\right\rvert \leq \left\lvert \sum_{k \in \mathcal{K}_{\Delta E}} \widetilde{w}(\mathcal{E}_k) \langle \hat{B}, \mathcal{P}_k(\hat{A})\rangle \right\rvert + \left\lvert \sum_{j \in \mathcal{K}_{\Delta E}^c} \widetilde{w}(\mathcal{E}_j) \langle \hat{B}, \mathcal{P}_j(\hat{A})\rangle \right\rvert.
    \end{equation}
 
    By linearity and the Cauchy-Schwarz inequality for the inner product $\langle \cdot, \cdot \rangle$, for a sum restricted to any set $\mathcal{K} \in \lbrace \mathcal{K}_{\Delta E}, \mathcal{K}_{\Delta E}^c \rbrace$ we have
    \begin{align}
        \left\lvert \sum_{k \in \mathcal{K}} \widetilde{w}(\mathcal{E}_k) \langle \hat{B}, \mathcal{P}_k(\hat{A})\rangle \right\rvert &= \left\lvert  \left\langle \hat{B}, \sum_{k \in \mathcal{K}} \widetilde{w}(\mathcal{E}_k) \mathcal{P}_k(\hat{A})\right\rangle \right\rvert \nonumber \\
        &\leq \sqrt{\langle \hat{B}, \hat{B}\rangle \left(\left\langle \sum_{k \in \mathcal{K}} \widetilde{w}(\mathcal{E}_k) \mathcal{P}_k(\hat{A}), \sum_{k \in \mathcal{K}} \widetilde{w}(\mathcal{E}_k) \mathcal{P}_k(\hat{A})\right\rangle \right)}.
    \end{align}
    Using the orthonormality of eigenspaces, Eq.~\eqref{eq:eigenspaceorthonormality}, and the properties of the inner product, this inequality can be simplified to
    \begin{equation}
        \left\lvert \sum_{k \in \mathcal{K}} \widetilde{w}(\mathcal{E}_k) \langle \hat{B}, \mathcal{P}_k(\hat{A})\rangle \right\rvert \leq  \sqrt{\left(\sum_{k \in \mathcal{K}} \left\lvert \widetilde{w}(\mathcal{E}_k)\right\rvert^2 \langle \hat{A}, \mathcal{P}_k(\hat{A})\rangle \right)}\ \sqrt{\langle \hat{B}, \hat{B}\rangle}.
    \end{equation}
     Applying Eq.~\eqref{eq:w_w0_constraint} to the terms in $\mathcal{K}_{\Delta E}^c$ and noting that $\widetilde{w}(E) \leq 1$ due to $w(t) \geq 0$ (see e.g. Refs.~\cite{FTpositivityConvex, FTpositivity2014}) (applied to the terms in $\mathcal{K}_{\delta E}$), we get:
     \begin{equation}
          \left\lvert \int\diff t\ w(t) \langle \hat{B}, \mathcal{U}_t(\hat{A})\rangle\right\rvert \leq \left(\sqrt{\left(\sum_{k \in \mathcal{K}_{\Delta E}} \langle \hat{A}, \mathcal{P}_k(\hat{A})\rangle \right)} + w_0 \sqrt{\left(\sum_{k \in \mathcal{K}_{\Delta E}^c} \langle \hat{A}, \mathcal{P}_k(\hat{A})\rangle \right)} \right) \sqrt{\langle \hat{B}, \hat{B}\rangle}.
          \label{eq:correlator_to_eigenspaces_bound}
     \end{equation}

     Momentarily setting aside the above inequality, let us now turn to the eigendecomposition of the completely positive weighted average of the autocorrelator, as per Eq.~\eqref{eq:genunitary_eigendecomposition}:
     \begin{equation}
         \int\diff t\ w_+(t) \langle \hat{A}, \mathcal{U}_t(\hat{A})\rangle = \sum_k \widetilde{w}_+(\mathcal{E}_k) \langle \hat{A}, \mathcal{P}_k(\hat{A})\rangle.
     \end{equation}
     As $\widetilde{w}_+(\Delta E) \geq 0$ by assumption, and further $\langle \hat{A}, \mathcal{P}_k(\hat{A})\rangle = \langle \mathcal{P}_k(\hat{A}), \mathcal{P}_k(\hat{A})\rangle \geq 0$ by Eq.~\eqref{eq:eigenspaceorthonormality} and Eq.~\eqref{eq:ip_postivity}, each term on the right hand side is non-negative. It follows that
     \begin{equation}
         \int\diff t\ w_+(t) \langle \hat{A}, \mathcal{U}_t(\hat{A})\rangle \geq \sum_{k \in \mathcal{K}_{\Delta E}} \widetilde{w}_+(\mathcal{E}_k) \langle \hat{A}, \mathcal{P}_k(\hat{A})\rangle \geq 0.
         \label{eq:autocorrelator_Fourier_constraint}
     \end{equation}
     Now, using Eq.~\eqref{eq:wplus_W_constraint}, we have
     \begin{equation}
         \sum_{k \in \mathcal{K}_{\Delta E}} \widetilde{w}_+(\mathcal{E}_k) \langle \hat{A}, \mathcal{P}_k(\hat{A})\rangle \geq W  \sum_{k \in \mathcal{K}_{\Delta E}} \langle \hat{A}, \mathcal{P}_k(\hat{A})\rangle,
     \end{equation}
     due to which by Eq.~\eqref{eq:autocorrelator_Fourier_constraint},
     \begin{equation}
         \sum_{k \in \mathcal{K}_{\Delta E}} \langle \hat{A}, \mathcal{P}_k(\hat{A})\rangle \leq \frac{1}{W}\int\diff t\ w_+(t) \langle \hat{A}, \mathcal{U}_t(\hat{A})\rangle.
         \label{eq:autocorrelator_energyband_constraint}
     \end{equation}

     Separately, again due to $\langle \hat{A}, \mathcal{P}_k(\hat{A})\rangle = \langle \mathcal{P}_k(\hat{A}), \mathcal{P}_k(\hat{A})\rangle \geq 0$, we have
     \begin{equation}
         0 \leq  \sum_{k \in \mathcal{K}_{\Delta E}^c} \langle \hat{A}, \mathcal{P}_k(\hat{A})\rangle \leq \sum_k \langle \hat{A},  \mathcal{P}_k(\hat{A})\rangle \leq \langle \hat{A}, \hat{A}\rangle,
         \label{eq:projected_Ruto_norms_constraint}
     \end{equation}
     where the last inequality follows from the completeness relation, Eq.~\eqref{eq:eigenspacecompleteness}.

     Using Eq.~\eqref{eq:autocorrelator_energyband_constraint} and \eqref{eq:projected_Ruto_norms_constraint} in Eq.~\eqref{eq:correlator_to_eigenspaces_bound}, we get Eq.~\eqref{eq:auto_to_cross_correlator} of Theorem~\ref{thm:autocorrelatorimpliescorrelator_Gen}.

\printbibliography

\end{document}